\documentclass[11pt,letterpaper]{amsart}
\usepackage[margin=1in]{geometry}

\usepackage[utf8]{inputenc}
\usepackage[T1]{fontenc}

\usepackage{amsmath}
\usepackage{amssymb}
\usepackage{amsthm}
\usepackage{mathtools}
\usepackage{fixmath}
\usepackage{interval}
\usepackage{algpseudocode}
\usepackage[normalem]{ulem}
\usepackage{xspace}
\usepackage{tikz}
\usetikzlibrary{arrows}

\usepackage{hyperref}
\usepackage[capitalize,nameinlink]{cleveref}
\hypersetup{pdfborder={0 0 0}, hypertexnames=false, breaklinks}

\numberwithin{equation}{section}
\numberwithin{figure}{section}

\theoremstyle{plain}
\newtheorem{theorem}{Theorem}[section]
\newtheorem{lemma}[theorem]{Lemma}
\newtheorem{corollary}[theorem]{Corollary}
\newtheorem{proposition}[theorem]{Proposition}
\theoremstyle{definition}
\newtheorem{definition}[theorem]{Definition}
\newtheorem{example}[theorem]{Example}
\newtheorem{assumption}{Assumption}

\theoremstyle{remark}
\newtheorem{remark}[theorem]{Remark}

\crefname{line}{line}{lines}
\Crefname{line}{Line}{Lines}
\crefname{section}{section}{sections}
\crefname{subsection}{subsection}{subsections}
\Crefname{section}{Section}{Sections}
\Crefname{subsection}{Subsection}{Subsections}
\Crefname{figure}{Figure}{Figures}
\crefname{assumption}{Assumption}{Assumptions}
\Crefname{assumption}{Assumption}{Assumptions}

\crefformat{equation}{\textup{#2(#1)#3}}
\crefrangeformat{equation}{\textup{#3(#1)#4--#5(#2)#6}}
\crefmultiformat{equation}{\textup{#2(#1)#3}}{ and \textup{#2(#1)#3}}
{, \textup{#2(#1)#3}}{, and \textup{#2(#1)#3}}
\crefrangemultiformat{equation}{\textup{#3(#1)#4--#5(#2)#6}}%
{ and \textup{#3(#1)#4--#5(#2)#6}}{, \textup{#3(#1)#4--#5(#2)#6}}{, and \textup{#3(#1)#4--#5(#2)#6}}
\Crefformat{equation}{#2Equation~\textup{(#1)}#3}
\Crefrangeformat{equation}{Equations~\textup{#3(#1)#4--#5(#2)#6}}
\Crefmultiformat{equation}{Equations~\textup{#2(#1)#3}}{ and \textup{#2(#1)#3}}
{, \textup{#2(#1)#3}}{, and \textup{#2(#1)#3}}
\Crefrangemultiformat{equation}{Equations~\textup{#3(#1)#4--#5(#2)#6}}%
{ and \textup{#3(#1)#4--#5(#2)#6}}{, \textup{#3(#1)#4--#5(#2)#6}}{, and \textup{#3(#1)#4--#5(#2)#6}}

\renewcommand{\geq}{\geqslant}
\renewcommand{\leq}{\leqslant}
\renewcommand{\le}{\leq}
\renewcommand{\ge}{\geq}

\newcommand{\R}{\mathbb{R}}
\newcommand{\abs}[1]{|{#1}|}
\DeclareMathOperator*{\val}{\mathsf{val}}
\newcommand{\card}[1]{|{#1}|}
\newcommand{\polyh}{\mathcal{P}}
\DeclareMathOperator*{\argmax}{arg\,max}
\DeclareMathOperator*{\argmin}{arg\,min}
\newcommand{\dgraph}{\vec{\mathcal{G}}}
\newcommand{\vertices}{V}
\newcommand{\edges}{E}
\newcommand{\Max}{\mathrm{Max}}
\newcommand{\Min}{\mathrm{Min}}
\newcommand{\Maxvertices}{\vertices_{\Max}}
\newcommand{\Minvertices}{\vertices_{\Min}}
\newcommand{\dunion}{\uplus}
\newcommand{\payoff}{r}
\newcommand{\shapley}{F}
\newcommand{\ones}{\mathbf{1}}
\newcommand{\gameval}{\lambda}
\newcommand{\Prob}{\mathbb{P}}
\newcommand{\hilb}[1]{\| #1 \|_{\mathrm{H}}}
\newcommand{\supnorm}[1]{\| #1 \|_{\infty}}
\newcommand{\bias}{u}
\newcommand{\poly}{\mathrm{poly}}
\newcommand{\Generic}{\mathcal{U}}
\newcommand{\dens}{f}
\newcommand{\nspace}{\mathcal{X}}
\newcommand{\Cond}{\Delta}
\newcommand{\disc}{\gamma}
\newcommand{\gauss}{\mu}
\newcommand{\sk}{\mathrm{sl}}
\newcommand{\DiscPI}{\textsc{DiscountedPI}}
\newcommand{\IncDiscPI}{\textsc{IncreasingDiscountPI}}
\newcommand{\IncDiscPItwo}{\textsc{IncreasingDiscountPI2}}
\newcommand{\ZPHVI}{\textsc{ZPHValueIteration}}
\newcommand{\eps}{\varepsilon}
\newcommand{\KM}{Krasnoselskii--Mann\xspace}
\newcommand{\arxiv}[1]{\href{http://arxiv.org/abs/#1}{arXiv:#1}}

\title[Smoothed analysis of policy iteration and value iteration]{Smoothed analysis of policy iteration and value iteration for deterministic discounted and mean-payoff games played on any graph}

\author{Bruno Loff}
\address{LASIGE, Faculdade de Ciências, Universidade de Lisboa}
\email{bruno.loff@gmail.com}

\author{Mateusz Skomra}
\address{LAAS-CNRS, Universit{\'e} de Toulouse, CNRS, Toulouse, France}
\email{mateusz.skomra@laas.fr}

\begin{document}

\begin{abstract}
  Our understanding of turn-based stochastic games is reminiscent of our understanding of linear programming in the 1970s. Instead of the simplex method, we have two families of algorithms --- policy iteration and value iteration --- which work well in practice, but which are known to run in exponential time on some input instances. And we are still lacking an algorithm that runs in polynomial time on all instances.
  
  A recent result by Loff and Skomra showed that, for deterministic discounted and mean-payoff games played on an ergodic graph, a certain variant of the policy iteration method runs in smoothed polynomial time. This offered an explanation of why policy iteration performs well in practice, analogous to Spielman and Teng's explanation for the simplex method.
  
  In this paper, we improve their analysis in order to drop the ergodicity assumption: we show that their policy iteration algorithms actually run in smoothed polynomial time for \emph{any} game graph, not only for ergodic graphs. We do this by defining a condition number, a natural generalization of \cite{LoffSkomra:2024}, showing that sufficiently-perturbed games are well-conditioned with high probability, and showing that their algorithms are efficient on well-conditioned instances. Although this is a similar outline to \cite{LoffSkomra:2024}, our proof requires a substantially new approach, based on log-Lipschitz distributions which made an appearance in recent works on the smoothed analysis of the simplex method.
  
  For our second contribution, we propose three new value iteration algorithms for solving deterministic mean-payoff games (also on any graph, not necessarily ergodic). All three algorithms are based on iterative methods for finding fixed points of non-expansive operators, which we show can be applied to the Shapley operator. The first two algorithms are simple applications of Krasnoselskii--Mann iteration and Halpern iteration. The third algorithm is quadratically faster but more complicated, it is a two-stage algorithm which runs Zwick--Paterson and Halpern iterations in parallel. We prove the correctness of these algorithms and show that all of them run in smoothed polynomial time.
\end{abstract}

\maketitle

\bigskip\noindent\textbf{Acknowledgements.} ChatGPT was used to do a proof-check pass. The authors assume full responsibility for all content.

Funded by the European Union (ERC, HOFGA, 101041696). Views and opinions expressed are however those of the author(s) only and do not necessarily reflect those of the European Union or the European Research Council. Neither the European Union nor the granting authority can be held responsible for them. Also supported by FCT through the LASIGE Research Unit, ref. UID/00408/2025. This work received funding from the French National Research Agency (ANR), grant ANR-26-CE48-4720 (COSPA).

\newpage
\setcounter{tocdepth}{2}
\tableofcontents

\newpage

\section{Introduction}

  \subsection{Discounted and mean-payoff games}

  A deterministic turn-based two-player zero-sum infinite game is played in an edge-weighted directed graph $G = ([n], E, r)$, in which the vertices are split into two disjoint sets, $[n] = \Maxvertices \dunion \Minvertices$. At each round $t$, the game is in some state $i(t)$, and one of the players (player Max if $i \in \Maxvertices$ or player Min if $i \in \Minvertices$) chooses an edge $(i(t),i(t+1)) \in E$, which determines the state at the next round $i(t+1)$ and a resulting payoff $r_{i(t),i(t+1)}$. This is repeated for an infinite number of turns. One player is trying to maximize the \emph{total payoff at the limit}, and the other tries to minimize it.

  There are two natural ways of determining the \emph{total payoff at the limit}. In the \emph{discounted} variant, payoffs received at round $t$ are multiplied by a \emph{discount factor} of $\disc^t$ (for some $0 \le \disc < 1$), and the limit payoff is: $(1 - \disc) \sum_{t=0}^\infty \disc^t \cdot r_{i(t), i(t+1)}$. This is equivalent to saying that the game is forced to stop after every round with probability $1-\disc$, and asking for the expected payoff at the limit. In the \emph{mean-payoff} variant, we measure the liminf or limsup (it turns out this choice will not matter), as the number of rounds goes to infinity, of the average payoff received so far (i.e. total payoff divided by the number of rounds): $\liminf_{t\to\infty} \frac{1}{t} \sum_{s=1}^t r_{i(s), i(s+1)}$.

  Many variants of this basic scenario have been considered, e.g. in the original work by Lloyd Shapley \cite{shapley_stochastic}, he considered the most general case where both players could play simultaneously (so, \emph{concurrent} instead of \emph{turn-based} games), and where the next state depends probabilistically, instead of deterministically, on the current state and the players' chosen actions (\emph{stochastic} instead of \emph{deterministic} games). Shapley~\cite{shapley_stochastic} originally proved the existence of a value and optimal (mixed) strategies for the discounted variant.

  A desired property of such games is the existence of  \emph{memoryless} optimal strategies, where the choice of what to do only depends on the current state $i$, and not on the past history. However, for mean-payoff games, Gillette~\cite{gillette} gave an example of a mean-payoff game, where the players play simultaneously at each round, whose optimal strategies cannot be memoryless.\footnote{In fact, it was only in the 1980s \cite{mertens_neyman} that simultaneous-move mean-payoff games were proven to have a value. For every $\eps > 0$, $\eps$-optimal (not memoryless) strategies exist for each player, i.e., ensuring that the payoff is $\eps$-close to the value.} This motivated Gillette to introduce the \emph{turn-based} variant that we described above. Gillette claimed that turn-based two-player \emph{mean-payoff} games have a value and that optimal memoryless strategies do exist for both players. Gillette's proof was actually wrong, but it was later corrected by Liggett and Lippman~\cite{liggett_lippman}, so the statement \emph{is} true.

  Since then, many methods have been proposed for solving stochastic games, i.e. computing the value and optimal strategies of a given game. For a textbook reference, see, e.g. \cite{games_on_graphs}. Alternatively, see \cite{laraki_survey} for a generalized treatment of zero-sum games that are repeated for an infinite number of turns.

  \subsection{Two algorithmic paradigms: value and policy iteration}

  Our present understanding of stochastic games is reminiscent of our understanding of linear programming in the 1970s, before Leonid Khachiyan discovered the ellipsoid method in 1979: the problem of solving turn-based stochastic games is in $\mathsf{NP} \cap \mathsf{coNP}$, but it remains unknown whether there exists a polynomial-time algorithm. Furthermore, in full analogy with the simplex method for linear programming, the known algorithms for stochastic games work very well in practice\footnote{For stochastic games, see \cite{kretinsky_2022_comparison,meggendorfer_2024_playing}. For one-player games, i.e. Markov Decision Processes, see e.g. \cite{hartmanns_2026_revised}. For deterministic games, see, e.g., \cite{chaloupka2009parallel,cadilhac_2025_fast}. }, but for the most well-known algorithms we have constructed instances that force them to run for exponential time.\footnote{E.g. for value iteration see \cite{BeffaraVorobyov:2001} (for the GKK algorithm), \cite{zwick_paterson} (for their own method), \cite{Dorfman:2024} (for the algorithm in \cite{DorfmanKaplanZwick:2019}), and for policy iteration see \cite{friedmann2011exponential,hansen_subexponential_paths} (for the standard Howard and random facet rules), \cite{VanDijkLohoMaat:2024} (for symmetric policy iteration), and \cite{DisserLohoMaatNils:2026} (for ranking-based rules).}

  Most algorithms for solving two-player discounted and mean-payoff games can be broadly classified into two families: value-iteration algorithms and policy iteration algorithms.

  Indeed, the first algorithm ever invented for stochastic games was a value iteration algorithm \cite{bellman1957dynamic}, an algorithm for solving discounted Markov Decision Processes (MDPs). MDPs are a simpler, single-player variant of stochastic games.\footnote{For a modern value-iteration algorithm for MDPs, see \cite{sidford2018variance}, which contains a historical overview in Section 2.} Bellman's algorithm works by considering an operator which computes the value vector $\bar\gameval^{(t+1)} \in \R^n$ of a $t+1$-stage truncated discounted game, when given the value vector $\bar\gameval^{(t)}$ of a $t$-stage truncated discounted game: this operator is a contraction under the $\ell_\infty$ norm, and hence, iterating this operator will converge towards a fixed point, which is the solution of the untruncated game.  Shapley's \cite{shapley_stochastic} proof (of the existence of values and optimal strategies for concurrent two-player discounted games) was based on the same idea: it showed that the value of a discounted game can be expressed as a fixed point of a certain operator, now called the (discounted) Shapley operator. As we may see, here we are literally iterating a value estimate, hence the name \emph{value iteration}.

  For mean-payoff games, the (mean-payoff) Shapley operator is no longer a contraction, but one may still iterate the undiscounted Shapley operator $t$ times to compute the total-payoff values for a  $t$-stage version of the game, and then divide by $t$. It turns out that the $t$-th iterate divided by $t$ is $O(n\frac{\|r\|_\infty}{t})$-close to the true value vector $\gameval$, and so, if the weights $r$ are integer and $t$ is large enough, say $t = O(n^3 \|\payoff\|_{\infty})$, we can deduce the exact value of the game from this approximation.
  This is the Zwick--Paterson algorithm \cite{zwick_paterson}.

  The first algorithmic approach for turn-based two-player mean-payoff games was due to Gurvich, Karzanov, and Khachiyan \cite{gurvich}. Their method uses potential transformations $r_{ij}\mapsto r_{ij}+u_j-u_i$, which preserve the total weight of every cycle and hence preserve all mean-payoff values. The algorithm iterates such transformations until the transformed game reveals which vertices have nonnegative value. This potential-reduction viewpoint has recently been used as a unifying framework to explain several other algorithms \cite{cadilhac_2025_fast}.\footnote{It should be said that these algorithms are iterating over a potential function, and not just over a value vector, but ``value iteration algorithm'' refers more generally to algorithms that iterate over real-valued data, such as values and potentials. In contrast, ``policy iteration algorithms'' iterate over policies, which are discrete data.}

  Recently, a different approach has been used to solve mean-payoff MDPs \cite{LeeRyu:2025,LeeBravoCominetti:2025,ZurekChen:2025}, and ``irreducible'' stochastic games \cite{AkianGaubertNaepelsTerver:2023}\footnote{Irreducibility is a very strong criterion which is not appropriate to study deterministic games, the condition is something like ``for every choice of policies the resulting Markov chain is irreducible''.}, using Halpern iteration and \KM iteration. One of our contributions is to show that this approach can be used to solve deterministic mean-payoff games, which, as far as we know, had not been observed before. We will also show that the resulting algorithms run in smoothed polynomial time (see below for more details).

\medskip
Now for policy iteration. Policy iteration algorithms maintain a pair of policies in memory, and proceed by repeatedly modifying the policies, so that their quality improves monotonically according to some measure, until it can no longer be improved, at which point the measure must guarantee that we have found an optimal strategy for both players.

The first policy iteration algorithm was invented by Howard \cite{howard1960dynamic}, and finds an optimal strategy for deterministic (and some stochastic) Markov Decision Processes. This was later extended by Denardo and Fox \cite{denardo1968multichain} to work on all stochastic MDPs. The method was first extended to irreducible two-player mean-payoff games by Hoffman and Karp \cite{hoffman_karp}, and two-player discounted games by Denardo \cite{denardo1967contraction}, with later developments by many other authors \cite{rao1973algorithms,puri1995theory,CochetGaubertGunawardena:1999,gaubert1998duality,raghavan2003policy,cochetterrasson2006policy}. A good historical overview with more technical details appears in \cite{detournay_policy}, where a policy iteration algorithm first appeared that can handle all the variants of mean-payoff games. In the case of discounted games, an optimal strategy can be found in time polynomial in $\frac{1}{1-\gamma}$ \cite{ye2011simplex,hansen2013strategy}. Otherwise, for mean-payoff games, or for discount factors $\gamma$ exponentially close to $1$, no upper-bound is known on the number of iterations, significantly better than the number of policies, which is exponential in the number $n$ of states. More precisely, the best upper-bound on the number of iterations is $2^{\tilde O(\sqrt n)}$ \cite{halman,hansen_zwick_random_facet}.

\subsection{Smoothed analysis and the previous result for ergodic graphs}

In the case of the simplex method, the discrepancy between practical performance and the existence of exponential-time hard instances was finally explained by Spielman and Teng in 2004 \cite{spielman2004smoothed}. They have shown that, if one takes any linear programming instance $\max \{ c\cdot x \mid A \cdot x \ge b\}$, of dimensions $m\times n$ and bounded norm $\|[A \quad b]\|_{2,\infty} \le 1$, and perturbs each entry of $A$ and $b$, by a Gaussian with mean $0$ and standard deviation $\rho$, then the simplex method will solve the resulting perturbed system in time $\poly(m,n,1/\rho)$
\cite{spielman2004smoothed,dadush2018friendly}. It is then reasonable to expect the simplex method to work efficiently on real-world instances, since they incorporate real-world data which is prone to such perturbations.\footnote{See also \cite{huiberts2023upper,BachHuiberts:2025,bach2026beyond}.} This result of Spielman and Teng founded the area of \emph{smoothed analysis}, where one studies the efficiency of algorithms on such perturbed inputs. It is then said that \emph{the simplex method runs in smoothed polynomial time}.

  In an earlier paper \cite{LoffSkomra:2024}, Loff and Skomra devised algorithms for solving deterministic discounted and mean-payoff games in smoothed polynomial time. (Of course, one key difference with the simplex method is that we \emph{do know} how to solve linear programming in polynomial time, and the only question is \emph{why is the simplex method efficient in practice}; by contrast, no polytime algorithm is known for discounted and mean-payoff games.) The main theorem from that paper is:

  \begin{theorem}[main theorem from \cite{LoffSkomra:2024}]\label{thm:main-old}
    There exists a policy-iteration algorithm for solving $n$-state deterministic two-player (discounted or mean-payoff) games played on ergodic graphs, which runs in time $\poly(n,1/\rho)$ with high probability, on an input where normalized payoffs in $[-1,1]$ have been independently perturbed by a Gaussian with mean $0$ and standard deviation $\rho$.
  \end{theorem}

  Hence, for their earlier algorithm to work, the game has to be ``played on an ergodic graph''. Let us explain what that means. Let $G = ([n] = \Maxvertices \dunion \Minvertices, E)$ be a directed graph. A \emph{trap} by \emph{player Min} against \emph{player Max} is a strict subset $S \subsetneq [n]$ of the vertices, a starting vertex $s \in S$, and a choice of outgoing edges for each vertex in $S \cap \Minvertices$ such that, no matter how one chooses outgoing edges for each vertex in $S \cap \Maxvertices$, the resulting path (starting at $s$ and following the chosen edges) never leaves $S$. In other words, player Min can ``trap'' player Max in $S$. Likewise we define a trap by player Max against player Min. The graph $G$ is called \emph{ergodic} if there do not exist two disjoint traps, one for each player. The obvious example is a complete bipartite graph between $\Maxvertices$ and $\Minvertices$.

\subsection{Our main two results}

  In this paper we have two main results. We first show that the previous algorithm of Loff and Skomra actually works on any game, not just ergodic:

   \begin{theorem}[first main theorem]\label{thm:main-first}
    The policy-iteration algorithm of \cite{LoffSkomra:2024} also works for non-ergodic graphs. In other words, given \emph{any} $n$-state deterministic two-player (discounted or mean-payoff) game, played on \emph{any graph}, it will run in time $\poly(n,1/\rho)$ with high probability, on an input where normalized payoffs in $[-1,1]$ have been independently perturbed by a Gaussian with mean $0$ and standard deviation $\rho$.
  \end{theorem}

  Our second main result proves that \emph{value iteration} algorithms can also be made to run in smoothed polynomial time:

  \begin{theorem}[second main theorem]\label{thm:main-second}
    There exists a value iteration algorithm which, given \emph{any} $n$-state deterministic two-player (discounted or mean-payoff) game, played on \emph{any graph}, will run in time $\poly(n,1/\rho)$ with high probability, on an input where normalized payoffs in $[-1,1]$ have been independently perturbed by a Gaussian with mean $0$ and standard deviation $\rho$.
  \end{theorem}

  This second result should be compared with the results of the paper \cite{AllamigeonGaubertKatzSkomra:2022} (journal version: \cite{AllamigeonGaubertKatzSkomra:2025}). This paper defines a certain condition number for checking if a constant-value game has positive value, and shows that the Zwick--Paterson algorithm runs efficiently in this condition number. The authors make no probabilistic analysis of their condition number, i.e., they do not show that a randomly perturbed game is likely to be well-conditioned, even if the underlying graph is ergodic. Our own probabilistic analysis can be easily adapted to show that their condition number is indeed small for random games, so Zwick--Paterson iteration can indeed quickly detect the sign of the value of a random game played on an ergodic graph. Their paper also gives an algorithm which, for a constant-value game, uses Zwick--Paterson iterates and ``aggregates'' them produce $\varepsilon$-optimal strategies. However, it is no longer clear whether our probabilistic analysis can be made to apply to this algorithm. For this reason, we use the \KM and Halpern iterations instead of the Zwick--Paterson iteration. We show that, if a game played on an ergodic graph is well-conditioned in the sense of \cite{LoffSkomra:2024}, then the $\varepsilon$-optimal strategies produced by these value iteration algorithms quickly become optimal. No such property has been observed in \cite{AllamigeonGaubertKatzSkomra:2025}. Finally, their paper also extends this more complicated value iteration algorithm to work for non-ergodic games. To do this, they introduce a second condition number \cite[Theorem 11]{AllamigeonGaubertKatzSkomra:2025}, again without a probabilistic analysis. Their condition number can be applied to more general games such as stochastic mean payoff games and entropy games. For deterministic games, this second condition number can be simplified and we provide a probabilistic analysis of this simplified number. More precisely, we introduce a single condition number, which is based on \cite{LoffSkomra:2024}, and can be used to study both ergodic and non-ergodic games.

  \subsection{Our results versus \cite{LoffSkomra:2024}}

  Both our main theorems closely follow the approach of Loff and Skomra \cite{LoffSkomra:2024}. Our paper has five new technical contributions, the first three of which have counterparts in the Loff--Skomra paper, so let us compare our results with theirs. The first contribution of \cite{LoffSkomra:2024} was defining a certain condition number $\Delta(r)$, namely, this is a number depending on the edge weights $r$ of an \emph{ergodic} game, which is supposed to measure how hard it is to solve the game with edge weights $r$. Their second contribution is showing that perturbed inputs are well conditioned, meaning, they have small condition number with high probability. Their third contribution is showing that a certain policy iteration algorithm runs efficiently on well-conditioned inputs. We may summarize our contributions conceptually as follows (more details follow below):

  \begin{description}
  \setlength\itemsep{0.5em}
    \item[Contribution 1] Firstly, we define a condition number that works for games played on general, not necessarily ergodic graphs. Our condition number is strongly based on Loff and Skomra's: it is a natural generalization of their proposal to the non-ergodic case. Their condition number is based on the ergodic equation, and ours is based on the GKK equations.

    \item[Contribution 2] We then show, as they do, that perturbed inputs are well-conditioned with high probability. Our proof, however, is \emph{completely different} to the proof of \cite{LoffSkomra:2024}. As it turns out, for non-ergodic graphs the previous analysis fails at a crucial juncture, for a meaningful reason that is owed specifically to non-ergodicity. We overcome this obstacle by making use of log-Lipschitz distributions, inspired by recent work on the smoothed analysis of the simplex method.

    \item[Contribution 3] We then show that policy iteration runs efficiently on instances that are well-conditioned under our new notion. As in \cite{LoffSkomra:2024}, the proof goes through discounted games: we show that for every well-conditioned mean-payoff game, if the discount factor $\gamma$ is chosen inverse-polynomially-close to $1$, then the discounted game and the mean-payoff game have the same unique optimal policies (for general games this requires $\gamma$ being exponentially close to $1$). Known policy iteration algorithms solve discounted games in time polynomial in $n$ and $\frac{1}{1-\gamma}$, and therefore yield a polynomial-time algorithm on well-conditioned mean-payoff games. The main work is to reprove the Loff--Skomra discounted-to-mean-payoff transfer theorem for our new condition number.

    \item[Contribution 4] Loff and Skomra did not study value iteration. Our fourth contribution is to introduce three value-iteration algorithms for deterministic mean-payoff games. The first two are based on Halpern iteration and \KM iteration applied to the Shapley operator in the Hilbert seminorm. We ultimately show that these iterative methods are able to recover the value structure of the game, partitioning the game graph into constant-value subgames, and then solve each subgame independently.  The third algorithm accelerates the first part of this process: it uses Zwick--Paterson iteration to partition the game into constant-value subgames, while simultaneously using Halpern iteration to solve each subgame.

    \item[Contribution 5] Our fifth contribution is to prove that all three value-iteration algorithms run efficiently on well-conditioned inputs, and hence in smoothed polynomial time. The proof is different from the policy-iteration argument, and does not consider discounted games at all. The analysis proceeds in two steps. First, we show that, in a well-conditioned constant-value game, a sufficiently accurate approximate bias already induces the exact optimal policies (cf.~\Cref{def:bias,def:induced_policy}). Second, for general games, we show that value iteration eventually ignores all edges between vertices of different values, reducing the problem to constant-value subgames. This gives polynomial bounds for plain \KM and Halpern iteration, and a better bound for the combined Zwick--Paterson/Halpern algorithm.

  \end{description}
  Let us explain the Loff--Skomra paper, and our own contributions, in greater technical detail.

  \subsection{Overview of the Loff--Skomra paper}

  Their approach begins with the observation that the space $\R^E$ of possible edge-weights can be nearly partitioned by a collection of open polyhedra $\R^E \approx \mathcal U = \bigcup_{\sigma,\tau} \mathcal P^{\sigma,\tau}$, indexed by pairs of policies $(\sigma, \tau)$. By nearly partitioned we mean that $R^E \setminus \mathcal U$ has measure $0$. Each set $\mathcal P^{\sigma,\tau}$ is an open polyhedral cone, and is precisely the set of all weight vectors $r \in \R^E$ such that $(\sigma, \tau)$ is the unique \emph{bias-induced} pair of optimal policies. Let us explain this.

  It can be shown (see \cite{shapley_stochastic} for discounted games, \cite{liggett_lippman,ehrenfeucht_mycielski} for mean-payoff games) that deterministic turn-based two-player games (for both discounted or mean-payoff variants) have optimal memoryless strategies, called \emph{policies}, for both players. I.e., there exists a choice $\sigma$ of outgoing edges of each vertex in $\Maxvertices$, and a choice $\tau$ of outgoing edges for each vertex in $\Minvertices$, such that $\sigma$ is an optimal strategy for player Max and $\tau$ is an optimal strategy for player Min. In general a game can have multiple optimal strategies, but some such strategies are special. In the case of mean-payoff games played on an ergodic graph, one can define the following system of so-called \emph{ergodic} equations:
  \begin{align*}
    \forall i \in \Maxvertices, \quad &\gameval + \bias_i = \max_{(i,j) \in \edges} \{\payoff_{ij} + \bias_{j}\}, \\
    \forall i \in \Minvertices, \quad &\gameval + \bias_i = \min_{(i,j) \in \edges} \{\payoff_{ij} + \bias_{j}\}.
  \end{align*}
  It follows that any solution $(\gameval, \bar \bias) \in \R^{1+n}$ of this system gives us the value of the game $\gameval$ \cite{gurvich,generic_uniqueness}. The vector $\bar \bias \in \R^n$ is called a \emph{bias} of the game. Also, for a given solution $(\gameval, \bar\bias)$, the argmax/argmin of each maximum/minimum give us a pair $(\sigma,\tau)$ of optimal policies, and these policies (that arise from solutions of the ergodic equation) are called \emph{bias-induced policies}. Bias-induced policies are not unique in general, but \emph{for a generic set of payoffs $r \in \R^E$}, i.e. for all payoffs in $\R^E$ except for a set of measure $0$, the half-line $\R\cdot (1, \ldots, 1) + \bar u$ is unique (i.e. all solutions $\bar u$ differ by a constant vector) and there is a unique pair of bias-induced policies (cf. \cite{LoffSkomra:2024,generic_uniqueness}).

  It is not obvious, but it can be shown that, for any fixed $(\sigma,\tau)$, the predicate ``$(\sigma,\tau)$ is a bias-induced policy for $r$'' can be expressed by a system of linear inequalities in the edge weights $r \in \R^E$.
  If we treat $(\gameval,\bar\bias)$ as constants, this is obvious from the ergodic equations: to express ``the edge $(i,j)$ is \emph{not} in the optimal policy'', where $i \in \Maxvertices$, we can simply write
  \[
    \gameval + \bias_i > r_{ij} + \bias_j,
  \]
  which for a solution $\gameval, \bias$ is to say
  \[
    r_{i j} < \gameval + \bias_i - \bias_j.
  \]
  If we put all of these together for the various edges $(i,j)$ not in the policies $(\sigma, \tau$), we get a linear system of inequalities.
  The non-obvious part is how to express this as linear inequalities in the $r_{i j}$ variables alone, i.e. without mentioning $\bias$ or $\gameval$. It turns out that, for fixed policies $(\sigma,\tau)$, within $\mathcal P^{\sigma,\tau}$ the values of $u_i$ and $\lambda$ are linear functions of the weights $r$, so inequalities such as the above are indeed linear in the weights. The resulting system of linear inequalities defines precisely the polyhedron $\mathcal P^{\sigma,\tau}$.

  The Loff--Skomra paper then defines (their first contribution) a condition number $\Cond(r)$, which measures the  (suitably normalized) inverse-distance between the payoffs $r$ and the boundary of the cone $\mathcal P^{\sigma,\tau}$ containing $r$ (this is $\infty$ if $r$ is at the boundary). They show (their second contribution) that if we perturb any vector $r \in [-1,1]^E$, for example, $\tilde r = r + \mathcal N(0,\rho^2)^E$, so each weight is perturbed by a normal distribution with variance $\rho^2$, then the resulting perturbed weights $\tilde r$ have $\Cond(\tilde r) \le \poly(n,1/\rho)$ with high probability (over the choice of perturbation). They finally show (their third contribution) that a certain policy iteration algorithm runs in time $\poly(n, \Cond)$.

  Two natural follow-up questions arise. First, can we also solve the non-ergodic case? And second, value iteration has not been studied at all: can value-iteration algorithms also be made to run in smoothed polynomial time?

  \subsection{Contribution 1. A condition number for non-ergodic graphs}\label{sec:contribution-1}

  Starting with the first question: what if the game is not ergodic? The first technical issue is that the Loff--Skomra condition number is defined using the ergodic equation, and therefore only applies to ergodic graphs. However, a classical result by Gurvich, Karzanov, and Khachiyan provides a more complicated system of equations analogous to the ergodic equations. In the non-ergodic case, the value of the game depends on the node we start at (because the graph is not well-connected), and so the analogous system has $2n$ variables $(\gameval, \bias) \in \R^n \times \R^n$. The GKK equations are:
  \begin{align*}
    \forall i \in \Maxvertices, \quad &\gameval_{i} = \max_{(i,j) \in \edges} \gameval_{j}, \\
    \forall i \in \Minvertices, \quad &\gameval_{i} = \min_{(i,j) \in \edges} \gameval_{j}, \\
    \forall i \in \Maxvertices, \quad &\gameval_{i} + \bias_i = \max_{(i,j) \in \edges, \gameval_j = \gameval_i} \{\payoff_{ij} + \bias_{j}\} > \max_{(i,j) \in \edges, \gameval_j \neq \gameval_i} \{\payoff_{ij} + \bias_{j}\}, \\
    \forall i \in \Minvertices, \quad &\gameval_{i} + \bias_i = \min_{(i,j) \in \edges, \gameval_j = \gameval_i} \{\payoff_{ij} + \bias_{j}\} < \min_{(i,j) \in \edges, \gameval_j \neq \gameval_i} \{\payoff_{ij} + \bias_{j}\}.
  \end{align*}
  Intuitively, the first two equations partition the nodes of the graph into clusters, where the value $\lambda_i$ is the same for all vertices $i$ within the same cluster, but may differ across clusters. Then the two remaining equations select the best bias-adjusted payoff $\payoff_{ij} + \bias_j - \bias_i$ within each cluster, and state that no edges from a max/min node can reach a higher/lower-value cluster. Indeed, if a pair of policies $(\sigma,\tau)$ selects the argmax/argmin edge at each node, then $(\sigma,\tau)$ are optimal, and policies obtained in this way are also called \emph{bias-induced policies}.

  Once again, it can be shown that the predicate ``$(\sigma,\tau)$ is a bias-induced policy for $r$'' can be expressed by a system of linear equations in the weights $r_{i j}$. And so the space $\R^E$ of edge-weights can once again be partitioned by polyhedral cones, and one can define a condition number as before, by measuring inverse-distance to the boundary. This straightforward adaptation of \cite{LoffSkomra:2024} is the first contribution of our paper.

  \subsection{Contribution 2. Showing that perturbed inputs are well-conditioned}
  We now wish to show that a randomly-perturbed set of payoffs has small condition number with high probability. Here we come to a fundamental stumbling block, which is one of the main technical difficulties we overcome.

  To prove that a perturbed weight vector $(\tilde r_{ij})$ is well-conditioned, Loff and Skomra looked at the set of equations that need to be obeyed. As illustrated above, the linear system has inequalities, e.g. for a node $i \in \Maxvertices$, of the form
  \[
    r_{i j} < \lambda + u_i - u_j.
  \]
  The value on the right side of this equation, called the bias-adjusted value, does not, as it turns out, depend on the variable $r_{ij}$. So this is a linear inequality with a one-sided cutoff. Loff--Skomra then argue: if we fix all of the variables except for $r_{ij}$, and perturb the weight $r_{ij}$ with a not-too-small perturbation, then the right-hand side (the cutoff) becomes fixed, and there is a large probability that $r_{i j}$ will be \emph{non-negligibly bounded away from the cutoff}. Union-bounding over all edges, we conclude that with high probability every edge is bounded away from the cutoff, which is to say, the condition number (inverse-distance from the boundary, i.e., inverse-distance from the cutoff) is not too large.

  \medskip\noindent
  \emph{Here's why the old perturbation argument fails:} In the non-ergodic case, the predicate ``$(\sigma,\tau)$ is a bias-induced policy for $r$'' needs to obey additional inequalities. Namely, for each $i \in \Maxvertices$ and each $j$ which is in a different cluster than $i$, we have the inequality $\lambda_i > \lambda_j$. The value $\lambda_i$, expressed in terms of the weights $r_{ij}$, is the average weight of the cycle $C_i$ where the policies $(\sigma,\tau)$ end up, when starting from $i$. So this inequality, written in terms of the weights $r_{ij}$, says:
  \[
    \frac{1}{|C_i|} \sum_{(k,\ell) \in C_i} r_{k \ell} > \frac{1}{|C_j|} \sum_{(k,\ell) \in C_j} r_{k \ell}.
  \]
  Now suppose that we have several clusters, hence several of these inequalities, which is perfectly possible in a non-ergodic graph. Then each edge weight $r_{i j}$ can appear on multiple of these inequalities. The difficulty, when compared to the previous case, is that $r_{ij}$ can appear on \emph{different sides} in different inequalities, so we no longer have equations with a single-cutoff side, we have $A < r_{ij} < B$, (for several different $A$ and $B$, in fact), and it becomes necessary to prove that $r_{i j}$ is likely to be bounded away from $A$ and $B$; this requires us to prove, in particular, that $A$ and $B$ are bounded away from each other with high probability. Now notice that the payoffs $r_{i'j'}$ within $A$ and within $B$ may themselves appear in the center of a different triple $A' < r_{i'j'} < B'$ that also should be satisfied, and it is not clear how one can untangle the dependencies between the different triples. This is where we got stuck when trying to extend the previous proof.

  \medskip\noindent
  \emph{So we need a new strategy.} Here we take inspiration from the smoothed analysis of the simplex method, appearing in \cite{dadush2018friendly,HuibertsLeeZhang:2025,BachHuiberts:2025}. In these papers, they wish to prove that the simplex method (with the shadow vertex pivoting rule) runs in smoothed polynomial time if the input instance is perturbed by Gaussian noise. To do this, they first show the same result for log-Lipschitz noise, and then argue that the Gaussian distribution behaves similarly to a log-Lipschitz distribution.

  A log-Lipschitz distribution is a distribution over $\R^n$ whose probability density function $f$ obeys, for some constant $L \in \R_{> 0}$ and every $x,y \in \R^n$:
  \[
    e^{-L\|x - y\|} \le \frac{\dens(x)}{\dens(y)} \le e^{L\|x - y\|}.
  \]
  The most obvious example is a radially symmetric exponential density $\dens(x) \propto \exp\left(-\|x\|\right)$.

  With a log-Lipschitz distribution $\mu$ over $\R^n$, we can easily relate the probability of an event $A$ with the probability of the translated event $A + v$:
  \[
    \mu(A+v) \ge e^{-L\|v\|} \mu(A),
  \]
  where $L$ is the log-Lipschitz constant of $\mu$. Now suppose that $A$ is a polyhedral cone, and $v$ is a point of $A$ with $\|v\| \ll \frac{1}{L}$ which is ``$\delta$-deep within $A$'', meaning that a ball of radius $\delta$ around $v$ is still contained in $A$. Then, by the above equation, $A + v \subseteq A$ still has an $e^{-L\|v\|}=1-o(1)$-fraction of the mass of $A$, and every point $x$ in $A+v$ is also $\delta$-deep within $A$. If we recall that the condition number is inverse-distance to the boundary of $A$, and we imagine the shifted cone $A + v$ inside the base cone $A$ as having a $\delta$-thick border between its own border and that of $A$, this should mean the condition number of $x$ is $\le \frac{1}{\delta}$, or $\frac{\|x\|}{\delta}$ if we normalize by the norm of $x$. Indeed that is how it will work. And this will happen for every single cone $\mathcal P^{\sigma,\tau}$ in our partition of $\R^E$. Hence if $\mu$ is log-Lipschitz, then the probability that we sample a $\delta$-deep point is $1 - o(1)$. At that point, the proof amounts to finding such $v$, and then adapting it to work with Gaussian noise.

\subsection{Contribution 3. Policy iteration in smoothed polynomial time}\label{sec:contrib_pi}

At this point, the \cite{LoffSkomra:2024} paper shows (Theorem 6.1 of \cite{LoffSkomra:2024}, arXiv v1 version) that if we take discount factor $\gamma > 1 - \frac{1}{\poly(n,\Delta)}$, where $\Delta$ is the condition number, then the optimal strategies for the discounted game are also optimal for the mean-payoff game. One can then use the known policy iteration algorithms for solving discounted games in time $\poly(n, \frac{1}{1-\gamma}) = \poly(n,\Delta)$.

All we need to do here is reprove their Theorem 6.1 so that it works with our new condition number. This is our \Cref{conditioned_discount}, whose proof is an adaptation of Loff and Skomra's proof (of their Theorem 6.1). We give a brief overview of the proof here, so as to compare to the analogous result we prove for value iteration.

The proof begins by fixing the optimal bias-induced mean-payoff policies
$(\sigma,\tau)$.  It then shows that if Max plays $\sigma$ in the discounted game, then
Min's unique best response is still $\tau$; the argument for Max is symmetric.

Let $\tau'$ be a discounted best response to $\sigma$. The proof rules out all possible ways in which $\tau'$ could differ from $\tau$. First, $\tau'$ cannot choose an edge that leaves the component of $\dgraph^{\sigma,\tau}$ containing the current vertex. Such an edge would move to a component whose mean-payoff value is separated from the original one by at least $\|\payoff\|/\Cond(\payoff)$. It can be shown that discounted values differ from mean-payoff values by at most $O(n(1-\gamma)\|\payoff\|)$, so this separation would contradict the assumed bound on $1-\gamma$.

After this, every deviation of $\tau'$ remains inside a single component of $\dgraph^{\sigma,\tau}$. We next show that $\tau'$ cannot create a new cycle inside that component. Summing the bias optimality inequalities around such a cycle gives equality on the edges where $\tau'$ agrees with $\tau$, while any edge on which they differ produces a strict loss controlled by the same lower bound in $\Cond(\payoff)$. Consequently the mean weight of the new cycle is separated from the original mean-payoff value. The discount-to-mean estimate again shows that this separation is too large to be compatible with $\tau'$ being a discounted best response.

A similar argument forbids deviations on the paths leading to the cycle. Hence no deviation is possible, and $\tau$ is the unique discounted best response to $\sigma$. By symmetry, $\sigma$ is the unique discounted best response to $\tau$, and therefore the mean-payoff optimal pair of strategies $(\sigma,\tau)$ is also the unique optimal pair in the discounted game.

\subsection{Contribution 4. Three value-iteration algorithms}

The two main families of algorithms for solving two-player stochastic games are
policy iteration and value iteration.  Algorithms from both families work well
in practice.  The Loff--Skomra paper \cite{LoffSkomra:2024} gives a smoothed
polynomial-time analysis for policy iteration, but it does not address value
iteration.  Our fourth contribution is to introduce three value-iteration
algorithms for deterministic mean-payoff games.

Value iteration algorithms are based on the \emph{Shapley operator} $\shapley \colon \R^n \to \R^n$, defined by
\[
\shapley(x)_i = \begin{cases}
\max_{(i,j) \in \edges} \{\payoff_{ij} + x_j\} &\text{if } i \in \Maxvertices, \\
\min_{(i,j) \in \edges} \{\payoff_{ij} + x_j\} &\text{if } i \in \Minvertices. \\
\end{cases}
\]
For games of constant value, we may solve the mean-payoff game by solving the additive eigenproblem
\begin{equation}\label{eq:ergodic-intro}
\xi \ones + \bias = \shapley(\bias).
\end{equation}
Here $\xi$ is the value of the game and $\bias$ is a bias vector. Thus, from the point of view of value iteration, the problem is to find a vector which satisfies this equation.

Our first two value-iteration algorithms are obtained by applying two classical fixed-point iterations to the Shapley operator: Halpern iteration and \KM iteration. These methods had already appeared in related work on MDPs and stochastic games \cite{GaubertStott:2020,AkianGaubertNaepelsTerver:2023,LeeRyu:2025,LeeBravoCominetti:2025,ZurekChen:2025}, but, as far as we know, this is the first time they are shown to solve deterministic two-player mean-payoff games. The main point is that the Shapley operator is non-expansive in the Hilbert seminorm $\hilb{\cdot}$, and is additively homogeneous. Therefore, \uline{in constant-value games}, the ergodic equation can be viewed as a fixed-point equation in the quotient space $\R^n/\R\ones$, where the Hilbert seminorm induces a true norm, making it a Banach space. Hence, standard convergence bounds for Halpern and \KM iteration imply that the iterates produce vectors $\tilde \bias$ for which for which
\[
\hilb{\tilde\bias-\shapley(\tilde\bias)}
\]
is small.

This gives a useful notion of approximate bias. We do not require $\tilde \bias$ to be close to an actual bias vector. Instead, we only require it to approximately satisfy the ergodic equation, as in, $\hilb{\tilde\bias-\shapley(\tilde\bias)} \le \varepsilon$. This weaker notion is the one naturally produced by \KM and Halpern iteration, and it is sufficient for recovering optimal policies once the approximation is accurate enough, as explained in the
next subsection.

The same two iterations can also be run directly on arbitrary, non-constant value games. In this case, their normalized iterates approximate the value vector and eventually ignore all edges between vertices of different values. After this happens, the game decomposes into constant-value subgames, and the same approximate-bias mechanism applies inside each component. Thus plain \KM iteration and plain Halpern iteration already give algorithms for arbitrary deterministic mean-payoff games.

Our third algorithm is designed to improve the resulting complexity bound. It combines Zwick--Paterson iteration with Halpern iteration. The Zwick--Paterson part is used to detect edges between vertices of different values, i.e. in different clusters, while the Halpern part is repeatedly restarted on the currently detected clusters. Once all cross-value edges have been found, the Halpern iteration runs directly on the constant-value components. This avoids the loss that appears in the analysis of plain Halpern iteration on general games, and leads to a better dependence on the condition number.

\subsection{Contribution 5. Value iteration in smoothed polynomial time}

Our fifth contribution is to prove that the value-iteration algorithms described above run in polynomial time on well-conditioned inputs. This is the value iteration analogue of what Contribution 3 said about policy iteration, but the proof is quite different. For policy iteration, the key step was to compare the discounted and mean-payoff games and show that, for $\gamma$ sufficiently close to $1$, they have the same optimal policies. For value iteration, there is no discounted game in the background. Instead, the proof uses the fact that approximate solutions of the ergodic equation can be interpreted as exact solutions for nearby games.

More precisely, \uline{starting with constant-value games}, if a vector $\tilde \bias$ satisfies
\[
\hilb{\tilde \bias - \shapley(\tilde \bias)} \le \varepsilon,
\]
then $\tilde \bias$ is an exact bias for a constant-value game whose weights are within $\varepsilon$ of the original weights. Moreover, this perturbation is obtained by shifting all outgoing edges of each vertex by the same amount, so $\tilde\bias$ induces the same policies for both the original weights and the modified weights.

To turn this approximate statement into an exact one, we use the geometry of the
polyhedral regions defined by the bias-induced policies.  In the constant-value,
single-optimal-cycle case, a well-conditioned weight vector $\payoff$ lies far from the
boundary of the region corresponding to its bias-induced policy pair
$(\sigma,\tau)$.  Therefore, all sufficiently small perturbations of the weights
remain in the same region and have the same bias-induced policies.  Since an
approximate bias $\tilde\bias$ for $\payoff$ is an exact bias for a nearby game $\tilde\payoff$, if $\tilde\bias$ is sufficiently accurate it must induce the same optimal policies for $\tilde\payoff$, namely $(\sigma,\tau)$. Since we may choose $\tilde\payoff$ so that $\tilde \bias$ induces the same bias for $\payoff$ and $\tilde\payoff$, then $\tilde\bias$ also reveals the optimal policies of the original game $\payoff$.  The condition number quantifies how accurate the approximation has to be, and the convergence estimates for \KM and Halpern iteration then give polynomial bounds in the condition number.

For general games, the value vector need not be constant.  The analysis first
shows that value iteration separates the edges connecting vertices of different
values. This can be done by any of the three value-iteration procedures we
study: Zwick--Paterson iteration, \KM iteration, or Halpern iteration.  After
these edges are ignored, the game decomposes into constant-value subgames, and
the preceding approximate-bias analysis applies independently inside each
component.

This gives three algorithms with different performance guarantees.  Plain \KM
iteration solves well-conditioned games after a number of iterations polynomial
in $n$ and our condition number $\Cond(\payoff)$, with a quadratic dependence on $\Cond$.  Plain Halpern iteration also works on arbitrary games, but our analysis
again gives a quadratic dependence, because after the value-separation phase it
behaves like a shifted Halpern iteration on the reduced game.  The third
algorithm, which runs Zwick--Paterson iteration in parallel with Halpern
iteration and resets the latter whenever a new cross-value edge is detected,
avoids this loss and gives the better bound of
\[
O\bigl(n+n^2\Cond(\payoff)\bigr)
\]
iterations of the main loop.\\

\noindent \textbf{Concurrent work.} Concurrent to this work, Asadi, Chatterjee, and Luo \cite{AsadiChatterjeeLuo:2027} have independently established smoothed polynomial-time complexity for deterministic discounted and mean-payoff games on general graphs. Their work focuses on a restarting variant of policy iteration and provides expected-runtime guarantees which are stronger than our high-probability result described in \cref{sec:contrib_pi}. In contrast, our paper extends the smoothed analysis to value iteration algorithms (\KM, Halpern, and Zwick–Paterson/Halpern), which are not covered in \cite{AsadiChatterjeeLuo:2027}. Both works remove the ergodicity assumption of \cite{LoffSkomra:2024} but employ distinct technical tools: our analysis is based on a condition number and log-Lipschitz distributions, whereas \cite{AsadiChatterjeeLuo:2027} relies on Blackwell optimality and suboptimality gap separations.

\subsection{Organization of the paper}

The preceding overview presented the contributions of the paper in conceptual order: first the extension of the Loff--Skomra condition-number framework beyond ergodic graphs, then the smoothed analysis of this condition number, and finally the resulting policy-iteration and value-iteration algorithms.  The body of the paper is organized somewhat differently, following the logical dependencies of the proof.

In \Cref{sec:preliminaries}, we collect the background material used throughout the paper.  We recall the basic definitions and facts about mean-payoff games, the Shapley operator and the Hilbert seminorm, the Gurvich--Karzanov--Khachiyan optimality equations, bias-induced policies, and the decomposition of an arbitrary game into constant-value components.

The next two sections develop the theory in the setting that closely approximates the ergodic case.  In \Cref{sec:constant_value}, we study constant-value games.  We introduce approximate biases, show that an approximate bias can be viewed as an exact bias for a nearby constant-value game, and prove that \KM and Halpern iteration produce such approximate biases. In \Cref{sec:one_cycle_condition}, we study the polyhedral cone decomposition of the weight space by optimal bias-induced policies, in the case of a single optimal cycle. This gives the first link between a condition number and the number of value-iteration steps needed to recover optimal policies.

In \Cref{sec:general_condition}, we pass from this special constant-value setting to arbitrary game graphs.  We define the polyhedral cone decomposition for the general, not necessarily ergodic case, prove that these cones are disjoint and cover almost all weight vectors, and use them to define the global condition number. This section is the formal generalization of the Loff--Skomra condition-number framework from ergodic graphs to general graphs.

In \Cref{sec:value_decomposition}, we address the main new difficulty for value iteration on non-constant-value games: such games must first be decomposed into constant-value components.  We introduce a value condition number measuring the gap between distinct state values, and show that Zwick--Paterson, \KM, and Halpern iteration can detect, or eventually ignore, all edges connecting vertices of different values, in time polynomial in the inverse of the smallest gap. This reduces the analysis of general games to the constant-value analysis developed earlier.

In \Cref{sec:complexity_condition}, we pursue the desired algorithmic consequences of the preceding sections.  First, we obtain complexity bounds for plain \KM iteration, plain Halpern iteration, and the combined Zwick--Paterson/Halpern algorithm, which are polynomial in our new condition number. Then we do the same for policy iteration, showing that the policy iteration algorithms of \cite{LoffSkomra:2024} also work for non-ergodic graphs. This is done by proving that on well-conditioned instances, sufficiently-discounted games have the same unique optimal policies as the corresponding mean-payoff game.

Finally, \Cref{sec:probabilistic_condition} proves that sufficiently perturbed inputs are well conditioned with high probability. The proof first treats log-Lipschitz perturbations, where the cone geometry can be combined with a simple translation argument, and then transfers the estimate to Gaussian perturbations using Gaussian shift inequalities.  Combining this probabilistic bound with the deterministic bounds from \Cref{sec:complexity_condition} yields the smoothed polynomial-time guarantees for both the value-iteration algorithms and the policy-iteration algorithms.

\newpage
\section{Preliminaries}\label{sec:preliminaries}

\subsection{Definition of mean-payoff games}\label{sec:preliminaries-MPGs}

A \emph{(deterministic, two-player) mean-payoff game (MPG)} is a game played on directed weighted graph $\dgraph = (V, \edges, \payoff)$ in which the $|V| = n$ vertices are split into two disjoint sets, $V = [n] = \Maxvertices \dunion \Minvertices$. We say that $\Maxvertices$ is the set of vertices controlled by \emph{player Max} and $\Minvertices$ is the set of vertices controlled by \emph{player Min}. One of these sets can be empty: in this case, we get a one-player game, called a ``(Deterministic) Markov Decision Process'' (DMDP). Furthermore, we denote by $m \coloneqq \card{\edges}$ the number of edges of the graph. The game is played as follows. At the beginning of the game, a token is placed on one of the vertices $i_0 \in [n]$. Subsequently, the player controlling that vertex chooses an outgoing edge $(i_0, i_1) \in \edges$ and the token is moved to $i_1$. Then, the player controlling $i_1$ chooses the next move of the token, and this procedure is repeated ad infinitum. The final payoff received by player Max is defined as
\[
\liminf_{N \to \infty} \frac{1}{N}(\payoff_{i_0i_1} + \payoff_{i_1i_2} + \dots + \payoff_{i_{N-1}i_N}) \, .
\]
The objective of player Max is to maximize the payoff, while the objective of player Min is to minimize the payoff. This means that the game is of zero sum. Moreover, the players try to optimize the average payoff of the game in the long run, which justifies the name ``mean-payoff game''. The game is well-defined provided that each player can always make the next move, i.e., that the graph satisfies the following assumption.

\begin{assumption}
Every vertex of $\dgraph$ has at least one outgoing edge.
\end{assumption}

In the remaining part of the paper, we always assume that $\dgraph$ satisfies the assumption above. We further assume that $\dgraph$ has no multiple edges, but it can have loops. The weights $\payoff \in \R^m$ of the graph are often called \emph{payoffs} because they correspond to immediate payoffs received by player Max in the game and the final payoff of player Max is often called the \emph{reward}. However, since in this work we study what happens when $\payoff$ changes or is picked at random, for the sake of clarity we prefer to use a more distinct terminology: we refer to $\payoff$ as the \emph{weights} of the underlying graph, and we reserve the word ``payoff'' for the final outcome of the game.

There are many ways to play a mean-payoff game, but the most relevant to us is to consider the situation in which players use policies. A \emph{policy of player Max} is a function $\sigma \colon  \Maxvertices \to [n]$ that associates to a given vertex controlled by Max a possible transition, i.e., satisfies $(i,\sigma(i)) \in \edges$ for all $i \in \Maxvertices$. Analogously, a \emph{policy of player Min} is a function $\tau \colon \Minvertices \to [n]$ such that $(i,\tau(i)) \in \edges$ for all $i \in \Minvertices$. Policies are \emph{memoryless} strategies because they do not depend on the history of the game and are \emph{stationary} because they only depend on the current state of the game. If both players play according to policies $(\sigma,\tau)$, then the token reaches some cycle $(i_k, i_{k+1}, \dots, i_{k+l-1})$ of the graph and stays there forever. In particular, the resulting payoff of Max is given by
\[
g_{i_{0}}(\sigma,\tau) = \frac{1}{l}(\payoff_{i_ki_{k+1}} + \dots + \payoff_{i_{k+l-1}i_k}) \, .
\]
In other words, the payoff is the mean weight of the final cycle reached by the token. It turns out that mean-payoff games have a value and a pair of optimal policies.

\begin{theorem}[\cite{liggett_lippman,ehrenfeucht_mycielski}]\label{th:value_mpg}
For any mean-payoff game there exists a \emph{unique} vector $\gameval \in \R^n$ and a (not necessarily unique) pair of policies $(\sigma^*,\tau^*)$ such that for every pair of policies $(\sigma,\tau)$ and every initial state $i \in [n]$ we have
\begin{equation}\label{eq:optimal_policies}
g_i(\sigma,\tau^*) \le \gameval_i \le g_i(\sigma^*,\tau) \, .
\end{equation}
\end{theorem}

\begin{definition}[Value $\gameval$ and optimal policies $(\sigma^*,\tau^*)$]
  The unique vector $\gameval$ given above is called the \emph{value} of the game. The $i$th coordinate of $\gameval$ is called the \emph{value of state $i$}.
  Any policies $(\sigma^*,\tau^*)$ satisfying \cref{eq:optimal_policies} for all $i,\sigma,\tau$ are called \emph{optimal}. To any cycle which is reached by a particular pair of optimal policies we will call an \emph{equilibrium cycle}.
\end{definition}

We should emphasize that a mean-payoff game can have more than one pair of optimal policies, with different equilibrium cycles.

An intuitive interpretation of optimality is as follows: if player Max uses an optimal policy $\sigma^*$, then they are guaranteed to receive at least $\gameval_i$, no matter what Min does. Analogously, if Min uses an optimal policy $\tau^*$, then they are guaranteed that the final payoff is at most $\gameval_i$, no matter what Max does. The theorem above can actually be stated in a stronger way: optimal policies are optimal against \emph{any} possible strategy of the opponent. In other words, by using $\sigma^*$, Max is guaranteed a payoff at least $\gameval_i$ even if Min uses an arbitrarily complicated strategy that depends on the history of the game. See \cite{payoffchapter} for a general introduction to mean-payoff games.

Some of our results use policies that are ``almost'' optimal instead of being optimal. This is formalized in the following definition.

\begin{definition}[$\varepsilon$-optimal policies]
Let $\gameval \in \R^n$ be the value of a mean payoff game. For any $\varepsilon \ge 0$ we say that a policy $\sigma$ of player Max is $\varepsilon$-optimal if it satisfies
\[
g_i(\sigma,\tau) \ge \gameval_i - \varepsilon
\]
for all $i \in [n]$ and every policy $\tau$ of player Min. Analogously, a policy $\tau$ of player Min is $\varepsilon$-optimal if it satisfies
\[
g_i(\sigma,\tau) \le \gameval_i + \varepsilon
\]
for all $i \in [n]$ and every policy $\sigma$ of player Max. In particular, $0$-optimal policies are optimal.
\end{definition}

To analyze our algorithms, we must understand what happens if the weights of the mean payoff game are slightly perturbed. The fundamental observation is that the value of the game also changes slightly, because the value function is $1$-Lipschitz. This is well known and we give the proof for the sake of completeness.

\begin{lemma}\label{value_lipschitz}
Fix a game graph $\dgraph$ and let $\payoff, \tilde{\payoff} \in \R^m$ be two weight vectors. Moreover, let $\gameval, \tilde{\gameval} \in \R^n$ be the values of the mean payoff games with weights $\payoff$ and $\tilde{\payoff}$ respectively. Then, we have $\supnorm{\gameval - \tilde{\gameval}} \le  \supnorm{\payoff - \tilde{\payoff}}$.
\end{lemma}
\begin{proof}
By the definition of maximum and minimum, we have the weak duality $\min_\tau \max_\sigma  g_i(\sigma,\tau) \ge \max_\sigma \min_\tau  g_i(\sigma,\tau)$, where the max and min are taken over all policies of the players Max and Min respectively. The non-trivial direction follows from the existence of optimal policies (\cref{th:value_mpg}), meaning, for every $i \in [n]$ we have $\min_\tau \max_\sigma  g_i(\sigma,\tau) \le \gameval_i \le \max_\sigma \min_\tau  g_i(\sigma,\tau)$. Hence  $\gameval_i = \max_{\sigma} \min_{\tau} g_i(\sigma,\tau) = \min_{\tau} \max_{\sigma} g_i(\sigma,\tau)$. Furthermore, note if we vary $\payoff$, then each function $g_i(\sigma,\tau)$ is linear (because it is equal to the mean weight of some cycle). Consider the univariate function $\phi \colon \R \to \R$ defined as $\phi(\alpha) \coloneqq \min_{\tau} \max_{\sigma} g_i(\sigma,\tau)\bigl(\payoff + \alpha(\tilde{\payoff} - \payoff) \bigr) $. Then, $\phi$ is continuous and piecewise-affine as a min-max of affine functions. We also have $\phi(0) = \gameval_i$ and $\phi(1) = \tilde{\gameval}_i$. Furthermore, note that every affine piece of $\phi$ is $\supnorm{\payoff - \tilde{\payoff}}$-Lipschitz. Hence, the entire function $\phi$ is $\supnorm{\payoff - \tilde{\payoff}}$-Lipschitz and we get $\abs{\gameval_i - \tilde{\gameval}_i} = \abs{\phi(0) - \phi(1)} \le \supnorm{\payoff - \tilde{\payoff}}$.
\end{proof}

\begin{corollary}\label{perturbed_value}
Fix a game graph $\dgraph$ and let $\payoff, \tilde{\payoff} \in \R^m$ be two weight vectors. Furthermore, suppose that $\sigma$ is a policy of Max that is optimal in the game with weights $\payoff$. Then, $\sigma$ is $2\supnorm{\payoff - \tilde{\payoff}}$-optimal in the game with weights $\tilde{\payoff}$. The same is true for policies of player Min.
\end{corollary}
\begin{proof}
Let $\gameval, \tilde{\gameval} \in \R^n$ be the values of the games with weights $\payoff, \tilde{\payoff}$ respectively. By \cref{value_lipschitz} we get $\supnorm{\gameval - \tilde{\gameval}} \le  \supnorm{\payoff - \tilde{\payoff}}$. Let $\tilde{\gameval}^{\sigma}$ be the value of the one-player game with weights $\tilde{\payoff}$ obtained by fixing $\sigma$. Then, by applying \cref{value_lipschitz} to the graph obtained by fixing $\sigma$ and using the optimality of $\sigma$ we get $\supnorm{\gameval - \tilde{\gameval}^{\sigma}} \le  \supnorm{\payoff - \tilde{\payoff}}$. Hence $\supnorm{\tilde{\gameval}^{\sigma} - \tilde{\gameval}} \le 2\supnorm{\payoff - \tilde{\payoff}}$ and therefore $\tilde{\gameval}^{\sigma}_i \ge \tilde{\gameval}_i - 2\supnorm{\payoff - \tilde{\payoff}}$ for all $i \in [n]$. In particular, $\sigma$ is $2\supnorm{\payoff - \tilde{\payoff}}$-optimal. The proof for player Min is analogous.
\end{proof}

\subsection{Preliminaries on Shapley operator and Hilbert seminorm}\label{sec:preliminaries-shapley}

In this section, we give some preliminary information about the Shapley operators of a mean payoff game and about the Hilbert seminorm, which we use heavily.

\begin{definition}[Shapley operator $\shapley$]
The \emph{Shapley operator} of a mean payoff game is the function $\shapley \colon \R^n \to \R^n$ defined by
\[
\shapley(x)_i = \begin{cases}
\max_{(i,j) \in \edges} \{\payoff_{ij} + x_j\} &\text{if } i \in \Maxvertices, \\
\min_{(i,j) \in \edges} \{\payoff_{ij} + x_j\} &\text{if } i \in \Minvertices. \\
\end{cases}
\]
\end{definition}

We use the following notation: if $x, y \in \R^n$ are two vectors then we denote $x \le y$ if the inequality $x_i \le y_i$ holds for all $i \in [n]$. Moreover, we denote $\ones \coloneqq (1,1,\dots,1)$. We will often use the following two properties of the Shapley operator.

\begin{lemma}
The Shapley operator is
\begin{itemize}
\item monotone, i.e., if $x,y \in \R^n$ are such that $x \le y$, then $\shapley(x) \le \shapley(y)$;
\item additively homogeneous, i.e., $\shapley(\alpha \ones + x) = \alpha \ones + \shapley(x)$ for all $x \in \R^n$, $\alpha \in \R$.
\end{itemize}
\end{lemma}
\begin{proof}
Both properties follow easily from the definition and the properties of the max and min functions.
\end{proof}

\begin{definition}[Hilbert seminorm $\hilb{\cdot}$]
  The \emph{Hilbert seminorm}
  of a vector $x \in \R^n$ is defined as
\[
\hilb{x} \coloneqq \max_i{x_i} - \min_i{x_i} = \max_{i,j}\abs{x_i - x_j} = 2 \min_{\alpha \in \R}\supnorm{x + \alpha \ones}\, .
\]
\end{definition}

\begin{lemma}\label{hilbert}
The Hilbert seminorm has the following properties:
\begin{itemize}
\item for all $x,y \in \R^n$ we have $\hilb{x + y} \le \hilb{x} + \hilb{y}$,
\item for all $x \in \R^n$, $\alpha \in \R$ we have $\hilb{\alpha x} = \abs{\alpha} \hilb{x}$,
\item we have $\hilb{x} = 0$ if and only if $x = \alpha \ones$ for some $\alpha \in \R$.
\end{itemize}
In particular, the Hilbert seminorm is indeed a seminorm.
\end{lemma}
\begin{proof}
The first property follows from the triangle inequality,
\[
\hilb{x + y} = \max_{i,j}\abs{(x_i+y_i) - (x_j+y_j)} \le \max_{i,j}(\abs{x_i-x_j} + \abs{y_i-y_j}) \le \hilb{x} + \hilb{y} \, .
\]
The second property follows by
\[
\hilb{\alpha x} = \max_{i,j}\abs{\alpha x_i - \alpha x_j} = \abs{\alpha}\max_{i,j}\abs{x_i - x_j} = \abs{\alpha}\hilb{x} \, .
\]
The third property follows immediately from the definition.
\end{proof}

The fact that the Shapley operator is monotone and additively homogeneous implies that it is non-expansive in the Hilbert seminorm. We give a proof of this fact for the sake of completeness and we refer the reader to \cite{GunawardenaKeane:1995,gaubert_gunawardena} for more information.

\begin{lemma}[\cite{GunawardenaKeane:1995}]\label{nonexp_hilbert}
A Shapley operator of a mean payoff game is non-expansive in the Hilbert seminorm, i.e., for all $x, y \in \R^n$ we have $\hilb{\shapley(x) - \shapley(y)} \le \hilb{x - y}$.
\end{lemma}
\begin{proof}
Fix $x,y \in \R^n$ and denote $\alpha \coloneqq \max_i(x_i - y_i)$, $\beta \coloneqq \min_{i}(x_i - y_i)$, so that $\hilb{x - y} = \alpha - \beta$. Under this notation, we have $\beta \ones + y \le x \le \alpha \ones + y$.  Hence, by monotonicity and additive homogeneity we get $\beta \ones + \shapley(y) \le \shapley(x) \le \alpha \ones + \shapley(y)$. This implies that $\max_i(\shapley(x)_i - \shapley(y)_i) \le \alpha$ and $\min_i(\shapley(x)_i - \shapley(y)_i) \ge \beta$. Hence $\hilb{\shapley(x) - \shapley(y)} \le \alpha - \beta = \hilb{x - y}$.
\end{proof}

\subsection{Preliminaries on the GKK theorem and bias induced policies}\label{sec:preliminaries-gkk}

\begin{definition}[Induced policies]\label{def:induced_policy}
Given a vector $u \in \R^n$, we say that a policy $\sigma$ of player Max is \emph{induced by $u$} if it satisfies
\[
\sigma(i) \in \argmax_{j: (i,j) \in \edges} \{\payoff_{ij} + u_j\}
\]
for all $i \in \Maxvertices$. Analogously, we say that a policy $\tau$ of player Min is induced by $u$ if it satisfies
\[
\tau(i) \in \argmin_{j: (i,j) \in \edges} \{\payoff_{ij} + u_j\}
\]
for all $i \in \Minvertices$.
\end{definition}

In our value iteration algorithms, we will find a vector $u \in \R^n$ such that the policies induced by $u$ are optimal. The next theorem guarantees that such a vector exists.

\begin{theorem}[\cite{gurvich}]\label{optimality_eq}
Given a deterministic mean payoff game, the following system has a solution $(\gameval, \bias) \in \R^n \times \R^n$:
\begin{align*}
\forall i \in \Maxvertices, \quad &\gameval_{i} = \max_{(i,j) \in \edges} \gameval_{j}, \\
\forall i \in \Minvertices, \quad &\gameval_{i} = \min_{(i,j) \in \edges} \gameval_{j}, \\
\forall i \in \Maxvertices, \quad &\gameval_{i} + \bias_i = \max_{(i,j) \in \edges, \gameval_j = \gameval_i} \{\payoff_{ij} + \bias_{j}\} > \max_{(i,j) \in \edges, \gameval_j \neq \gameval_i} \{\payoff_{ij} + \bias_{j}\}, \\
\forall i \in \Minvertices, \quad &\gameval_{i} + \bias_i = \min_{(i,j) \in \edges, \gameval_j = \gameval_i} \{\payoff_{ij} + \bias_{j}\} < \min_{(i,j) \in \edges, \gameval_j \neq \gameval_i} \{\payoff_{ij} + \bias_{j}\}.
\end{align*}
Furthermore, if $(\gameval, \bias)$ is a solution, then $\gameval$ is the value vector of the game (in particular, $\gameval$ is unique). Even more, if a policy $\sigma$ of player Max is induced by $\bias$, then it is optimal. Likewise, if a policy $\tau$ of player Min is induced by $\bias$, then it is optimal.
\end{theorem}

\begin{definition}[Bias $\bias$]\label{def:bias}
Let $\gameval \in \R^n$ be the value vector of a mean payoff game. We say that a vector $\bias \in \R^n$ is a \emph{bias} of this game if $(\gameval,\bias)$ is a solution of system from \cref{optimality_eq}. Moreover, we say that a policy of one of the players is \emph{bias-induced} if it is induced by some bias.
\end{definition}

\begin{remark}
We note that in the theory of Markov Decision Processes, the terms ``bias'' and ``bias-optimal policy'' are used to describe a closely related, but slightly different concept. In this theory, the bias arises from describing the mean-payoff problem as a limit of discounted problems. This leads to considering ``sensitive discount optimality equations'' or ``multichain optimality equations'', see, e.g. \cite[Chapters 9 and 10]{puterman} for more information. For the simplicity of terminology, we use the word ``bias'' to refer to the concept from \cref{def:bias}, even though it does not perfectly match the use of this term in the theory of Markov Decision Processes.
\end{remark}

A particular case of optimality equations arises when the value of the game is constant, i.e. when it does not depend on the initial state. Then, the system from \cref{optimality_eq} can be reduced to the ergodic equation.

\begin{definition}[Constant value game]
We say that a mean payoff game has \emph{constant value} if its value does not depend on the initial vertex, i.e., if $\gameval = \xi \ones$ for some $\xi \in \R$.
\end{definition}

\begin{corollary}\label{ergodic_equation}
A mean payoff game with Shapley operator $\shapley$ has constant value if and only if the \emph{ergodic equation}
\begin{equation}\label{eq:ergodic}
\xi \ones + \bias = \shapley(\bias)
\end{equation}
has a solution $(\xi, \bias) \in \R \times \R^{n}$. Moreover, if $(\xi, \bias)$ is a solution of the ergodic equation, then $\xi$ is the value of every state of the game. In particular, $\xi$ is unique. Furthermore, a vector $\bias$ is a bias of a constant value game if and only if $(\xi, \bias)$ satisfies the ergodic equation.
\end{corollary}
\begin{proof}
Suppose that a mean payoff game has constant value and let $\bias$ be a bias given by \cref{optimality_eq}. Then, the pair $(\xi, \bias)$ satisfies the ergodic equation. Conversely, if $(\xi, \bias)$ satisfy the ergodic equation, then $(\xi \ones, \bias)$ satisfy the system from \cref{optimality_eq}, so $\bias$ is a bias and the value of the game is equal to $\xi \ones$.
\end{proof}

To prove our results, we will need a bound on the Hilbert seminorm of some bias $\bias^*$. Such estimates appeared in \cite[Theorem~13]{AllamigeonGaubertKatzSkomra:2025} in the context of stochastic mean payoff games with constant value. For deterministic games, the following proposition gives a slightly better estimate and can be applied to games with non-constant value. Our proof follows the proof of a related bound given in \cite{gurvich}.

\begin{proposition}\label{short_bias}
For any $\varepsilon >0$, there exists a bias $\bias^*$ such that $\supnorm{\bias^*} \le (1+\varepsilon)n\hilb{\payoff}$. Likewise, there exists a bias $\bias^*$ such that $\hilb{\bias^*} \le (2+\varepsilon)n\hilb{\payoff}$. If the game has constant value, then the above is true also for $\varepsilon = 0$.
\end{proposition}

\begin{proof}
Let $\gameval \in \R^n$ be the value vector of the game and let $(\sigma,\tau)$ be a pair of policies induced by some bias. Fix $\varepsilon > 0$ and consider the polyhedron formed by all $\bias \in \R^n$ such that
\begin{equation}\label{eq:tu_system}
\begin{aligned}
&\forall (i,j) \in \edges, i \in \Maxvertices, j = \sigma(i), & \qquad &&\bias_i - \bias_j &= \payoff_{ij} - \gameval_i \\
&\forall (i,j) \in \edges, i \in \Minvertices, j = \tau(i), & \qquad&&\bias_i - \bias_j &= \payoff_{ij} - \gameval_i \\
&\forall (i,j) \in \edges, i \in \Maxvertices, j  \neq \sigma(i),
\gameval_j = \gameval_i,& \qquad &&\bias_i - \bias_j & \ge \payoff_{ij} - \gameval_i \\
&\forall (i,j) \in \edges, i \in \Minvertices, j \neq \tau(i), \gameval_j  = \gameval_i, & \qquad &&\bias_i - \bias_j & \le \payoff_{ij} - \gameval_i \\
&\forall (i,j) \in \edges, i \in \Maxvertices, j \neq \sigma(i),
\gameval_j \neq \gameval_i, & \qquad &&\bias_i - \bias_j &\ge \payoff_{ij} - \gameval_i + \varepsilon \\
&\forall (i,j) \in \edges, i \in \Minvertices, j \neq \tau(i), \gameval_j \neq \gameval_i, & \qquad &&\bias_i - \bias_j &\le \payoff_{ij} - \gameval_i - \varepsilon \,
\end{aligned}
\end{equation}
Note that this polyhedron is nonempty for every sufficiently small $\varepsilon$ because $(\sigma,\tau)$ are induced by some bias. Furthermore, any $\bias$ that belongs to this polyhedron is a bias of the game. This polyhedron is described by a matrix which has exactly one $+1$ and one $-1$ entry in every row. Such a matrix is totally unimodular as discussed in \cite[Section 19.3, Example 2]{schrijver}. Hence, the proof of \cite[Theorem 19.1]{schrijver} shows that, up to permuting the indices, there exists a bias of the form $\bias^* = \begin{pmatrix}U^{-1} b \\ 0\end{pmatrix}$, where the matrix $U$ is totally unimodular and $b$ is a vector with coordinates coming from the right-hand side of \cref{eq:tu_system}. Since $\gameval_l$ is in the interval $\interval{\min_{ij} \payoff_{ij}}{\max_{ij} \payoff_{ij}}$ for all $l \in [n]$, we have $\supnorm{b} \le \hilb{\payoff} + \varepsilon$. Therefore, Cramer's rule shows that $\supnorm{\bias^*} \le n\hilb{\payoff} + n\varepsilon$. Since $\hilb{\bias^*} \le 2\supnorm{\bias^*}$, we get $\hilb{\bias^*} \le 2n\hilb{\payoff} + 2n\varepsilon$. Moreover, if the game has constant value, then the inequalities in \cref{eq:tu_system} do not use $\varepsilon$. To get the multiplicative bound, notice that we can use any $\varepsilon$, in particular $\varepsilon \hilb{\payoff}$.
\end{proof}

We also note the following result about biases of constant value games.

\begin{theorem}[{\cite[Theorem~3.10]{generic_uniqueness}}]\label{biases_connected}
Suppose that a mean payoff game has constant value. Then, the set of biases is connected.
\end{theorem}

To finish this preliminary section, we discuss the decomposition of a mean payoff game into a collection of constant value games. This is done by removing the edges between vertices of different values.

\begin{definition}[Reduced game]
Given a mean payoff game on a digraph $\dgraph = (\vertices,\edges)$ with value vector $\gameval \in \R^n$, let
\[
\edges' \coloneqq \{(i,j) \in \edges \colon \gameval_i = \gameval_j\}.
\]
The \emph{reduced game} is the mean payoff game played on the graph
$\dgraph' = (\vertices,\edges')$, with the same weights as in the original game.

Let $\vertices_1,\dots,\vertices_p$ be the connected components of the underlying undirected graph of $\dgraph'$. For every $\ell \in [p]$, let
\[
\edges_\ell \coloneqq \edges' \cap (\vertices_\ell \times \vertices_\ell),
\qquad
\dgraph_\ell \coloneqq (\vertices_\ell,\edges_\ell).
\]
We call the games played on $\dgraph_1,\dots,\dgraph_p$ the connected components of the reduced game.
\end{definition}

Note that if $\sigma$ is a policy of player Max on the reduced graph $\dgraph'$ and $\dgraph''$ is a connected component of $\dgraph'$, then player Max can restrict the policy $\sigma$ to the vertices of $\dgraph''$ and use this restricted policy in the game played on $\dgraph''$. This is because $\sigma$ does not use edges that go between different connected components of $\dgraph'$. By abuse of notation, we will still refer to this restricted policy as $\sigma$ and use the same convention for policies $\tau$ of player Min. The following known proposition discusses the properties of the reduced game and its connected components. We give the proof for the sake of completeness.

\begin{proposition}\label{game_decomposition}
Fix a mean payoff game on the digraph $\dgraph$ with value vector $\gameval$. Then, (i) its reduced game played on $\dgraph'$ has the same value as the game on $\dgraph$. Furthermore, if $\dgraph_{1}, \dots, \dgraph_{p}$ are the connected components of $\dgraph'$, then for every $\ell \in [p]$ the game on $\dgraph_{\ell}$ has constant value. More precisely, if $\vertices_{\ell}$ is the vertex set of $\dgraph_{\ell}$, then the value of the game $\dgraph_{\ell}$ is equal to $\gameval_i$ for any $i \in \vertices_{\ell}$. Moreover, (ii) if $\sigma$ is a policy of player Max on the reduced graph $\dgraph'$ that is optimal on all games on $(\dgraph_{\ell})_{\ell}$, then $\sigma$ is optimal in the game played on the original graph $\dgraph$. An analogous statement holds for player Min.
\end{proposition}
\begin{proof}
  (i) Let $\bias$ be a bias of the game played on $\dgraph$. By examining the system from \cref{optimality_eq}, we see that the pair $(\gameval,\bias)$ still solves this system if we remove from the graph the edges $\{(i,j) \in \edges \colon \gameval_i \neq \gameval_j\}$. In particular, $\bias$ is a bias of the game played on the reduced graph, and $\gameval$ is the value of the reduced game. Let $\dgraph_{\ell}$ be a connected component of the reduced game and denote by $\vertices_{\ell}, \edges_{\ell}$ the vertex set and the edge set of $\dgraph_{\ell}$. By definition of the connected component, the value $\gameval_i$ is the same for all $i \in \vertices_{\ell}$. Furthermore, the restricted vector $\tilde{\bias} \in \R^{\vertices_{\ell}}$ defined by $\tilde{\bias}_i \coloneqq \bias_i$ for $i \in \vertices_{\ell}$ satisfies the equations
\begin{align*}
\forall i \in \Maxvertices \cap \vertices_{\ell}, \quad &\gameval_{i} + \tilde{\bias}_i = \max_{(i,j) \in \edges_{\ell}} \{\payoff_{ij} + \tilde{\bias}_{j}\} \, , \\
\forall i \in \Minvertices \cap \vertices_{\ell}, \quad &\gameval_{i} + \tilde{\bias}_i = \min_{(i,j) \in \edges_{\ell}} \{\payoff_{ij} + \tilde{\bias}_{j}\} \, .
\end{align*}
Since the value $\gameval_{i}$ does not depend on $i \in \vertices_{\ell}$, the pair $(\gameval_i, \tilde{\bias})$ satisfies the ergodic equation on $\dgraph_{\ell}$ and, by \cref{ergodic_equation}, the game on this graph has constant value equal to $\gameval_i$.

(ii) Suppose that $\sigma$ is a policy for player Max on the reduced graph $\dgraph'$ that is optimal for all games on $(\dgraph_{\ell})_{\ell}$. Let $\tau$ be any policy of player Min in the game played on the entire graph $\dgraph$. Consider $\tau(i)$ for some $i \in \Minvertices$, and notice that, because $\gameval$ is the value of the game on $\dgraph$, then $\tau(i)$ either chooses a vertex from the same component, in which case $\gameval_{\tau(i)} = \gameval_i$, or otherwise it will choose a vertex from a component which is not better for player Min, and so in either case we always have $\gameval_{\tau(i)} \ge \gameval_{i}$. Moreover, for every $i \in \Maxvertices$ we have $\gameval_{\sigma(i)} = \gameval_{i}$ because $\sigma$ uses only edges from $\edges'$. We have then concluded that traversing an edge from the policies $(\sigma,\tau)$ can only preserve or increase the value. Hence, if we fix $i \in \vertices$ and consider the directed cycle reached by the path given by fixing $(\sigma,\tau)$ and starting from $i$, then all the vertices $j$ of this cycle have the same value $\gameval_j$ and this value is not smaller than $\gameval_i$. Hence, the vertices of the cycle belong to the same set $\vertices_{\ell}$ and the mean weight of the cycle is at least $\gameval_j \ge \gameval_i$ by the fact that $\sigma$ is optimal on $\dgraph_{\ell}$. Hence, $\sigma$ guarantees the payoff at least $\gameval_i$ for any initial vertex $i$ in the game played on the entire graph $\dgraph$. An analogous reasoning is valid for a policy $\tau$ of player Min.
\end{proof}

\Cref{game_decomposition} suggests the following two stage algorithm for solving any mean payoff game. First, we find the edges $(i,j)$ such that $\gameval_{i} \neq \gameval_{j}$. Second, we remove these edges from the graph and then solve a collection of constant value games played on the connected components of the reduced game. We will show that, for a randomly perturbed game, value iteration algorithms can solve both tasks efficiently.

\section{Approximate solutions to constant value games via Halpern or KM iteration}\label{sec:constant_value}

In this section, we study the properties of games with constant value and we show that the \KM iteration and the Halpern iteration can find optimal policies of such games. To do so, we study ``$\varepsilon$-approximate'' biases of constant value games, which are related to $\varepsilon$-optimal strategies. There are two ways in which one can interpret what an approximate bias is. One way to define this concept would be to say that a vector is an approximate bias if it is close to some bias. This is a strong notion, and there are no efficient algorithms that find approximations of a bias in this sense. A weaker notion is to say that an approximate bias is a vector that approximately satisfies the ergodicity equation. This is the notion that we are using in this paper. We measure this approximation using the Hilbert seminorm.

\begin{definition}[$\varepsilon$-approximate bias]
  Given $\varepsilon \ge 0$, we say that a vector $\bias \in \R^n$ is an \emph{$\varepsilon$-approximate bias} if it satisfies $\hilb{\bias - \shapley(\bias)} \le \varepsilon$.
\end{definition}

We note that a game with non-constant value may have an $\varepsilon$-approximate bias. Nevertheless, in this case the value is almost constant, as explained by the following simple but important proposition.

\begin{proposition}\label{bias_approx_value}
Take a mean payoff game with weights $\payoff \in \R^m$ and suppose that $\bias \in \R^n$ is an $\varepsilon$-approximate bias for some $\varepsilon \ge 0$. Then, (i) $\bias$ is a bias of a constant value game with some weights $\tilde{\payoff} \in \R^m$ that satisfy $\supnorm{\tilde{\payoff} - \payoff} \le \varepsilon$. Moreover, (ii) we can choose $\tilde{\payoff}$ in such a way that $\bias$ induces the same policies in the games with weights $\payoff$ and $\tilde{\payoff}$.

In particular, (iii) the value vector $\gameval \in \R^n$ of the game with weights $\payoff$ satisfies $\hilb{\gameval} \le 2\varepsilon$, and (iv) the policies induced by $\bias$ in this game are $(2\varepsilon)$-optimal.
\end{proposition}
\begin{proof}
  (i)
Let $\bias \in \R^n$ be an $\varepsilon$-approximate bias of a mean payoff game with weights $\payoff$ and $\shapley$ be the Shapley operator of this game. Let $\rho = \shapley(\bias)_1 - \bias_1 \in \R$ and let $\delta_i = \rho - (\shapley(\bias)_i - \bias_i)$ for all $i \in [n]$. By the definition of the approximate bias, we have $|\delta_i| \le \varepsilon$ for all $i$. Define the weights $\tilde{\payoff}$ as $\tilde{\payoff}_{ij} = \payoff_{ij} + \delta_i$ for every edge $(i,j) \in \edges$. In other words, for every vertex $i$ of the graph, we shift the weights of all the outgoing edges of $i$ by the same amount $\delta_i$. We have $\supnorm{\tilde{\payoff} - \payoff} \le \varepsilon$. Moreover, for any $i \in \Maxvertices$ we have
\[
\max_{(i,j) \in \edges}\{\tilde{\payoff}_{ij} + \bias_j\} = \max_{(i,j) \in \edges}\{\payoff_{ij} + \bias_j\} + \delta_i = \shapley(\bias)_i + \delta_i = \rho + \bias_i \, .
\]
Analogously, for any $i \in \Minvertices$ we have $\min_{(i,j) \in \edges}\{\tilde{\payoff}_{ij} + \bias_j\} = \rho + \bias_i$. Hence, if $\tilde{\shapley}$ is the Shapley operator of the game with weights $\tilde{\payoff}$, then $(\rho, \bias)$ solve the ergodic equation for $\tilde{\shapley}$. In particular, $\bias$ is a bias of this game, and $\rho$ is its (constant) value, as given by \cref{ergodic_equation}. (ii) Moreover, any vector $\bias$ induces the same policies in the game with weights $\payoff$ as in the game with weights $\tilde{\payoff}$ (cf. the definition of induced policies, \cref{def:induced_policy}). This is because we shifted the weights of all the outgoing edges of a given vertex by the same amount. To show (iii), note that \cref{value_lipschitz} gives the inequality $\supnorm{\gameval - \rho \ones} \le \supnorm{r - \tilde r} \le \varepsilon$ and so $\abs{\gameval_i - \gameval_j} = \abs{(\gameval_i - \rho) - (\gameval_j - \rho)} \le 2\varepsilon$ for all $i,j$. Therefore $\hilb{\gameval} \le 2\varepsilon$. Moreover, any policies induced by $\bias$ are optimal in the game with weights $\tilde{\payoff}$ by \cref{optimality_eq}. Hence (iv) follows from \cref{perturbed_value}.
\end{proof}

We now present two algorithms that can find an $\varepsilon$-approximate bias of a constant value game. Our algorithms are based on the \KM and Halpern iterations. To give some context, we briefly explain the motivation behind these methods, referring the reader to \cite{CominettiSotoVaisman:2014,ContrerasCominetti:2023} and the references therein for more information.

If a function $f \colon \R^n \to \R^n$ is a contraction, then the Banach fixed point theorem implies that $f$ has a unique fixed point. Moreover, for any $x^0$, the sequence $(x^k)_k$ defined by $x^k = f(x^{k-1})$ for $k \ge 1$ converges to the unique fixed point of $f$, at rate which depends on the contraction rate of $f$. If we only suppose that $f$ is non-expansive, i.e., $\| f(x) - f(y)\| \le \| x - y\|$ for all $x,y \in \R^n$, then none of these statements is true. In particular, $f$ may have zero fixed points, one fixed point, or many fixed points. As an example, a translation in $\R^2$ has no fixed points, a rotation has a single fixed point, and a reflection across a line has infinitely many fixed points. Even if $f$ has at least one fixed point, the iteration method given by $x^k = f(x^{k-1})$ may not converge to any fixed point. Even worse, the error $\|x^k - f(x^k)\|$ may not converge to zero (as shown by a rotation in $\R^2$). The \KM and Halpern iteration methods overcome some of these issues. More precisely, if $f$ is a non-expansive function that has at least one fixed point, then these methods provide a sequence $x^k$ such that the error $\|x^k - f(x^k)\|$ does converge to zero. In general, the \KM iteration is given by a sequence $x^k$ that satisfies $x^k = (1 - \alpha_k) x^{k-1} + \alpha_k f(x^{k-1})$ for some parameters $\alpha_k \in \interval[open]{0}{1}$, while the Halpern iteration is given by a sequence that satisfies $x^k = (1 - \alpha_k) y^0 + \alpha_k f(x^{k-1})$, where the point $y^0$ is an ``anchor'' (and is often taken to be equal to $x^0$) and the parameters $\alpha_k$ are in $\interval[open]{0}{1}$. For an appropriate choice of the parameters, the error $\|x^k - f(x^k)\|$ is bounded by $O(1/\sqrt{k})$ for the \KM iteration and by $O(1/k)$ for the Halpern iteration. For the sake of simplicity, in this paper we only consider the canonical choices of the parameters. Hence, for us, the \emph{\KM iteration} is given by a sequence that satisfies
\[
x^k = \frac{1}{2} x^{k-1} + \frac{1}{2} f(x^{k-1})
\]
for $k \ge 1$, while the \emph{Halpern iteration} is given by a sequence that satisfies
\[
x^k = \frac{2}{k+2} x^{0} + \frac{k}{k+2} f(x^{k-1})
\]
for $k \ge 1$. We rely on the following estimates of the rate of convergence of these methods.

\begin{theorem}\label{halpern}
  Let $(\nspace, \| \cdot \|)$ be a normed space and let $f \colon \nspace \to \nspace$ be a non-expansive function, i.e., a function that satisfies $\| f(x) - f(y) \| \le \| x - y\|$ for all $x,y \in \nspace$. Furthermore, suppose that this function has at least one fixed point $x^{*} \in \nspace$. Then, we have the following convergence results:
  \begin{enumerate}
    \item \cite{BaillonBruck:1996,CominettiSotoVaisman:2014} Consider the \KM sequence $(x^k)_{k \ge 0}$ given by $x^{k} = \frac{1}{2}x^{k-1} + \frac{1}{2}f(x^{k-1})$ for $k \ge 1$. Then, for every $k \ge 1$ we have
    \[
      \| x^k - f(x^k)\| \le \frac{4}{\sqrt{\pi k}}\|x^{0} - x^*\| \le \frac{3}{\sqrt{k}}\|x^{0} - x^*\| \, .
    \]

    \item \cite{SabachShtern:2017,ContrerasCominetti:2023} Consider the Halpern sequence $(x^k)_{k \ge 0}$ given by $x^{k} = \frac{2}{k+2}x^{0} + \frac{k}{k+2}f(x^{k-1})$ for $k \ge 1$. Then, for every $k \ge 0$ we have
    \[
      \| x^k - f(x^k)\| \le \frac{8}{k+1}\|x^{0} - x^*\| \, .
    \]
  \end{enumerate}
\end{theorem}

If the value of a constant value game is known, then the ergodic equation \cref{eq:ergodic} shows that bias is a fixed point of a non-expansive operator. Hence, by applying \cref{halpern} directly we would get an approximate bias. However, implementing this method requires to know the value of the game. Nevertheless, the two iteration schemes still provide an approximate bias even if we do not know the value of the game. This is because the Shapley operator is additively homogeneous, so that we can interpret the ergodic equation as a fixed point equation in an  ``additively projective'' space. This interpretation has appeared in related contexts in \cite{GaubertStott:2020,AkianGaubertNaepelsTerver:2023,LeeBravoCominetti:2025} and we do the same here.

\begin{theorem}\label{halpern_for_bias}
Let $\shapley$ be a Shapley operator of a mean payoff game with constant value and let $\bias^*$ be any bias of this game. Then, we have the following convergence results:
\begin{enumerate}
\item Consider the \KM sequence $(\bias^k)_{k \ge 0}$ given by $u^k = \frac{1}{2}u^{k-1} + \frac{1}{2}\shapley(u^{k-1})$ for $k \ge 1$. Then, for every $k \ge 1$ we have
\[
\hilb{u^k - \shapley(u^k)} \le \frac{3}{\sqrt{k}}\hilb{u^{0} - u^*} \, .
\]
\item Consider the Halpern sequence $(\bias^k)_{k \ge 0}$ given by $u^k = \frac{2}{k+2}u^{0} + \frac{k}{k+2}\shapley(u^{k-1})$ for $k \ge 1$. Then, for every $k \ge 0$ we have
\[
\hilb{u^k - \shapley(u^k)} \le \frac{8}{k+1}\hilb{u^{0} - u^*} \, .
\]
\end{enumerate}
In both cases, for any sufficiently large $k$, all policies induced by $\bias^k$ are optimal.
\end{theorem}
\begin{proof}
The proof is the same as the proofs of analogous results in \cite{GaubertStott:2020,AkianGaubertNaepelsTerver:2023,LeeBravoCominetti:2025}. We give the details for the sake of completeness. Let $\nspace$ be the space $\nspace = \R^n/\alpha \ones$, i.e., the space of equivalence classes of vectors in $\R^n$ in which two vectors are equivalent if their difference is of the form $\alpha \ones$. In symbols, $x \sim y \iff \exists \alpha \in \R, x - y = \alpha \ones$. For any $x \in \R^n$ we denote by $[x] \in \nspace$ the equivalence class to which $x$ belongs. We define the vector addition in $\nspace$ as $[x] + [y] = [x+y]$ for $x, y \in \R^n$ and the scalar multiplication as $\beta [x] = [\beta x]$ for $x \in \R^n, \beta \in \R$. One can check that these operations are well-defined and that they turn $\nspace$ into a vector space. Moreover, we define a norm on $\nspace$ as $\| [x] \| = \hilb{x}$. One can check that this function is well-defined. Moreover, \cref{hilbert} shows that this function is indeed a norm on $\nspace$. Consider the function $f \colon \nspace \to \nspace$ defined as $f([x]) = [\shapley(x)]$. Since the Shapley operator $\shapley$ is additively homogeneous, the function $f$ is well-defined. Furthermore, \cref{nonexp_hilbert} shows that the function $f$ is non-expansive as
\begin{align*}
\| f([x]) - f([y]) \| &= \| [\shapley(x)] - [\shapley(y)] \| = \| [\shapley(x)- \shapley(y)] \| = \hilb{\shapley(x) - \shapley(y)} \\
&\le \hilb{x - y} = \| [x - y]\| = \|[x] - [y] \| \, .
\end{align*}
Denote the value of the game by $\gameval \ones$. Since $\bias^*$ is a bias of the game, $[\bias^*]$ is a fixed point of $f$ as $f([\bias^*]) = [\shapley(\bias^*)] = [\gameval \ones + \bias^*] = [\bias^*]$. Consider the \KM sequence $u^{k} = \frac{1}{2}u^{k-1} + \frac{1}{2}\shapley(u^{k-1})$. Then, for all $k \ge 1$ we have
\[
[u^k] = [\frac{1}{2}u^{k-1} + \frac{1}{2}\shapley(u^{k-1})] = \frac{1}{2}[u^{k-1}] + \frac{1}{2}[\shapley(u^{k-1})] = \frac{1}{2}[u^{k-1}] + \frac{1}{2}f([u^{k-1}]) \, .
\]
Hence, we can apply \cref{halpern} to the sequence $[u^k]$. As a result, for any $k \ge 1$ we get $\| [\bias^k] - f([\bias^k]) \| \le \frac{3}{\sqrt{k}} \|[u^{0}] - [\bias^*]\|$. Since $\| [\bias^k] - f([\bias^k]) \| = \| [\bias^k - \shapley(\bias^k)]\| = \hilb{\bias^k - \shapley(\bias^k)}$ and $\|[u^{0}] - [\bias^*]\| = \hilb{u^{0} - \bias^*}$, we get claim (1). The proof of (2) for the Halpern sequence follows in the same way.

To prove optimality of the sequences for large $k$, note that for every $\varepsilon > 0$ there exists $k_{\varepsilon}$ such that the inequality $\hilb{\bias^k - \shapley(\bias^k)} \le \varepsilon$ holds for all $k \ge k_{\varepsilon}$. Hence, by \cref{bias_approx_value}, any policy induced by $\bias^k$ is $(2\varepsilon)$-optimal. Moreover, the set of policies is finite. Hence, there exists $\varepsilon > 0$ such that any $(2\varepsilon)$-optimal policy is in fact optimal. Thus, for a sufficiently large $k$ any policy induced by $\bias^k$ is optimal.
\end{proof}

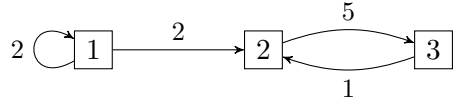
\begin{figure}[t]
\centering
\begin{tikzpicture}[scale=0.75,>=stealth',max/.style={draw,rectangle,minimum size=0.5cm},min/.style={draw,circle,minimum size=0.5cm},av/.style={draw, circle,fill, inner sep = 0pt,minimum size = 0.2cm}]

\node[max] (max1) at (-3, 0) {$1$};
\node[max] (max2) at (0, 0) {$2$};

\node[max] (max3) at (3,0) {$3$};

\draw[->] (max1) to[out=210,in=150,looseness=7] node[left, font=\small]{$2$} (max1);
\draw[->] (max1) to node[above, font=\small]{$2$} (max2);
\draw[->] (max2) to[out=20,in=160] node[above, font=\small]{$5$} (max3);
\draw[->] (max3) to[out=200,in=-20] node[below, font=\small]{$1$} (max2);

\end{tikzpicture}
\caption{An example of a one-player game for which the Zwick--Paterson value iteration does not find the optimal policy.} \label{fig:cycling_example}
\end{figure}

\begin{example}
Consider the game presented in \cref{fig:cycling_example}, which is taken from \cite[Example~4]{DellaVecchiaDiMarco:2012}. This is a one-player game for player Max. The player has only two policies: they can use the edge $(1,2)$ or the edge $(1,1)$. The policy that uses $(1,1)$ gives an average payoff equal to $2$ at state $1$, while the policy that uses $(1,2)$ gives an average payoff equal to $3$ at state $1$. Therefore, using $(1,2)$ is the only optimal policy. Moreover, the game has constant value equal to $3$ at every initial state. The Shapley operator of this game is given by $\shapley \colon \R^3 \to \R^3$,
\[
\shapley(x_1, x_2, x_3) = (
\max\{2 + x_1, 2 + x_2\}, 5 + x_3, 1 + x_2) \, .
\]
Let $u^0 = 0$ and apply the Zwick--Paterson value iteration, $u^k = \shapley(u^{k-1})$. Then, $u^1 = (2,5,1)$ and for $k \ge 2$ by induction we get
\[
u^k = \begin{cases}
(3k+1, 3k, 3k) &\text{if $k$ is even,} \\
(3k,3k+2,3k-2) &\text{if $k$ is odd.}
\end{cases}
\]
Note that $\frac{u^k}{k}$ converges to the value of the game $\lambda = (3,3,3)$, as it is supposed to.
Furthermore, the vector $u^k$ induces the optimal policy for odd values of $k$, but it does \emph{not} induce the optimal policy for even values of $k$. This shows that the naive application of the Zwick--Paterson value iteration fails at finding optimal policies. Furthermore, we have
\[
\shapley(u^k) - u^k = \begin{cases}
(2,5,1) &\text{if $k$ is even,}\\
(4,1, 5) &\text{if $k$ is odd.}
\end{cases}
\]
Hence, the sequence $\hilb{\shapley(u^k) - u^k}$ does not converge.

Let us compare this to the \KM iteration, $u^0 = 0$ and $u^k = \frac{1}{2}u^{k-1} + \frac{1}{2}\shapley(u^{k-1})$. Then, we get $u^1 = (1,\frac{5}{2},\frac{1}{2})$. Moreover, one can check by induction that we have
\[
u^k = (\frac{3k}{2} - 2^{-k},\frac{3k+2}{2},\frac{3k-2}{2})
\]
for all $k \ge 1$. Hence, $u^k$ induces the optimal policy for all $k \ge 1$. Furthermore, $\shapley(u^k) - u^k = (3 + 2^{-k},3,3)$, so $\hilb{\shapley(u^k) - u^k}$ goes to zero, as claimed in \cref{halpern_for_bias}.
\end{example}

\begin{figure}[t]
\centering
\begin{tikzpicture}[scale=0.75,>=stealth',max/.style={draw,rectangle,minimum size=0.5cm},min/.style={draw,circle,minimum size=0.5cm},av/.style={draw, circle,fill, inner sep = 0pt,minimum size = 0.2cm}]

\node[max] (max1) at (-3, 0) {$1$};
\node[max] (max2) at (0, 0) {$2$};

\draw[->] (max1) to[out=210,in=150,looseness=7] node[left, font=\small]{$0$} (max1);
\draw[->] (max1) to node[above, font=\small]{$-W$} (max2);
\draw[->] (max2) to[out=-30,in=30,looseness=7] node[right, font=\small]{$1$} (max2);

\end{tikzpicture}
\caption{An example of a one-player game for which the value iteration algorithms converge slowly.} \label{fig:slow_convergence}
\end{figure}
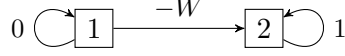

\begin{example}
In the previous example, the \KM iteration converged much faster than the bounds given in \cref{halpern_for_bias}. This is not always the case. To illustrate this, consider the game presented in \cref{fig:slow_convergence}. As before, this is a one-player game for player Max. The parameter $W \ge 1$ is a natural number. As previously, the only optimal policy is the one that uses the edge $(1,2)$ and the game has constant value equal to $1$. The Shapley operator of this game is given by
\[
\shapley(x_1,x_2) = (\max\{x_1, -W + x_2\},1+x_2) \, .
\]
Let $u^0 = 0$ and consider the \KM sequence $u^k = \frac{1}{2}u^{k-1} + \frac{1}{2}\shapley(u^{k-1})$. Then, for all $k \le 2W$ we have $u^k = (0, \frac{k}{2})$ and the non-optimal policy that uses $(1,1)$ is induced by $u^k$. Afterwards, for $k \ge 2W+1$ we have $u^{k} = (\frac{k -2W -2}{2} + 2^{-(k-2W)},\frac{k}{2})$ and the policy that uses $(1,2)$ becomes the only policy induced by $u^k$. Furthermore, we get
\[
\shapley(u^k) - u^k = \begin{cases}
(0,1) &\text{for $k \le 2W$,}\\
(1 - 2^{-(k-2W)}, 1) &\text{for $k \ge 2W+1$.}
\end{cases}
\]
Therefore, the sequence $\hilb{\shapley(u^k) - u^k}$ is constant and equal to $1$ for the first $2W$ iterations, and starts going to $0$ only afterwards. Note that this is consistent with \cref{halpern_for_bias}. Indeed, the bias of this game is unique (up to adding a constant to every coordinate) and equal to $\bias^* = (-W-1,0)$. Hence $\hilb{\bias^*} = W + 1$ and \cref{halpern_for_bias} gives $\hilb{\shapley(u^k) - u^k} \le \frac{3(W+1)}{\sqrt{k}}$, which is true for all $k$.

Likewise, consider the Halpern iteration $u^0 = 0$ and $u^k = \frac{k}{k+2}\shapley(u^{k-1})$. Then, for $k \le 3W+1$ we have $u^k = (0,\frac{k}{3})$, and for $k \ge 3W+1$ we have $u^k = (\frac{k(k-3W-1)}{3(k+2)},\frac{k}{3})$. Therefore,
\[
\shapley(u^k) - u^k = \begin{cases}
(0,1) &\text{for $k \le 3W$,}\\
(1 - \frac{2W+2}{k+2}, 1) &\text{for $k \ge 3W+1$.}
\end{cases}
\]
Hence, the sequence $\hilb{\shapley(u^k) - u^k}$ is constant and equal to $1$ for the first $3W$ iterations and starts going to zero only afterwards. This is consistent with \cref{halpern_for_bias} because this theorem gives the bound $\hilb{\shapley(u^k) - u^k} \le \frac{8(W+1)}{k+1}$.
\end{example}

\section{Condition number for games with a single equilibrium cycle}\label{sec:one_cycle_condition}

In the previous section, we showed that the \KM iteration and the Halpern iteration are capable of solving constant value games. In this section, we extend this result by showing that these iteration schemes converge quickly to the optimal policies provided that a constant value game is ``well-conditioned''. Since for a generic choice of weights all cycles of the game arena have different means, a ``generic'' constant value game has in fact only one equilibrium cycle. A condition number for such games was proposed in \cite{LoffSkomra:2024}, along with an algorithm that converges to the optimal solution in polynomial time in this condition number. In this section, we show that \KM iteration and Halpern iteration have the same property as the algorithm from \cite{LoffSkomra:2024}, namely that the same condition number governs the behavior of both iteration schemes and of the algorithm from \cite{LoffSkomra:2024}. We note that the paper \cite{LoffSkomra:2024} only considers games played on ergodic graphs, but most of the analysis from \cite{LoffSkomra:2024} is applicable to any mean payoff game that happens to have exactly one equilibrium cycle. In particular, the results from \cite{LoffSkomra:2024} that we cite below do not require the assumption that the underlying graph is ergodic.

\begin{lemma}[{\cite[Lemma~4.18]{LoffSkomra:2024}}]\label{one_cycle_graph}
Suppose that every vertex of $\dgraph$ has exactly one outgoing edge and that $\dgraph$ has only one directed cycle. Then, the bias is unique up to adding a constant to every coordinate. More precisely, if $(\gameval,\bias)$ is a solution of the ergodic equation and $p \in [n]$ is a vertex on the directed cycle, then for every $l \in [n]$ the difference $\bias_l - \bias_p$ is the distance from $l$ to $p$ in the digraph $\dgraph$ with weights $\payoff_{ij} - \gameval$.
\end{lemma}

\begin{definition}[$\Xi$, $\gameval^{\sigma,\tau}$, $\bias^{\sigma,\tau}$]\label{one_cycle_bias_def}
  We denote by $\Xi$ the set of pairs $(\sigma,\tau)$ of policies such that the induced subgraph $\dgraph^{\sigma,\tau}$ has only one cycle. Given a pair $(\sigma, \tau) \in \Xi$, we denote by $\gameval^{\sigma,\tau} \colon \R^{m} \to \R$ the function that takes as an argument the vector of weights and outputs the value of the MPG played on $\dgraph^{\sigma,\tau}$; equivalently, this is the mean weight of the unique directed cycle in $\dgraph^{\sigma,\tau}$. Analogously, we denote by $\bias^{\sigma,\tau} \colon \R^{m} \to \R^{n}$ the function that takes as an argument the weight vector and outputs the bias of the MPG played on $\dgraph^{\sigma,\tau}$ that satisfies $(\bias^{\sigma,\tau}(\payoff))_{i_0} = 0$, where $i_0$ is the vertex with the minimal index on the cycle of $\dgraph^{\sigma,\tau}$.
\end{definition}

\begin{lemma}[{\cite[Lemma~4.20]{LoffSkomra:2024}}]\label{linear_functions}
For every pair $(\sigma,\tau) \in \Xi$, the functions $\gameval^{\sigma,\tau}, \bias^{\sigma,\tau}$ are linear. Furthermore, for all $\payoff, \tilde{\payoff} \in \R^{m}$ we have $\supnorm{\gameval^{\sigma,\tau}(\payoff) - \gameval^{\sigma,\tau}(\tilde{\payoff})} \le \supnorm{\payoff - \tilde{\payoff}}$ and $\supnorm{\bias^{\sigma,\tau}(\payoff) - \bias^{\sigma,\tau}(\tilde{\payoff})} \le 2n\supnorm{\payoff - \tilde{\payoff}}$.
\end{lemma}

\begin{definition}[$\polyh^{\sigma,\tau}$]\label{def:polyhedral_cells}
  Given a pair of policies $(\sigma,\tau) \in \Xi$ we denote by $\polyh^{\sigma,\tau} \subseteq \R^m$ the open polyhedral cone given by the linear inequalities
  \begin{align*}
  &\forall i \in \Maxvertices, (i,j) \in \edges, j \neq \sigma(i), &\gameval^{\sigma,\tau}(\payoff) + \bias^{\sigma,\tau}_i(\payoff) > \payoff_{ij} + \bias^{\sigma,\tau}_j(\payoff) \, , \\
  &\forall i \in \Minvertices, (i,j) \in \edges, j \neq \tau(i), &\gameval^{\sigma,\tau}(\payoff) + \bias^{\sigma,\tau}_i(\payoff) < \payoff_{ij} + \bias^{\sigma,\tau}_j(\payoff) \, .
  \end{align*}
\end{definition}

\begin{lemma}\label{unique_bias}
If $\payoff \in \polyh^{\sigma,\tau}$ for some $(\sigma,\tau) \in \Xi$, then the game has constant value equal to $\gameval^{\sigma,\tau}(\payoff)$. Furthermore, the ergodic equation has a single solution (up to adding a constant to the bias), and $(\sigma,\tau)$ is the only pair of bias-induced policies. More precisely, the bias is given by $\bias^{\sigma,\tau}(\payoff)$.
\end{lemma}

\begin{proof}
The proof is given in \cite[Theorem~4.24]{LoffSkomra:2024}, among other claims. Some of the other claims require the graph to be ergodic, but this one does not. We give the relevant part of the proof from \cite{LoffSkomra:2024}.

First, note that if $\payoff \in \polyh^{\sigma,\tau}$ then the pair $\bigl(\gameval^{\sigma,\tau}(\payoff), \bias^{\sigma,\tau}(\payoff)\bigr)$ solves the ergodic equation. Indeed, in the graph $\dgraph^{\sigma,\tau}$, the definitions of $\gameval^{\sigma,\tau}$ and $\bias^{\sigma,\tau}$ from \cref{one_cycle_bias_def} give $\gameval^{\sigma,\tau}(\payoff)+\bias_i^{\sigma,\tau}(\payoff) = \payoff_{i\sigma(i)}+\bias_{\sigma(i)}^{\sigma,\tau}(\payoff)$ for $i \in \Maxvertices$, and the analogous equality with $\tau(i)$ for $i \in \Minvertices$. The inequalities defining $\polyh^{\sigma,\tau}$ say that all other outgoing edges are strictly worse for the corresponding player. Hence, the Shapley operator satisfies $\shapley(\bias^{\sigma,\tau}(\payoff)) = \gameval^{\sigma,\tau}(\payoff)\ones + \bias^{\sigma,\tau}(\payoff)$. Thus, by \cref{ergodic_equation}, $\gameval^{\sigma,\tau}(\payoff)$ is the value of the game, and $\bias^{\sigma,\tau}(\payoff)$ is a bias. Furthermore, note that $\bias^{\sigma,\tau}(\payoff)$ induces exactly one pair of policies, namely $(\sigma,\tau)$. To finish the proof, it is enough to show that $\bias^{\sigma,\tau}(\payoff)$ is the unique bias of the game (up to adding a constant to all of its coordinates). To do so, denote $(\gameval, \bias) \coloneqq \bigl(\gameval^{\sigma,\tau}(\payoff), \bias^{\sigma,\tau}(\payoff)\bigr)$.  Let $\varepsilon > 0$ be such that
\begin{align*}
&\forall (i,j) \in \edges, i \in \Maxvertices, j \neq \sigma(i), \gameval + \bias_i - \bias_j > \payoff_{ij} + \varepsilon \, , \\
&\forall (i,j) \in \edges, i \in \Minvertices, j \neq \tau(i), \gameval + \bias_i - \bias_j < \payoff_{ij} - \varepsilon \, .
\end{align*}
Let $i_0$ be the vertex with the minimal index on the cycle of $\dgraph^{\sigma,\tau}$, and recall that $(\bias^{\sigma,\tau}(r))_{i_0} = 0$. Then consider the set $B \subseteq \R^n$ defined as
\[
B \coloneqq \{z \in \R^{n} \colon \forall i, \abs{z_i - z_{i_0} - \bias_i} \le \frac{\varepsilon}{4} \} \, .
\]
Note that $B$ contains the line $\{\bias + \alpha \ones \colon \alpha \in \R\}$. Suppose that some bias $v$ of the game is contained in $B$. Then, for every $i \in \Maxvertices, j \neq \sigma(i)$ we have $\lambda + v_i - v_j = \lambda + u_i - u_j + (v_i - v_{i_0} -u_i) - (v_j - v_{i_0} - u_j) > \payoff_{ij}$. Analogously, for every $i \in \Minvertices, j \neq \tau(i)$ we have $\lambda + v_i - v_j  < \payoff_{ij}$. Hence, $v$ is a bias of the game on $\dgraph^{\sigma,\tau}$ and \cref{one_cycle_graph} implies that $v$ is on the line $\{\bias + \alpha \ones \colon \alpha \in \R\}$, which is contained in $\mathsf{int}(B) = \{z \in \R^{n} \colon \forall i, \abs{z_i - z_{i_0} - u_i} < \frac{\varepsilon}{4} \}$. Therefore, the set of biases of the MPG with weights $\payoff$ is included in the union of two disjoint open sets
\[
(\R^{n} \setminus B) \cup \mathsf{int}(B)  \, .
\]
Since the set of biases is connected by \cref{biases_connected}, it can have a nonempty intersection with only one of these sets. Since $u$ itself is in $B$, all biases are therefore included in $B$ and so they all lie on the line $\{\bias + \alpha \ones \colon \alpha \in \R\}$. Therefore, the bias is unique up to adding a constant to all of its coordinates.
\end{proof}

We can now define the condition number of a game with weights that belong to some cone $\polyh^{\sigma, \tau}$. This definition is slightly different to the definition that appeared in \cite{LoffSkomra:2024}. Nevertheless, the two definitions are equivalent, as we discuss after giving the definition.

\begin{definition}[Slack $\sk(\payoff)$ and condition number $\Cond(\payoff)$]\label{def:condition-number}
Suppose that $\payoff \in \polyh^{\sigma,\tau}$ for some $(\sigma,\tau) \in \Xi$. Then, the \emph{slack} of the vector $\payoff$, denoted $\sk(\payoff)$, is the largest positive number $\delta$ that satisfies the inequalities
  \begin{align*}
  &\forall i \in \Maxvertices, (i,j) \in \edges, j \neq \sigma(i), &\gameval^{\sigma,\tau}(\payoff) + \bias^{\sigma,\tau}_i(\payoff) \ge \delta + \payoff_{ij} + \bias^{\sigma,\tau}_j(\payoff) \, , \\
  &\forall i \in \Minvertices, (i,j) \in \edges, j \neq \tau(i), &\gameval^{\sigma,\tau}(\payoff) + \bias^{\sigma,\tau}_i(\payoff) \le -\delta + \payoff_{ij} + \bias^{\sigma,\tau}_j(\payoff) \, .
  \end{align*}
  In other words, $\sk(\payoff)$ is equal to
  \[
  \min\{\abs{\payoff_{ij} - \gameval^{\sigma,\tau}(\payoff) + \bias_j^{\sigma,\tau}(\payoff) - \bias_i^{\sigma,\tau}(\payoff)} \colon (i,j) \in \edges, (i \in \Minvertices \wedge  j \neq \tau(i) ) \vee (i \in \Maxvertices \wedge j \neq \sigma(i) ) \} \, .
  \]
We define a \emph{condition number} of the MPG with weights $\payoff$ as $\Cond(\payoff) = 0$ if every vertex of $\dgraph$ has out-degree $1$ and, otherwise, as
\[
\Cond(\payoff) \coloneqq \frac{\hilb{\payoff}}{\sk(\payoff)} \, .
\]
\end{definition}

\begin{remark}
In \cite{LoffSkomra:2024}, a similar definition of a condition number is studied. Using the notation above, the condition number defined in \cite{LoffSkomra:2024} is equal to $\frac{1}{\sk(\payoff)}{\max\{\abs{\payoff_{ij} - \gameval^{\sigma,\tau}(\payoff)} \colon (i,j) \in \edges\}}$. This quantity differs from $\Cond(\payoff)$ by at most a factor of $2$. Indeed, since $\gameval^{\sigma,\tau}(\payoff)$ is in the interval $\interval{\min_{ij} \payoff_{ij}}{\max_{ij} \payoff_{ij}}$, we have $\abs{\payoff_{ij} - \gameval^{\sigma,\tau}(\payoff)} \le \hilb{\payoff}$ for all $i,j$. On the other hand, the triangle inequality gives $\abs{\payoff_{ij} - \payoff_{i'j'}} \le \abs{\payoff_{ij} - \gameval^{\sigma,\tau}(\payoff)} + \abs{\payoff_{i'j'} - \gameval^{\sigma,\tau}(\payoff)}$ and hence
\[
\hilb{\payoff} \le 2\max\{\abs{\payoff_{ij} - \gameval^{\sigma,\tau}(\payoff)} \colon (i,j) \in \edges\} \, .
\]
As a result, the two condition numbers differ by a factor of at most $2$. Hence, we can use either definition to study the conditioning of randomized games.
\end{remark}

In the next lemma we prove a property of the slack. Namely, if a weight vector in $\polyh^{\sigma, \tau}$ has large slack, then any game with slightly perturbed weights still belongs to $\polyh^{\sigma, \tau}$.

\begin{lemma}\label{perturbation_in_polyh}
Suppose that $\payoff \in \polyh^{\sigma, \tau}$ for some $(\sigma, \tau) \in \Xi$. Let $\tilde{\payoff} \in \R^m$ be such that $\supnorm{\payoff - \tilde{\payoff}} \le \frac{1}{4n+3}\sk(\payoff)$. Then, $\tilde{\payoff} \in \polyh^{\sigma, \tau}$.
\end{lemma}
\begin{proof}
Let $\varepsilon \coloneqq \frac{1}{4n+3}\sk(\payoff)$. By \cref{linear_functions} we get that $\supnorm{\gameval^{\sigma,\tau}(\payoff) - \gameval^{\sigma,\tau}(\tilde{\payoff})} \le \varepsilon$ and $\supnorm{\bias^{\sigma,\tau}(\payoff) - \bias^{\sigma,\tau}(\tilde{\payoff})} \le 2n\varepsilon$. Hence, for every $(i,j) \in \edges, i \in \Maxvertices, j \neq \sigma(i)$ we have
\begin{align*}
  \gameval^{\sigma, \tau}(\tilde{\payoff}) + \bias^{\sigma,\tau}(\tilde{\payoff})_i - \bias^{\sigma,\tau}(\tilde{\payoff})_j
  & \ge \gameval^{\sigma, \tau}(\payoff) + \bias^{\sigma,\tau}(\payoff)_i - \bias^{\sigma,\tau}(\payoff)_j - (4n+1)\varepsilon \\
  & \ge r_{ij} + \sk(\payoff)  - (4n+1)\varepsilon \\
  & \ge r_{ij} + (4n+3)\varepsilon  - (4n+1)\varepsilon\\
  & = \payoff_{ij} + 2\varepsilon> \tilde{\payoff}_{ij} \, .
\end{align*}
Analogously, for every $(i,j) \in \edges, i \in \Minvertices, j \neq \tau(i)$ we have $\gameval^{\sigma, \tau}(\tilde{\payoff}) + \bias^{\sigma,\tau}(\tilde{\payoff})_i - \bias^{\sigma,\tau}(\tilde{\payoff})_j < \tilde{\payoff}_{ij}$. Hence, $\tilde{\payoff} \in \polyh^{\sigma, \tau}$.
\end{proof}

As a corollary, we show that any $\varepsilon$-approximate bias of a game with weights in $\polyh^{\sigma, \tau}$ induces the policies $(\sigma,\tau)$.

\begin{corollary}\label{approx_bias_policies}
Suppose that $\payoff \in \polyh^{\sigma, \tau}$ for some $(\sigma, \tau) \in \Xi$. Let $\varepsilon \coloneqq \frac{1}{4n+3}\sk(\payoff)$ and suppose that $\tilde{\bias}$ is an $\varepsilon$-approximate bias of the game. Then, $\tilde{\bias}$ induces a single pair of policies, namely $(\sigma,\tau)$.
\end{corollary}
\begin{proof}
By \cref{bias_approx_value} we see that $\tilde{\bias}$ is a bias of a game with some weights $\tilde{\payoff}$ that satisfy $\supnorm{\payoff - \tilde{\payoff}} \le \varepsilon$. Moreover, $\tilde{\bias}$ induces the same policies in the games with weights $\payoff$ and $\tilde{\payoff}$. By \cref{perturbation_in_polyh} we have $\tilde{\payoff} \in \polyh^{\sigma,\tau}$ and, by applying \cref{unique_bias} to $\tilde{\payoff}$, the vector $\tilde{\bias}$ induces a single pair of policies, namely $(\sigma,\tau)$.
\end{proof}

We now show that the \KM iteration and the Halpern iteration converge quickly to the solution of a game if it is well-conditioned in the sense defined above.

\begin{theorem}\label{km_for_one_cycle}
Suppose that $\payoff \in \polyh^{\sigma,\tau}$ for some $(\sigma,\tau) \in \Xi$. Then, we have the following convergence results:
\begin{enumerate}
\item Consider the \KM iteration sequence defined by $\bias^0 = 0$ and $\bias^k = \frac{1}{2}\bias^{k-1} + \frac{1}{2}\shapley(\bias^{k-1})$. Then, for all $k \ge 36n^2(4n+3)^2\Cond(\payoff)^2$, any policy induced by $\bias^k$ is optimal.
\item Consider the Halpern iteration sequence defined by $\bias^0 = 0$ and $\bias^k = \frac{k}{k+2}\shapley(\bias^{k-1})$. Then, for all $k \ge 16n(4n+3)\Cond(\payoff) - 1$, any policy induced by $\bias^k$ is optimal.
\end{enumerate}
\end{theorem}
\begin{proof}
If every vertex of $\dgraph$ has out-degree $1$, then the claim is trivial. Otherwise, let $(\gameval,\bias)$ be the value and a bias of the game. Denote $\varepsilon \coloneqq \frac{1}{4n+3}\sk(\payoff)$ and suppose that $\tilde{\bias}$ is an $\varepsilon$-approximate bias. Then, by \cref{approx_bias_policies}, the vector $\tilde{\bias}$ induces a single pair of policies, namely $(\sigma,\tau)$. Now, if $(u^{k})_k$ is the \KM sequence initialized at $u^0 = 0$, then \cref{short_bias} and \cref{halpern_for_bias} show that for all $k \ge 0$ we have
\[
\hilb{\bias^k - \shapley(\bias^k)} \le \frac{6n\hilb{\payoff}}{\sqrt{k}} \, .
\]
Hence, $u^k$ is an $\varepsilon$-approximate bias for all
\[
k \ge \left(\frac{6n\hilb{\payoff}}{\varepsilon}\right)^2 = 36n^2(4n+3)^2\Cond(\payoff)^2 \, .
\]
If $(u^{k})_k$ is the Halpern sequence initialized at $u^0 = 0$, then the same argument shows that for all $k \ge 1$ we have
\[
\hilb{\bias^k - \shapley(\bias^k)} \le \frac{16n\hilb{\payoff}}{k+1} \, .
\]
Hence, $u^k$ is an $\varepsilon$-approximate bias for all $k \ge 16n(4n+3)\Cond(\payoff) - 1$.
\end{proof}

\begin{remark}
We note that \cref{km_for_one_cycle} already extends the results of \cite{LoffSkomra:2024}, by showing that value iteration algorithms have smoothed polynomial complexity for games played on ergodic graphs. Indeed, if we suppose that the underlying graph of the game is ergodic, then the analysis presented in \cite{LoffSkomra:2024} shows that the union $\cup_{(\sigma,\tau) \in \Xi} \polyh^{\sigma,\tau}$ covers almost the entire space of all possible weights (up to a set of Lebesgue measure zero). Furthermore, \cite{LoffSkomra:2024} gives a probabilistic analysis of the condition number $\Cond(\payoff)$ for ergodic graphs. The aim of the following sections is to extend the analysis of the condition number to arbitrary graphs.
\end{remark}

\section{Condition number for general games}\label{sec:general_condition}

In this section, we generalize the condition number from games with only one equilibrium cycle (studied in the previous section) to arbitrary games. To this end, we generalize the definition of the sets $\polyh^{\sigma, \tau}$ to arbitrary pair of policies. This is done as follows.

Fix any pair of policies $(\sigma, \tau)$ and let $\dgraph^{\sigma, \tau}$ be the directed graph obtained from $\dgraph$ by fixing $(\sigma, \tau)$. Suppose that this graph has exactly $p$ directed cycles $\mathcal{C}_{1}, \dots, \mathcal{C}_{p}$. For any such cycle $\mathcal{C}_{\ell}$, let $\vertices^{\sigma,\tau}_{\ell}$ denote the set of vertices $i$ such that the path from $i$ goes to the cycle $\mathcal{C}_{\ell}$. Note that the sets $\vertices^{\sigma,\tau}_{\ell}$ are disjoint and that they form a partition of all the vertices of $\dgraph$. We consider the directed graph $\dgraph^{\sigma, \tau}_{\ell}$ defined as
\begin{equation}\label{eq:zero_player_decomposition}
\dgraph^{\sigma, \tau}_{\ell} \coloneqq \bigl(\vertices^{\sigma,\tau}_{\ell}, \{(i,\sigma(i)) \colon i \in \Maxvertices \cap \vertices^{\sigma,\tau}_{\ell} \} \cup  \{(i,\tau(i)) \colon i \in \Minvertices \cap \vertices^{\sigma,\tau}_{\ell} \}  \bigr) \, .
\end{equation}
Note that every vertex in $\dgraph^{\sigma, \tau}_{\ell}$ has exactly one outgoing edge and $\dgraph^{\sigma, \tau}_{\ell}$ has only one directed cycle. This allows us to generalize the definition of the functions $\gameval^{\sigma, \tau}$ and $\bias^{\sigma, \tau}$ from \cref{one_cycle_bias_def}. More precisely, given a vertex $i \in \vertices$ that belongs to the graph $\dgraph^{\sigma, \tau}_{\ell}$ we define $\gameval_i^{\sigma,\tau}(\payoff)$ and $\bias_i^{\sigma,\tau}(\payoff)$ using \cref{one_cycle_bias_def} on the graph $\dgraph^{\sigma, \tau}_{\ell}$. By \cref{linear_functions}, the functions are $\gameval_i^{\sigma,\tau}(\payoff)$ and $\bias_i^{\sigma,\tau}(\payoff)$ are linear. We further denote $\gameval^{\sigma,\tau}(\payoff) \coloneqq \bigl(\gameval_i^{\sigma,\tau}(\payoff)\bigr)_{i \in [n]}$ and $\bias^{\sigma,\tau}(\payoff) = \bigl(\bias_i^{\sigma,\tau}(\payoff)\bigr)_{i \in [n]}$. To define the set $\polyh^{\sigma, \tau}$ we need one more piece of notation: given a vertex $i$ we denote by $\vertices^{\sigma,\tau}(i)$ the set $\vertices^{\sigma,\tau}_{\ell}$ to which $i$ belongs.

\begin{definition}[$\polyh^{\sigma,\tau}$, $\Generic$]\label[definition]{general_cone}
Given any pair of policies $(\sigma, \tau)$ we denote by $\polyh^{\sigma,\tau} \subseteq \R^m$ the open (possibly empty) polyhedral cone formed by the weights $\payoff \in \R^m$ that satisfy the following system of strict linear inequalities:
\begin{align*}
&\forall (i,j) \in \edges, i \in \Maxvertices, \vertices^{\sigma,\tau}(i) \neq \vertices^{\sigma,\tau}(j), &\gameval_{i}^{\sigma,\tau}(\payoff) > \gameval_{j}^{\sigma,\tau}(\payoff), \\
&\forall (i,j) \in \edges, i \in \Minvertices, \vertices^{\sigma,\tau}(i) \neq \vertices^{\sigma,\tau}(j), &\gameval_{i}^{\sigma,\tau}(\payoff) < \gameval_{j}^{\sigma,\tau}(\payoff), \\
&\forall (i,j) \in \edges, i \in \Maxvertices, \vertices^{\sigma,\tau}(i) = \vertices^{\sigma,\tau}(j), \sigma(i) \neq j, &\gameval^{\sigma,\tau}_{i}(\payoff) + \bias^{\sigma,\tau}_i(\payoff) > \payoff_{ij} + \bias^{\sigma,\tau}_j(\payoff), \\
&\forall (i,j) \in \edges, i \in \Minvertices, \vertices^{\sigma,\tau}(i) = \vertices^{\sigma,\tau}(j), \tau(i) \neq j, &\gameval^{\sigma,\tau}_{i}(\payoff) + \bias^{\sigma,\tau}_i(\payoff) < \payoff_{ij} + \bias^{\sigma,\tau}_j(\payoff).
\end{align*}
We further denote $\Generic \coloneqq \cup_{\sigma,\tau} \polyh^{\sigma, \tau}$.
\end{definition}

We describe the properties of the cones $\polyh^{\sigma,\tau}$ in the next few lemmas. We start by showing that $(\sigma,\tau)$ are bias-induced if $\payoff \in \polyh^{\sigma,\tau}$. This relies on a simple observation that we will use multiple times.

\begin{lemma}\label{same_value_zone}
Suppose that $\payoff \in \polyh^{\sigma,\tau}$. If $(i,j) \in \edges$ is a edge of the graph, then we have $\gameval_i^{\sigma,\tau}(\payoff) = \gameval_j^{\sigma,\tau}(\payoff)$ if and only if $\vertices^{\sigma,\tau}(i) = \vertices^{\sigma,\tau}(j)$.
\end{lemma}
\begin{proof}
If $\vertices^{\sigma,\tau}(i) = \vertices^{\sigma,\tau}(j)$, then $\gameval_i^{\sigma,\tau}(\payoff) = \gameval_j^{\sigma,\tau}(\payoff)$ by definition of $\gameval^{\sigma,\tau}$. Conversely, if $\gameval_i^{\sigma,\tau}(\payoff) = \gameval_j^{\sigma,\tau}(\payoff)$, then $\vertices^{\sigma,\tau}(i) = \vertices^{\sigma,\tau}(j)$ by the definition of $\polyh^{\sigma,\tau}$.
\end{proof}

\begin{lemma}\label{cone_gives_bias}
Suppose that $\payoff \in \polyh^{\sigma,\tau}$. Then, the value of the game is given by the vector $\gameval^{\sigma,\tau}(\payoff)$ and the policies $(\sigma,\tau)$ are optimal and bias-induced.
\end{lemma}
\begin{proof}
To simplify the notation, denote $\gameval \coloneqq \gameval^{\sigma,\tau}(\payoff)$ and $\bias \coloneqq \bias^{\sigma,\tau}(\payoff)$. We will show that there exists $t > 0$ such that $t\gameval + \bias$ is a bias of the game. Fix $i \in \Maxvertices$. By definition of $(\gameval,\bias)$, we have $\gameval_i = \gameval_{\sigma(i)}$ and $\gameval_i + \bias_{i} = \payoff_{i\sigma(i)} + \bias_{\sigma(i)}$. Fix $j \in \vertices$ such that $(i,j) \in \edges$ and $j \neq \sigma(i)$. If $\vertices^{\sigma,\tau}(i) \neq \vertices^{\sigma,\tau}(j)$, then $\gameval_i > \gameval_j$ and if $\vertices^{\sigma,\tau}(i) = \vertices^{\sigma,\tau}(j)$, then $\gameval_i = \gameval_j$. Furthermore, if $\vertices^{\sigma,\tau}(i) = \vertices^{\sigma,\tau}(j)$, then $\gameval_i + \bias_{i} > \payoff_{ij} + \bias_{j}$. Hence, by \cref{same_value_zone} we get $\gameval_i = \max_{(i,k) \in \edges} \gameval_k$ and $\gameval_i + \bias_{i} = \max_{(i,k) \in \edges, \gameval_k = \gameval_i} \{\payoff_{ik} + \bias_{k}\}$. Moreover, the maximum in $\max_{(i,k) \in \edges, \gameval_k = \gameval_i} \{\payoff_{ik} + \bias_{k}\}$ is achieved by a single edge, namely $(i, \sigma(i))$. By doing an analogous reasoning for $i \in \Minvertices$, we get that $(\gameval, \bias)$ satisfy the system
\begin{equation}\label{eq:optimality_inproof}
\begin{aligned}
\forall i \in \Maxvertices, \quad &\gameval_{i} = \max_{(i,j) \in \edges} \gameval_{j}, \\
\forall i \in \Minvertices, \quad &\gameval_{i} = \min_{(i,j) \in \edges} \gameval_{j}, \\
\forall i \in \Maxvertices, \quad &\gameval_{i} + \bias_i = \max_{(i,j) \in \edges, \gameval_j = \gameval_i} \{\payoff_{ij} + \bias_{j}\}, \\
\forall i \in \Minvertices, \quad &\gameval_{i} + \bias_i = \min_{(i,j) \in \edges, \gameval_j = \gameval_i} \{\payoff_{ij} + \bias_{j}\}.
\end{aligned}
\end{equation}
Hence, for every sufficiently large $t > 0$, the pair $(\gameval, t\gameval + \bias)$ satisfies the system from \cref{optimality_eq}. Indeed, denote $\tilde{\bias} \coloneqq t\gameval + \bias$ and note that the pair $(\gameval, \tilde{\bias})$ satisfies \cref{eq:optimality_inproof} for all values of $t > 0$. Moreover, if $i \in \Maxvertices$ and $(i,j) \in \edges$ is such that $\gameval_j \neq \gameval_i$, then $\gameval_j < \gameval_i = \gameval_{\sigma(i)}$ by the first set of equalities. Hence, for any sufficiently large $t$ we have $t\gameval_{\sigma(i)} + \payoff_{i\sigma(i)} + \bias_{\sigma(i)} > t\gameval_{j} + \payoff_{ij} + \bias_{j}$ and therefore $\gameval_{i} + \tilde{\bias}_i = \max_{(i,j) \in \edges, \gameval_j = \gameval_i} \{\payoff_{ij} + \tilde{\bias}_{j}\} > \max_{(i,j) \in \edges, \gameval_j \neq \gameval_i} \{\payoff_{ij} + \tilde{\bias}_{j}\}$. Furthermore, the maximum in $\max_{(i,j) \in \edges, \gameval_j = \gameval_i} \{\payoff_{ij} + \tilde{\bias}_{j}\}$ is achieved by a single edge, namely $(i, \sigma(i))$. The reasoning for $i \in \Minvertices$ is analogous. Hence, $\tilde{\bias}$ is a bias of the game and $(\sigma,\tau)$ is the only pair of policies that this bias induces.
\end{proof}

We now show that if $\payoff \in \polyh^{\sigma,\tau}$, then the reduced game on the graph $\dgraph' = (\vertices, \edges')$ with $\edges' \coloneqq \edges \setminus \{(i,j) \in \edges, \gameval_{i} \neq \gameval_{j}\}$ has exactly $p$ connected components with vertex sets $\vertices_1, \ldots, \vertices_p$, and these vertex sets are precisely the sets $\vertices^{\sigma,\tau}_\ell$, $\ell \in [p]$. Moreover, the game restricted to each vertex set $\vertices_\ell$ has exactly one equilibrium cycle, namely $\mathcal C_\ell$.

\begin{lemma}\label{constant_value_subgame}
  Suppose that $\payoff \in \polyh^{\sigma,\tau}$. Then, (i) the reduced game played on $\dgraph'$ has exactly $p$ connected components given by $\dgraph_{\ell} \coloneqq \bigl(\vertices_{\ell}, \edges \cap (\vertices_{\ell} \times \vertices_{\ell})\bigr)$ for $\ell \in [p]$, with $\vertices_\ell = \vertices^{\sigma,\tau}_\ell$. Furthermore, (ii) the constant value mean payoff game played on $\dgraph_{\ell}$ has a unique bias (up to adding a constant to its coordinates) equal to $\bigl(\bias^{\sigma,\tau}(\payoff)\bigr)_{i \in \vertices_\ell}$ and a single pair of bias induced policies, namely the restricted versions of $(\sigma,\tau)$. In particular, this game has a single equilibrium cycle. Moreover, (iii) if we denote $\gameval = \gameval^{\sigma,\tau}(\payoff)$ and $\bias = \bias^{\sigma,\tau}(\payoff)$, then the game on $\dgraph_{\ell}$ has a well-defined condition number $\Cond_{\ell}(\payoff)$ which is equal to $0$ if every vertex of $\dgraph_{\ell}$ has out-degree $1$ and to
  \[
    \Cond_{\ell}(\payoff)  = \frac{\hilb{r^{(\ell)}}}{\sk(r^{(\ell)})}\\
  \]
  where
  \begin{align*}
    \hilb{r^{(\ell)}}& = \max\{\abs{\payoff_{ij} - \payoff_{i'j'}} \colon (i,j),(i',j') \in \edges \cap (\vertices_{\ell} \times \vertices_{\ell})\},\\
    \sk(r^{(\ell)}) & = \min\{\abs{\payoff_{ij} - \gameval_i + \bias_j - \bias_i} \colon (i,j) \in \edges \cap (\vertices_{\ell} \times \vertices_{\ell}), \payoff_{ij} - \gameval_i + \bias_j - \bias_i \neq 0\}.
  \end{align*}
  otherwise.
\end{lemma}
\begin{proof}

  (i) Recall that the reduced game removes all edges between vertices with different values, and the components are then the connected components of the induced undirected graph. By \cref{same_value_zone,cone_gives_bias}, the reduced game has edge set
  \[
    \edges' = \{(i,j) \in \edges \colon \vertices^{\sigma,\tau}(i) = \vertices^{\sigma,\tau}(j)\}.
  \]
  Thus no edge of $\dgraph'$ connects two distinct sets $\vertices^{\sigma,\tau}_\ell$. On the other hand, each set $\vertices^{\sigma,\tau}_\ell$ is connected in $\dgraph'$ after forgetting edge orientations: indeed, every vertex in $\vertices^{\sigma,\tau}_\ell$ reaches the cycle $\mathcal C_\ell$ by following the policies $(\sigma,\tau)$, and every edge on such a path has both endpoints in $\vertices^{\sigma,\tau}_\ell$, hence belongs to $\edges'$. Therefore the connected components of $\dgraph'$ are precisely the sets $\vertices^{\sigma,\tau}_\ell$, $\ell \in [p]$. We denote these components by $\vertices_\ell \coloneqq \vertices^{\sigma,\tau}_\ell$. Since every edge of $\edges$ with both endpoints in $\vertices_\ell$ belongs to $\edges'$, the corresponding component is
  \[
    \dgraph_\ell = \bigl(\vertices_\ell, \edges \cap (\vertices_\ell \times \vertices_\ell)\bigr).
  \]
  This proves claim (i).

  To prove the remaining claims, note that, since $\vertices_\ell = \vertices^{\sigma,\tau}_\ell$, the policies $(\sigma,\tau)$ induce a single cycle on the graph $\dgraph_{\ell}$, namely the cycle $\mathcal{C}_{\ell}$. Furthermore, by the definition of $\polyh^{\sigma,\tau}$, if we fix $(i,j) \in \edges \cap (\vertices_{\ell} \times \vertices_{\ell})$ such that $i \in \Maxvertices$ and $j \neq \sigma(i)$, then $\gameval^{\sigma,\tau}_{i}(\payoff) + \bias^{\sigma,\tau}_i(\payoff) > \payoff_{ij} + \bias^{\sigma,\tau}_j(\payoff)$. Likewise, if $i \in \Minvertices$ and $j \neq \tau(i)$, then $\gameval^{\sigma,\tau}_{i}(\payoff) + \bias^{\sigma,\tau}_i(\payoff) < \payoff_{ij} + \bias^{\sigma,\tau}_j(\payoff)$. Since fixing $(\sigma,\tau)$ induces a single cycle on $\dgraph_{\ell}$, we get that the weights on $\dgraph_{\ell}$ belong to the cone given in \cref{def:polyhedral_cells} for policies $(\sigma,\tau)$ on $\dgraph_{\ell}$. Hence, claims (ii) and (iii) follow from \cref{unique_bias} and the definition of condition number of games with one equilibrium cycle (\cref{def:condition-number}).
\end{proof}

In the last three lemmas, we have shown that if $\payoff \in \polyh^{\sigma,\tau}$, then the game admits a decomposition into games with single equilibrium cycles. We now want to show that if $\payoff$ is generic, then it belongs to exactly one cone $\polyh^{\sigma,\tau}$. This is done in the next proposition.

\begin{proposition}\label{policy_cone}
If $\payoff \in \polyh^{\sigma,\tau}$, then $(\sigma,\tau)$ is the only pair of bias induced policies in the MPG with weights $\payoff$. In particular, the sets $\polyh^{\sigma,\tau}$ are disjoint. Furthermore, the set $\R^m \setminus \Generic$ has Lebesgue measure $0$.
\end{proposition}
\begin{proof}
Let $\gameval \in \R^n$ be the value of the game and $\tilde{\bias} \in \R^n$ be any bias. By \cref{cone_gives_bias}, we have $\gameval = \gameval^{\sigma,\tau}(\payoff)$. Fix $\ell \in [p]$. By the definition of a bias and \cref{same_value_zone} for any $i \in \vertices_{\ell}$ we have
\[
\gameval_{i} + \tilde{\bias}_{i} = \begin{cases}
\max_{(i,j) \in \edges, j \in \vertices_{\ell}}\{\payoff_{ij} + \tilde{\bias}_j\} &\text{if } i \in \Maxvertices, \\
\min_{(i,j) \in \edges, j \in \vertices_{\ell}}\{\payoff_{ij} + \tilde{\bias}_j\} &\text{if } i \in \Minvertices.
\end{cases}
\]
Since the value of $\gameval_i$ is constant within $\vertices_{\ell}$, we get that $(\tilde{\bias}_i)_{i \in \vertices_\ell}$ is a bias of the game on $\dgraph_{\ell}$. Furthermore, if $(\tilde{\sigma},\tilde{\tau})$ are policies induced by $\tilde{\bias}$ on $\dgraph$, then $(\tilde{\bias}_i)_{i \in \vertices_\ell}$ induces the restricted version of $(\tilde{\sigma},\tilde{\tau})$ on $\dgraph_{\ell}$. Hence, by \cref{constant_value_subgame}, the policies $(\tilde{\sigma},\tilde{\tau})$ coincide with $(\sigma,\tau)$ on $\dgraph_{\ell}$. Since this is true for any $\ell$, we get that $(\tilde{\sigma},\tilde{\tau}) = (\sigma,\tau)$. This proves the first claim.

To prove the second claim, for any $(\sigma,\tau)$ let $\mathcal{H}^{\sigma,\tau}$ denote the set of hyperplanes obtained by replacing the inequalities from \Cref{general_cone} with equalities, i.e., the set of hyperplanes of the form $\gameval_{i}^{\sigma,\tau}(\payoff) = \gameval_{j}^{\sigma,\tau}(\payoff)$ for $\vertices^{\sigma,\tau}(i) \neq \vertices^{\sigma,\tau}(j)$ and $\gameval^{\sigma,\tau}_{i}(\payoff) + \bias^{\sigma,\tau}_i(\payoff) = \payoff_{ij} + \bias^{\sigma,\tau}_j(\payoff)$ for $\vertices^{\sigma,\tau}(i) = \vertices^{\sigma,\tau}(j)$ and $j \notin \{\sigma(i),\tau(i)\}$. Note that all of these linear equations are nontrivial, i.e., each of them defines a set of Lebesgue measure $0$. Let $\mathcal{H} \coloneqq \cup_{\sigma,\tau} \mathcal{H}^{\sigma,\tau}$. We will show that $(\R^m \setminus \Generic) \subseteq \mathcal{H}$. To do so, fix $\payoff \in \R^m \setminus \mathcal{H}$. By \cref{optimality_eq}, the game with weights $\payoff$ has at least one pair of bias induced policies. Let $\bias$ be any bias of this game and let $(\sigma,\tau)$ be a pair of policies induced by $\bias$. Since the policies $(\sigma, \tau)$ are optimal for the game with weights $\payoff$, the value of this game is equal to $\gameval^{\sigma,\tau}(\payoff)$. Suppose that the graph $\dgraph^{\sigma,\tau}$ has exactly $p$ directed cycles and let $\dgraph^{\sigma,\tau}_1, \dots, \dgraph^{\sigma,\tau}_p$ be the corresponding graphs with one directed cycle, as given by \cref{eq:zero_player_decomposition}. For any $q \in [p]$, let $\vertices_{q}$ be the set of vertices of the graph $\dgraph^{\sigma,\tau}_q$. By the definition of a bias, if we take any $q \in [p]$ and we restrict $\bias$ to the coordinates from $\vertices_q$, then the restricted vector is a bias of the game on $\dgraph^{\sigma,\tau}_{q}$. In particular, by \cref{one_cycle_graph}, there exists a constant $\alpha_{q} \in \R$ such that for every $i \in \vertices_{q}$ we have $\bias_i = \bias^{\sigma,\tau}_i(\payoff) + \alpha_{q}$. We now show that $\payoff \in \polyh^{\sigma,\tau}$. Since the value of the game with weights $\payoff$ is $\gameval^{\sigma,\tau}(\payoff)$, for every $i \in \Maxvertices$ we have $\gameval_i^{\sigma,\tau}(\payoff) = \max_{(i,j) \in \edges}\gameval_j^{\sigma,\tau}(\payoff)$. In particular, for every $j$ such that $(i,j) \in \edges, \vertices^{\sigma,\tau}(i) \neq \vertices^{\sigma,\tau}(j)$ we get $\gameval_i^{\sigma,\tau}(\payoff) \ge \gameval_j^{\sigma,\tau}(\payoff)$. Since $\payoff \notin \mathcal{H}$, this inequality is strict. Moreover, for every $j \in \vertices^{\sigma,\tau}(i)$ we have $\gameval_i^{\sigma,\tau}(\payoff) = \gameval_j^{\sigma,\tau}(\payoff)$. Hence, we get the equality $\{j \colon (i,j) \in \edges, \gameval_j^{\sigma,\tau}(\payoff) = \gameval_i^{\sigma,\tau}(\payoff) \}
= \{j \colon (i,j) \in \edges, j \in \vertices^{\sigma,\tau}(i)\}$. Hence the bias $\bias$, being a bias for the game on $\dgraph$, satisfies
\[
\gameval^{\sigma,\tau}_{i}(\payoff) + \bias_{i} = \max_{(i,j) \in \edges, j \in \vertices^{\sigma,\tau}(i)}\{\payoff_{ij} + \bias_j \} \,.
\]
If we take $q$ such that $\vertices_{q} = \vertices^{\sigma,\tau}(i)$ and we subtract $\alpha_q$ from both sides, then we get the equality
\[
\gameval^{\sigma,\tau}_{i}(\payoff) + \bias^{\sigma,\tau}_i(\payoff) = \max_{(i,j) \in \edges, j \in \vertices^{\sigma,\tau}(i)}\{\payoff_{ij} + \bias^{\sigma,\tau}_j(\payoff) \} \, .
\]
In particular, for any $(i,j) \in \edges$, $j \in \vertices^{\sigma,\tau}(i)$, $j \neq \sigma(i)$ we have $\gameval^{\sigma,\tau}_{i}(\payoff) + \bias^{\sigma,\tau}_i(\payoff) \ge \payoff_{ij} + \bias^{\sigma,\tau}_j(\payoff)$. Since $\payoff \notin \mathcal{H}$, this inequality is strict. Analogously, if $i \in \Minvertices$, then for every $(i,j) \in \edges,  \vertices^{\sigma,\tau}(i) \neq \vertices^{\sigma,\tau}(j)$ we get the inequality $\gameval_{i}^{\sigma,\tau}(\payoff) < \gameval_{j}^{\sigma,\tau}(\payoff)$ and for every $(i,j) \in \edges,  \vertices^{\sigma,\tau}(i) = \vertices^{\sigma,\tau}(j)$, $j \neq \tau(i)$ we get the inequality $\gameval^{\sigma,\tau}_{i}(\payoff) + \bias^{\sigma,\tau}_i(\payoff) < \payoff_{ij} + \bias^{\sigma,\tau}_j(\payoff)$. Hence $\payoff \in \polyh^{\sigma, \tau}$, which implies that $(\R^m \setminus \Generic) \subseteq \mathcal{H}$. This proves that $\R^m \setminus \Generic$ has Lebesgue measure $0$.
\end{proof}

\begin{remark}
As noted in \cref{general_cone}, some of the cones $\polyh^{\sigma,\tau}$ may be empty. The set of policies $(\sigma,\tau)$ that define a non-empty cone $\polyh^{\sigma,\tau}$ depends on the underlying graph. In the case of ergodic graphs, one can combine \cref{policy_cone} with \cite[Theorem 4.24]{LoffSkomra:2024} to get a complete characterization: a cone $\polyh^{\sigma,\tau}$ is non-empty if and only if the graph $\dgraph^{\sigma,\tau}$ obtained by fixing $(\sigma,\tau)$ has a single directed cycle. For arbitrary graphs, \cref{policy_cone} does not explain which cones $\polyh^{\sigma,\tau}$ are non-empty, but it implies that these non-empty cones cover almost the entire space of weights. This is sufficient for our purposes.
\end{remark}

Thanks to \cref{policy_cone}, we can now extend the definition of a condition number from games with only one equilibrium cycle to any generic game.

\begin{definition}[Slack $\sk(\payoff)$ and condition number $\Cond(\payoff)$]\label{def:condition-number_general}
Suppose that $\payoff \in \polyh^{\sigma,\tau}$ for some pair of policies $(\sigma,\tau)$. Then, the \emph{slack} of the vector $\payoff$, denoted $\sk(\payoff)$, is the largest positive number $\delta$ that satisfies the inequalities
\begin{align*}
&\forall (i,j) \in \edges, i \in \Maxvertices, \vertices^{\sigma,\tau}(i) \neq \vertices^{\sigma,\tau}(j), &\gameval_{i}^{\sigma,\tau}(\payoff) \ge \delta +\gameval_{j}^{\sigma,\tau}(\payoff), \\
&\forall (i,j) \in \edges, i \in \Minvertices, \vertices^{\sigma,\tau}(i) \neq \vertices^{\sigma,\tau}(j), &\gameval_{i}^{\sigma,\tau}(\payoff) \le -\delta +\gameval_{j}^{\sigma,\tau}(\payoff), \\
&\forall (i,j) \in \edges, i \in \Maxvertices, \vertices^{\sigma,\tau}(i) = \vertices^{\sigma,\tau}(j), \sigma(i) \neq j, &\gameval^{\sigma,\tau}_{i}(\payoff) + \bias^{\sigma,\tau}_i(\payoff) \ge \delta + \payoff_{ij} + \bias^{\sigma,\tau}_j(\payoff), \\
&\forall (i,j) \in \edges, i \in \Minvertices, \vertices^{\sigma,\tau}(i) = \vertices^{\sigma,\tau}(j), \tau(i) \neq j, &\gameval^{\sigma,\tau}_{i}(\payoff) + \bias^{\sigma,\tau}_i(\payoff) \le -\delta + \payoff_{ij} + \bias^{\sigma,\tau}_j(\payoff).
\end{align*}
We define a \emph{condition number} of the MPG with weights $\payoff$ as $\Cond(\payoff) = 0$ if every vertex of $\dgraph$ has out-degree $1$ and, otherwise, as
\[
\Cond(\payoff) \coloneqq \frac{\hilb{\payoff}}{\sk(\payoff)} \, .
\]
\end{definition}
\begin{remark}
We note that if $(\sigma,\tau) \in \Xi$, i.e. the pair of policies $(\sigma,\tau)$ induces a single directed cycle (cf. \Cref{one_cycle_bias_def}), then the value is constant and the definitions above are identical to the ones given in \cref{def:condition-number}. Hence, the definitions above generalize \cref{def:condition-number} to all generic weight vectors, i.e. to games whose optimal solutions may have more than one cycle. Indeed, \cref{policy_cone} shows that the condition number is uniquely defined for every $\payoff \in \Generic$ and that the set $\R^m \setminus \Generic$ has Lebesgue measure zero.
\end{remark}

The following lemma shows that $\Cond$ is an upper bound on the condition numbers of the games on the graphs $\dgraph_{\ell}$.

\begin{lemma}\label{decomposition_condition_upper}
If $\payoff \in \Generic$, then for every $\ell \in [p]$ we have $\Cond_{\ell}(\payoff) \le \Cond(\payoff)$.
\end{lemma}
\begin{proof}
Fix any $\ell \in [p]$. If every vertex of $\dgraph_{\ell}$ has out-degree $1$, then the claim is trivial. Otherwise, denote $\gameval = \gameval^{\sigma,\tau}(\payoff)$ and $\bias = \bias^{\sigma,\tau}(\payoff)$. By \cref{constant_value_subgame} we have
\[
\Cond_{\ell}(\payoff) = \frac{\max\{\abs{\payoff_{ij} - \payoff_{i'j'}} \colon (i,j),(i',j') \in \edges \cap (\vertices_{\ell} \times \vertices_{\ell})\}}{\min\{\abs{\payoff_{ij} - \gameval_i + \bias_j - \bias_i} \colon (i,j) \in \edges \cap (\vertices_{\ell} \times \vertices_{\ell}), \payoff_{ij} - \gameval_i + \bias_j - \bias_i \neq 0\} } \, .
\]
The numerator of this fraction is bounded from above by $\hilb{\payoff}$. Furthermore, we have
\begin{align*}
&\min\{\abs{\payoff_{ij} - \gameval_i + \bias_j - \bias_i} \colon (i,j) \in \edges \cap (\vertices_{\ell} \times \vertices_{\ell}), \payoff_{ij} - \gameval_i + \bias_j - \bias_i \neq 0\} \\
&\ge \min\{\abs{\payoff_{ij} - \gameval + \bias_j - \bias_i} \colon (i,j) \in \edges, \vertices^{\sigma,\tau}(i) = \vertices^{\sigma,\tau}(j), \payoff_{ij} - \gameval + \bias_j - \bias_i \neq 0\} \\
&\ge \sk(\payoff) \, ,
\end{align*}
which gives $\Cond_{\ell}(\payoff) \le \Cond(\payoff)$.
\end{proof}

In later sections, we will prove that the number $\Cond(\payoff)$ is small for games with weights perturbed by Gaussian noise. To do this, we will use the fact that the cones $\polyh^{\sigma,\tau}$ are ``wide''. This is formalized in the following lemma.

\begin{lemma}\label{wide_cone}
Fix a pair of policies $(\sigma,\tau)$ and suppose that the cone $\polyh^{\sigma,\tau}$ is nonempty. Then, $\polyh^{\sigma,\tau}$ contains a point $\payoff^*$ such that $\supnorm{\payoff^*} = 1$ and $\sk(\payoff^*) \ge \frac{1}{n}$.
\end{lemma}
\begin{proof}
Let $\payoff$ be some point in $\polyh^{\sigma,\tau}$. By \cref{cone_gives_bias}, the vector $\gameval = \gameval^{\sigma,\tau}(\payoff) \in \R^n$ is the value of the game with weights $\payoff$. Let $\dgraph_{\ell} \coloneqq \bigl(\vertices_{\ell}, \edges \cap (\vertices_{\ell} \times \vertices_{\ell})\bigr)$ for $\ell \in [p]$ be the constant value components of the game as given by \cref{constant_value_subgame}. By permuting the indices, we can suppose that the sets $\vertices_1, \dots, \vertices_p$ are ordered in such a way that the value on each subgame forms a non-decreasing sequence, i.e., that for all $\ell_1 \le \ell_2$ and $i \in \vertices_{\ell_1}, j \in \vertices_{\ell_2}$ we have $\gameval_i \le \gameval_j$. Consider the point $\payoff'$ defined as
\[
\payoff'_{ij} \coloneqq \begin{cases}
\ell &\textit{if $i \in \Maxvertices \cap \vertices_\ell$ and $j = \sigma(i)$}, \\
\ell &\textit{if $i \in \Minvertices \cap \vertices_\ell$ and $j = \tau(i)$}, \\
\ell-1 &\textit{if $i \in \Maxvertices \cap \vertices_\ell$ and $j \neq \sigma(i)$}, \\
\ell+1 &\textit{if $i \in \Minvertices \cap \vertices_\ell$ and $j \neq \tau(i)$}.
\end{cases}
\]
We claim that $\payoff' \in \polyh^{\sigma,\tau}$. To see that this is the case, fix $i \in [n]$ and let $\ell$ be such that $i \in \vertices_\ell$. Then, we have $\gameval^{\sigma,\tau}_i(\payoff') = \ell$ and $\bias^{\sigma,\tau}_i(\payoff') = 0$ by the definitions of $\payoff'$ and of the functions $\gameval^{\sigma,\tau},\bias^{\sigma,\tau}$. Suppose that $i \in \Maxvertices$ and fix an edge $(i,j)$. Furthermore, let $\ell'$ be such that $j \in \vertices_{\ell'}$. If $\ell' \neq \ell$, then the fact that $\payoff \in \polyh^{\sigma,\tau}$ implies that $\gameval_i > \gameval_j$. Hence $\ell > \ell'$ and $\gameval^{\sigma,\tau}_i(\payoff') \ge 1 + \gameval^{\sigma,\tau}_j(\payoff')$. Furthermore, if $\ell' = \ell$ and $j \neq \sigma(i)$, then $\payoff'_{ij} = \ell - 1$ and therefore $\gameval^{\sigma,\tau}_{i}(\payoff') + \bias^{\sigma,\tau}_i(\payoff') = \ell = 1 + (\ell - 1) = 1 + \payoff'_{ij} + \bias^{\sigma,\tau}_j(\payoff')$. Analogous statements hold if $i \in \Minvertices$. Therefore, $\payoff' \in \polyh^{\sigma,\tau}$ and $\sk(\payoff') \ge 1$. Hence, $\payoff^* = \payoff'/\supnorm{\payoff'}$ satisfies the claim.
\end{proof}

To finish this section, we note that the results from the previous sections do not immediately show that well-conditioned games can be solved efficiently. Indeed, if $\Cond(\payoff)$ is small, then every game on $\dgraph_{\ell}$ can be quickly solved by the \KM iteration or the Halpern iteration. However, we first have to find the decomposition of the game on $\dgraph$ into games on $\dgraph_{\ell}$. For this reason, in the next section we prove that this decomposition can be found quickly for well-conditioned games.

\section{Finding value and decomposing the game using value iteration}\label{sec:value_decomposition}

In this section, we prove that various value iteration algorithms can find the edges between vertices of different values. Moreover, they can do this quickly for well-conditioned games. To do so formally, we now define our second condition number for mean payoff games. This condition number should be compared with the condition numbers studied in \cite{mtns2018,AllamigeonGaubertKatzSkomra:2025}. One condition number introduced in these papers can be applied to constant value games with non-zero value and is given by $\frac{\inf_{\bias}\hilb{\bias}}{\abs{\gameval}}$, where $\gameval$ is the (constant) value of the game, and the infimum goes over all biases of the game. The papers \cite{mtns2018,AllamigeonGaubertKatzSkomra:2025} show that the Zwick--Paterson iteration can detect if the value is positive or negative using a number of iterations that is polynomial in this condition number. For games with non-constant value, another condition number, analyzed in \cite[Theorem~11]{AllamigeonGaubertKatzSkomra:2025}, allows to use the Zwick--Paterson value iteration to detect the set of states with maximal value. As discussed in \cite{AllamigeonGaubertKatzSkomra:2025}, both condition numbers can be applied to more general games such as stochastic mean payoff games and entropy games. The purpose of our condition number is slightly different: instead of finding the set of states with maximal value we want to find edges between states with different values. Furthermore, we only consider deterministic mean payoff games. Therefore, we define a modification of the condition number from \cite[Theorem~11]{AllamigeonGaubertKatzSkomra:2025} that is suitable for our purposes.

\begin{definition}[$\Cond^{\val}(\payoff)$]
For any game with weights $\payoff \in \R^m$ and value $\gameval \in \R^n$ we define the \emph{value condition number} $\Cond^{\val}(\payoff)$ as $\Cond^{\val}(\payoff) = 0$ if there exist no edges between vertices of different value and
\[
\Cond^{\val}(\payoff) \coloneqq \frac{\hilb{\payoff}}{\min_{(i,j) \in \edges, \gameval_i \neq \gameval_j}|\gameval_i - \gameval_j|}
\]
otherwise.
\end{definition}

We note that the value condition number is upper bounded by the condition number from \cref{def:condition-number_general}.

\begin{lemma}\label{value_cond_bound}
Suppose that $\payoff \in \Generic$. Then, $\Cond^{\val}(\payoff) \le \Cond(\payoff)$.
\end{lemma}
\begin{proof}
  If there are no edges between vertices of different value, then
  $\Cond^{\val}(\payoff)=0$, and the claim is immediate since
  $\Cond(\payoff)\ge 0$. Otherwise, the definition of slack gives us
  $\sk(\payoff) \le \min_{(i,j) \in \edges, \gameval_i \neq \gameval_j}|\gameval_i - \gameval_j|$,
  from which the claim follows.
\end{proof}

We now show that if $\Cond^{\val}(\payoff)$ is small, then the edges between vertices of different value can be efficiently found using various value iteration algorithms. More precisely, we show this for the Zwick--Paterson iteration, the \KM iteration, and the Halpern iteration. For the sake of simplicity, we will only consider value iteration algorithms starting from $0$, i.e., with $u^0 = 0$.

We start by analyzing the Zwick--Paterson iteration, i.e., the sequence given by $u^k \coloneqq \shapley(u^{k-1})$ for $k \ge 1$. In this case, \cite{zwick_paterson} showed that the normalized sequence $u^k/k$ converges to the value of the mean payoff game.

\begin{lemma}[{\cite[Theorem~2.2]{zwick_paterson}}]\label{val_iter_estimate}
Let $\gameval \in \R^n$ be the value vector of the mean payoff game with weights $\payoff$. Then, for every $k \ge 0$ and $i \in \vertices$ we have
\[
k\gameval_i -2n\supnorm{\payoff} \le \bias^k_i \le k\gameval_i + 2n\supnorm{\payoff} \, .
\]
In particular, the sequence $(u^k/k)$ converges to $\gameval$.
\end{lemma}

A translation argument gives a variant of \cref{val_iter_estimate} for the Hilbert seminorm.

\begin{lemma}\label{zp_estimate_hilbert}
Let $\gameval \in \R^n$ be the value vector of the mean payoff game. Then, for every $k \ge 0$ and $i \in \vertices$ we have
\[
k\gameval_i -n\hilb{\payoff} \le \bias^k_i \le k\gameval_i + n\hilb{\payoff} \, .
\]
\end{lemma}
\begin{proof}
Let $\alpha \in \R$ be such that $\hilb{\payoff} = 2\supnorm{\payoff + \alpha \ones}$. Consider the game with weights $\tilde{\payoff} = \payoff + \alpha \ones$. Its Shapley operator is then given by $\tilde{\shapley}(x) = \alpha \ones + \shapley(x)$ for any $x$. Hence, if we define the sequence $\tilde{\bias}^0 = 0$ and $\tilde{\bias}^k = \tilde{\shapley}(\bias^{k-1})$, then $\tilde{\bias}^k = k\alpha\ones + \bias^k$ for all $k$. Moreover, the value vector of this game satisfies $\tilde{\gameval} = \alpha\ones + \gameval$. By applying \cref{val_iter_estimate} to this game we get $\tilde{\bias}^k_i \le k\tilde{\gameval}_i + 2n\supnorm{\tilde{\payoff}}$, so $\bias^k \le k \gameval_i + n\hilb{\payoff}$. Analogously, $\bias^k \ge k \gameval_i - n\hilb{\payoff}$.
\end{proof}

This estimate leads to the following theorem that shows that we can use the Zwick--Paterson value iteration to detect the edges between vertices of different values.

\begin{theorem}\label{condition_of_zp}
Fix $\payoff \in \R^m$ and consider the Zwick--Paterson value iteration sequence: $u^0 = 0$ and $u^k = \shapley(u^{k-1})$. For any edge $(i,j) \in \edges$ we have the following two properties:
\begin{itemize}
\item if $\gameval_i = \gameval_j$, then $\abs{\bias_i^k - \bias_j^k} \le 2n\hilb{\payoff}$ for all $k$.
\item if $\gameval_i \neq \gameval_j$, then for any $k \ge 1 + 4n\Cond^{\val}(\payoff)$ we have $\abs{\bias_i^k - \bias_j^k} > 2n\hilb{\payoff}$.
\end{itemize}
\end{theorem}
\begin{proof}
Fix an edge $(i,j)$ and suppose that $\gameval_i \ge \gameval_j$. By \cref{zp_estimate_hilbert}, for any $k \ge 0$ we have
\[
k (\gameval_i - \gameval_j) - 2n\hilb{\payoff}  \le \bias^k_i - \bias^k_j \le k (\gameval_i - \gameval_j) + 2n\hilb{\payoff} \, .
\]
Hence, if $\gameval_i = \gameval_j$, then $\abs{\bias^k_i - \bias^k_j} \le 2n\hilb{\payoff}$ for any $k \ge 0$. Moreover, if $\gameval_i \neq \gameval_j$, then $\gameval_i - \gameval_j \ge \frac{\hilb{\payoff}}{\Cond^{\val}(\payoff)}$ by the definition of the condition number. Hence, if we assume that $k \ge 1 + 4n\Cond^{\val}(\payoff)$, then $k(\gameval_i - \gameval_j) > 4n\hilb{\payoff}$ and $\bias^k_i - \bias^k_j > 2n\hilb{\payoff}$. An analogous reasoning holds if $\gameval_i \le \gameval_j$.
\end{proof}

The theorem above gives a trivial algorithm that detects the edges between vertices of different value: we simply check if an edge satisfies $\abs{\bias_i^k - \bias_j^k} > 2n\hilb{\payoff}$. The drawback of this algorithm is that it does not provide a way to check that all such edges have already been found. Indeed, if we do not have an estimate for the number $\Cond^{\val}(\payoff)$, and we run the Zwick--Paterson value iteration, then it is not clear when to stop the algorithm. We will deal with this difficulty later, in \cref{sec:zph_vi}, where we show how to combine the Zwick--Paterson iteration with the Halpern iteration to solve arbitrary mean payoff games.

In the rest of this section, we analyze the \KM iteration and the Halpern iteration. We show that the normalized sequence of the \KM iteration and the Halpern iteration converges to the value vector, in a similar way as the Zwick--Paterson iteration.

Consider the \KM iteration started at $0$, i.e., the iteration sequence given by $u^0 \coloneqq 0$, $u^{k} \coloneqq \frac{1}{2}u^{k-1} + \frac{1}{2}\shapley(u^{k-1})$. In the next proposition, we show that the sequence $(u^k/k)$ converges to $\gameval/2$, where $\gameval$ is the value of the game. Before presenting the proof, we give some comments that put this result in a broader context.

Define the operator $G(x) \coloneqq \frac{1}{2}x + \frac{1}{2}\shapley(x)$. The Shapley operator $\shapley$ is non-expansive in the supremum norm. This also implies that $G$ is non-expansive in the supremum norm. Since this operator is semi-algebraic, general results on non-expansive operators imply that the limit $\lim_{k \to \infty} G^{k}(0)/k$ exists \cite{bolte2013}. Even more, since both $F$ and $G$ are piecewise-affine, one can apply Kohlberg's theorem on invariant half-lines \cite{kohlberg} to show that the limit $\lim_{k \to \infty} G^{k}(0)/k$ is equal to $\gameval/2$. This implies that the sequence $(u^k/k)$ converges to $\gameval/2$ also for stochastic mean payoff games and not only for deterministic games. (We refer the reader to \cite{polyhedra_equiv_mean_payoff} and \cite[Chapter~6]{skomra_phd} for more information about Kohlberg's theorem and its applications to mean payoff games.) For deterministic games, the proposition below gives a stronger version of this result: we not only show that the sequence $(u^k/k)$ converges, but also estimate the rate of this convergence. We note that for stochastic Markov Decision Processes, an estimate of the rate of convergence of normalized \KM iteration appeared in \cite[Theorem~7]{LeeRyu:2025}. We also note that the operator $G$ has a less abstract interpretation: it is a Shapley operator of some stochastic mean payoff game. Using this interpretation, one can prove the proposition below by a probabilistic argument. Since we do not consider stochastic games in this work, we present a proof based on a counting argument.

\begin{proposition}\label{km_val_iter}
Let $\gameval \in \R^n$ be the value vector of the mean payoff game with weights $\payoff$. Furthermore, let $(u^k)$ be the sequence obtained from the \KM iteration. Then, for every $k \ge 0$ and $i \in \vertices$ we have
\[
k\frac{\gameval_i}{2} -2n\supnorm{\payoff} \le \bias^k_i \le k\frac{\gameval_i}{2} + 2n\supnorm{\payoff} \, .
\]
In particular, the sequence $(2u^k/k)$ converges to $\gameval$.
\end{proposition}
\begin{proof}
We adapt the argument from \cite[Theorem~2.2]{zwick_paterson}. Let $\sigma$ be an optimal policy of player Max and consider the Shapley operator $\shapley^{\sigma}$ of the one-player game in which Max plays according to $\sigma$. Note that $\shapley^{\sigma}(x) \le \shapley(x)$ for any $x \in \R^n$. Consider the \KM iteration on the operator $\shapley^{\sigma}$, i.e., the sequence given by $w^0 = 0$ and $w^{k} = \frac{1}{2} w^{k-1} + \frac{1}{2}\shapley^{\sigma}(w^{k-1})$. By induction, we have $w^{k} \le u^{k}$ for all $k$. Indeed, $w^0 = u^0 = 0$. If $w^{k-1} \le u^{k-1}$ coordinatewise, then, using the monotonicity of $\shapley$ and the inequality $\shapley^\sigma(x) \le \shapley(x)$ for every $x$, we get $w^k = \frac{1}{2}w^{k-1} + \frac{1}{2}\shapley^\sigma(w^{k-1}) \le \frac{1}{2}u^{k-1} + \frac{1}{2}\shapley(u^{k-1}) = u^k$. Let $\dgraph^{\sigma}$ be the directed graph of the one-player game in which Max plays according to $\sigma$. We modify this graph by adding a loop from every vertex of $\dgraph^{\sigma}$ to itself, with weight $0$. To distinguish these new loops from the loops that were previously in the graph, we color the new loops in red. Denote by $\dgraph'$ the modified graph. Given a directed walk $P$ in $\dgraph'$, we denote by $\payoff(P)$ the total weight of the edges on the walk. We claim that for every $k \ge 0$ and every $i \in [n]$, there exists a set $\mathcal{P}_{k,i}$ of $2^k$ directed walks in $\dgraph'$ such that: (a) every walk in $\mathcal{P}_{k,i}$ starts in $i$ and has $k$ edges; (b) for every $0 \le l \le k$, $\mathcal{P}_{k,i}$ contains exactly $\binom{k}{l}$ walks with $l$ red loops; (c) we have $w^{k}_i = \sum_{P \in \mathcal{P}_{k,i}} \frac{1}{2^k}\payoff(P)$. We prove the claim by induction. For $k = 0$, we take $\mathcal{P}_{0,i}$ to be the set consisting of one trivial walk that has only one vertex $i$ and zero edges. Suppose that the claim is true for $k - 1$ and all $i \in [n]$. Fix a vertex $i \in [n]$ and note that there exists $j \in [n]$ such that $(i,j)$ is an edge in $\dgraph^{\sigma}$ and
\[
w^k_i = \frac{1}{2} w^{k-1}_i + \frac{1}{2} \payoff_{ij} + \frac{1}{2} w^{k-1}_j = \sum_{P \in \mathcal{P}_{k-1,i}}\frac{1}{2^k}\payoff(P) + \sum_{P \in \mathcal{P}_{k-1,j}}\frac{1}{2^k}(\payoff_{ij} + \payoff(P)) \, .
\]
Hence, we construct $\mathcal{P}_{k,i}$ as a union of two sets: first obtained by prepending the red loop $(i,i)$ to every walk in $\mathcal{P}_{k-1,i}$ and second obtained by prepending the edge $(i,j)$ to every walk in $\mathcal{P}_{k-1,j}$. Note that the set $\mathcal{P}_{k,i}$ constructed in this way has $\binom{k-1}{l-1} + \binom{k-1}{l} = \binom{k}{l}$ walks with exactly $l$ red loops. In particular, $\mathcal{P}_{k,i}$ satisfies the properties (a)--(c) above and this finishes the induction step.

Let $P \in \mathcal{P}_{k,i}$ be a walk with $l$ red loops. By deleting the red loops from $P$, we obtain a directed walk in $\dgraph^{\sigma}$ of length $k-l$. Since $\sigma$ is optimal, any directed cycle in this walk has mean weight at least $\gameval_i$. If this walk has at least $n$ edges, then it contains a directed cycle, and hence we get the inequality
\[
\payoff(P) \ge \gameval_i(k - l - n + 1) - \supnorm{\payoff}(n-1) \, .
\]
Therefore,
\begin{align*}
u^{k}_i &\ge w^{k}_i \ge \gameval_i(k - n + 1) - \supnorm{\payoff}(n-1) - \gameval_i \frac{1}{2^k}\sum_{l = 0}^{k} \binom{k}{l}l \\
&= \gameval_i(k - n + 1) - \supnorm{\payoff}(n-1) - \frac{1}{2} \gameval_i k \ge k \frac{\gameval_i}{2} - 2\supnorm{\payoff}n \, .
\end{align*}
The proof of the other inequality is analogous, by fixing an optimal policy $\tau$ of player Min.
\end{proof}

Using a translation argument we can obtain a variant of \cref{km_val_iter} for the Hilbert seminorm.

\begin{lemma}\label{km_estimate_hilbert}
Let $\gameval \in \R^n$ be the value vector of the mean payoff game and let $(u^k)$ be the sequence obtained from the \KM iteration. Then, for every $k \ge 0$ and $i \in \vertices$ we have
\[
k\frac{\gameval_i}{2} -n\hilb{\payoff} \le \bias^k_i \le k\frac{\gameval_i}{2} + n\hilb{\payoff} \, .
\]
\end{lemma}
\begin{proof}
Let $\alpha \in \R$ be such that $\hilb{\payoff} = 2\supnorm{\payoff + \alpha \ones}$. Consider the game with weights $\tilde{\payoff} = \payoff + \alpha \ones$. Its Shapley operator is then given by $\tilde{\shapley}(x) = \alpha \ones + \shapley(x)$ for any $x$. Consider the sequence defined by $\tilde{\bias}^0 = 0$ and $\tilde{\bias}^k = \frac{1}{2}\tilde{\bias}^{k-1} + \frac{1}{2}\tilde{\shapley}(\tilde{\bias}^{k-1})$. Then, $\tilde{\bias}^k = k\frac{\alpha}{2}\ones + \bias^k$ for all $k$. Indeed, this is true for $k = 0$ and for all other $k$ we have by induction that
\begin{align*}
\tilde{\bias}^k &= \frac{1}{2}\tilde{\bias}^{k-1} + \frac{1}{2}\tilde{\shapley}(\tilde{\bias}^{k-1}) = \frac{1}{2}\bigl((k-1)\frac{\alpha}{2}\ones + \bias^{k-1}\bigr) + \frac{1}{2}\bigl(\alpha \ones + \shapley((k-1)\frac{\alpha}{2}\ones + \bias^{k-1})\bigr) \\
&= \frac{(k-1)\alpha + 2 \alpha + (k-1)\alpha}{4}\ones + \bias^{k} = k\frac{\alpha}{2}\ones + \bias^k \, .
\end{align*}
Moreover, the value vector of this game satisfies $\tilde{\gameval} = \alpha\ones + \gameval$. By applying \cref{km_val_iter} to this game we get $\tilde{\bias}^k_i \le k\frac{\tilde{\gameval}_i}{2} + 2n\supnorm{\tilde{\payoff}}$, so $k \frac{\alpha}{2}\ones + \bias^k \le k \frac{\alpha \ones + \gameval}{2} + n\hilb{\payoff}$ and $\bias^k  \le k\frac{\gameval}{2} + n\hilb{\payoff}$ . Analogously, $\bias^k \ge k\frac{\gameval}{2} - n\hilb{\payoff}$.
\end{proof}

As for the Zwick--Paterson iteration, we could use \cref{km_estimate_hilbert} to give a simple algorithm that detects if an edge goes between vertices of different value. However, the strength of the \KM iteration lies in the fact that we do not have to explicitly eliminate such edges during the execution of the algorithm (we show this later, in \cref{sec:complexity_KM}). For our purposes, it is enough to observe that the \KM starts to ignore such edges for sufficiently large $k$. This is formalized in the theorem below. It is crucial to understand this theorem, in order to understand why the \KM iteration works on non-ergodic graphs.

Let $\edges' \coloneqq \edges \setminus \{(i,j) \in \edges \colon \gameval_i \neq \gameval_j\}$. Consider the reduced game $\dgraph'$ obtained from $\dgraph$ by removing the edges between vertices with different values, so that $\dgraph' = (\vertices, \edges')$. We denote by $\shapley'$ the Shapley operator of the reduced game.

\begin{theorem}\label{value_condition_of_km}
After $O(n \Cond^{\val}(\payoff))$ iterations, the \KM value iteration detects all the edges between vertices of different values and starts to ignore them. More formally, consider the \KM value iteration sequence on the original game: $u^0 = 0$ and $u^k = \frac{1}{2}u^{k-1} + \frac{1}{2}\shapley(u^{k-1})$. Furthermore, fix $k^{*} \ge \lfloor 1 + (4n+2)\Cond^{\val}(\payoff) \rfloor$ and consider the \KM value iteration sequence on the reduced game played on $\dgraph'$: $\tilde{u}^0 = u^{k^*}$ and $\tilde{u}^k = \frac{1}{2}\tilde{u}^{k-1} + \frac{1}{2}\shapley'(\tilde{u}^{k-1})$. Then, for every $k \ge k^*$ we have $u^{k} = \tilde{u}^{k - k^*}$. Moreover, $u^k$ induces the same policies in the game on $\dgraph$ and in the reduced game on $\dgraph'$.
\end{theorem}
\begin{proof}
Fix $k \ge k^*$. We claim that $\shapley(u^k) = \shapley'(u^k)$. To see this, fix $i \in \Maxvertices$ and note that $\shapley(u^k)_i = \max_{(i,j) \in \edges}\{\payoff_{ij} + u^k_j\}$ and $\shapley'(u^k)_i = \max_{(i,j) \in \edges, \gameval_i = \gameval_j}\{\payoff_{ij} + u^k_j\}$.

Consider any two edges $(i,j), (i,l) \in \edges$ such that $\gameval_i = \gameval_j$ and $\gameval_i \neq \gameval_l$. Since $i \in \Maxvertices$, we have $\gameval_i > \gameval_l$ and therefore $\gameval_j > \gameval_l$. Combining \cref{km_estimate_hilbert} and the fact that $k > (4n+2)\Cond^{\val}(\payoff)$ we get
\[
u^{k}_j - u^{k}_l \ge k \frac{\gameval_j - \gameval_l}{2} - 2n \hilb{\payoff} > \hilb{\payoff} \ge \payoff_{il} - \payoff_{ij} \, .
\]
Hence, $\payoff_{ij} + u^{k}_j > \payoff_{il} + u^{k}_l$. In particular, the maximum in $\max_{(i,j) \in \edges}\{\payoff_{ij} + u^k_j\}$ cannot be achieved by an edge that satisfies $\gameval_i \neq \gameval_j$. Therefore $\shapley(u^k)_i = \shapley'(u^k)_i$. The proof is analogous if $i \in \Minvertices$. In particular, $u^k$ induces the same policies in the game on $\dgraph$ and in the reduced game on $\dgraph'$. To finish the proof, note that $u^{k^*} = \tilde{u}^0$ by definition. For every $k > k^*$ we get by induction that
\[
u^{k} = \frac{1}{2} u^{k-1} + \frac{1}{2}\shapley(u^{k-1}) = \frac{1}{2} u^{k-1} + \frac{1}{2}\shapley'(u^{k-1}) = \frac{1}{2} \tilde{u}^{k - 1 - k^*} + \frac{1}{2}\shapley'(\tilde{u}^{k-1 - k^*}) = \tilde{u}^{k-k^*} \, .
\]
\end{proof}

We now prove analogous results for the Halpern iteration. As previously, we will only consider the Halpern iteration starting from $0$, i.e., the iteration sequence given by $u^0 \coloneqq 0$, $u^{k} \coloneqq \frac{k}{k+2}\shapley(u^{k-1})$. In this case, we show that the sequence $(u^k/k)$ converges to $\gameval/3$ and we give an estimate on the rate of this convergence. Contrary to the case of \KM iteration, it is less clear if general results on non-expansive operators can be applied to the Halpern iteration for non-constant value games. Nevertheless, for stochastic Markov Decision Processes, an estimate of the rate of convergence of the normalized Halpern iteration appeared in \cite[Theorem~9]{LeeRyu:2025} and our proof is based on this result.

\begin{proposition}\label{halpern_val_iter}
Let $\gameval \in \R^n$ be the value vector of the mean payoff game with weights $\payoff$. Furthermore, let $(u^k)$ be the sequence obtained from the Halpern iteration. Then, for every $k \ge 0$ and $i \in \vertices$ we have
\[
k\frac{\gameval_i}{3} -4n\supnorm{\payoff} \le \bias^k_i \le k\frac{\gameval_i}{3} + 4n\supnorm{\payoff} \, .
\]
In particular, the sequence $(3u^k/k)$ converges to $\gameval$.
\end{proposition}
\begin{proof}
Our proof is a two-player variant of the reasoning given in \cite[Theorem~9]{LeeRyu:2025}. Fix $\varepsilon > 0$ and let $\bias^{*}$ be a bias of the game such that $\supnorm{\bias^*} \le (2+\varepsilon)n\supnorm{\payoff}$ (such a bias exists by \cref{short_bias}). Furthermore, let $(\sigma,\tau)$ be a pair of policies induced by $\bias^{*}$. As in the proof of \cref{km_val_iter}, consider the Shapley operator $\shapley^{\sigma}$ of the one-player game in which Max plays according to $\sigma$ and note that $\shapley^{\sigma}(x) \le \shapley(x)$ for any $x \in \R^n$. Consider the Halpern iteration on the operator $\shapley^{\sigma}$, i.e., the sequence given by $w^0 = 0$ and $w^{k} = \frac{k}{k+2}\shapley^{\sigma}(w^{k-1})$. By induction, we have $w^{k} \le u^{k}$ for all $k$. Let $\dgraph^{\sigma}$ be the directed graph of the one-player game in which Max plays according to $\sigma$. Furthermore, for any natural numbers $p \le q$ let $A(p,q) \coloneqq \frac{p}{p+2} \frac{p+1}{p+3} \dots \frac{q}{q+2}$, with the convention that $A(p,p) = \frac{p}{p+2}$. We claim that for every $k \ge 1$ and every $i \in [n]$, there exists a directed walk $(i,i_1,\dots,i_k)$ in $\dgraph^{\sigma}$ such that
\[
w^k_i = A(k,k) \payoff_{i,i_1} + A(k-1,k) \payoff_{i_1,i_2} + \dots + A(1,k) \payoff_{i_{k-1},i_{k}}.
\]
Indeed, this claim follows by induction. For $k = 1$ and fixed $i \in [n]$, there exists $i_1$ such that $\shapley^\sigma(0)_i = \payoff_{i,i_1}$, hence $w^1_i = A(1,1)\payoff_{i,i_1}$. Suppose that the claim is true for $k-1$ and all $i \in [n]$. Fix a vertex $i \in [n]$ and note that there exists $j \in [n]$ such that $(i,j)$ is an edge in $\dgraph^{\sigma}$ and $w^k_i = A(k,k) (\payoff_{i,j} + w^{k-1}_j)$. Therefore, there exists a directed walk $(i,j,j_1,\dots,j_{k-1})$ in $\dgraph^{\sigma}$ such that
\[
w^k_i = A(k,k) (\payoff_{i,j} + w^{k-1}_j) = A(k,k)\bigl( \payoff_{i,j} + A(k-1,k-1)\payoff_{j,j_1} + \dots + A(1,k-1)\payoff_{j_{k-2},j_{k-1}} \bigr) \, .
\]
The induction claim follows from the fact that $A(k,k)A(p,k-1) = A(p,k)$ for all $p < k$.
To bound $w^k_i$ from below, we use the definition of the bias $\bias^{*}$. Let $i_0 = i$. By definition of the bias, for every pair $(i_{p},i_{p+1})$ we have
\[
\payoff_{i_p,i_{p+1}} \ge \gameval_{i_p} + \bias^*_{i_p} - \bias^*_{i_{p+1}} \ge \gameval_{i} + \bias^*_{i_p} - \bias^*_{i_{p+1}} \, .
\]
Hence, we get
\begin{align*}
w^k_i &\ge A(k,k)(\gameval_{i} + \bias^*_{i_0} - \bias^*_{i_{1}}) + \dots + A(1,k)(\gameval_{i} + \bias^*_{i_{k-1}} - \bias^*_{i_{k}}) \\
&= \bigl(A(1,k) + \dots + A(k,k)\bigr) \gameval_{i} + A(k,k)\bias^*_{i_0}  - A(1,k)\bias^*_{i_{k}} +  \sum_{p = 1}^{k-1} \bigl(A(p,k) - A(p+1,k)\bigr)\bias^{*}_{i_{k-p}} \, .
\end{align*}
To analyze the last expression, note that $A(1,k) + \dots + A(k,k) = \frac{k}{3}$. Indeed, this is true for $k = 1$ and for higher $k$ we have $A(1,k) + \dots + A(k,k) = A(k,k)\bigl(A(1,k-1) + \dots + A(k-1,k-1) + 1\bigr) = \frac{k}{k+2}(\frac{k-1}{3} + 1) = \frac{k}{3}$ by induction. Furthermore, the value of $A(p,k) - A(p+1,k)$ is negative. Hence, using the estimate $(2+\varepsilon)n\supnorm{\payoff} \ge \bias^*_{i_p} \ge -(2+\varepsilon)n\supnorm{\payoff}$ we get the inequality
\begin{align*}
u^k_i \ge w^k_i &\ge \frac{k}{3}\gameval_i - (2+\varepsilon)n\supnorm{\payoff}\bigl( A(k,k) + A(1,k) + \sum_{p = 1}^{k-1}(A(p+1,k) - A(p,k)) \bigr) \\
&= \frac{k}{3}\gameval_i -(4+2\varepsilon)n\supnorm{\payoff} A(k,k) \ge \frac{k}{3}\gameval_i - (4+2\varepsilon)n\supnorm{\payoff} \, .
\end{align*}
Since $\varepsilon > 0$ was arbitrary, we get $u^k_i \ge \frac{k}{3}\gameval_i - 4n\supnorm{\payoff}$. The proof of the other inequality is analogous, by fixing the optimal policy $\tau$ of player Min.
\end{proof}

As before, a translation argument gives a variant of \cref{halpern_val_iter} for the Hilbert seminorm.

\begin{lemma}\label{halpern_estimate_hilbert}
Let $\gameval \in \R^n$ be the value vector of the mean payoff game and let $(u^k)$ be the sequence obtained from the Halpern iteration. Then, for every $k \ge 0$ and $i \in \vertices$ we have
\[
k\frac{\gameval_i}{3} -2n\hilb{\payoff} \le \bias^k_i \le k\frac{\gameval_i}{3} + 2n\hilb{\payoff} \, .
\]
\end{lemma}
\begin{proof}
Let $\alpha \in \R$ be such that $\hilb{\payoff} = 2\supnorm{\payoff + \alpha \ones}$. Consider the game with weights $\tilde{\payoff} = \payoff + \alpha \ones$. Its Shapley operator is then given by $\tilde{\shapley}(x) = \alpha \ones + \shapley(x)$ for any $x$. Consider the sequence defined by $\tilde{\bias}^0 = 0$ and $\tilde{\bias}^k = \frac{k}{k+2}\tilde{\shapley}(\tilde{\bias}^{k-1})$. Then, $\tilde{\bias}^k = k\frac{\alpha}{3}\ones + \bias^k$ for all $k$. Indeed, this is true for $k = 0$ and for all other $k$ we have by induction that
\begin{align*}
\tilde{\bias}^k &= \frac{k}{k+2}\tilde{\shapley}(\tilde{\bias}^{k-1}) = \frac{k}{k+2}\bigl(\alpha \ones + \shapley((k-1)\frac{\alpha}{3}\ones + \bias^{k-1})\bigr) \\
&= \frac{k}{k+2}(\frac{k+2}{3} \alpha \ones +\shapley(u^{k-1}) ) = k\frac{\alpha}{3}\ones + \bias^k \, .
\end{align*}
Moreover, the value vector of this game satisfies $\tilde{\gameval} = \alpha\ones + \gameval$. By applying \cref{halpern_val_iter} to this game we get $\tilde{\bias}^k_i \le k\frac{\tilde{\gameval}_i}{3} + 4n\supnorm{\tilde{\payoff}}$, so $k \frac{\alpha}{3}\ones + \bias^k \le k \frac{\alpha \ones + \gameval}{3} + 2n\hilb{\payoff}$ and $\bias^k  \le k\frac{\gameval}{3} + 2n\hilb{\payoff}$ . Analogously, $\bias^k \ge k\frac{\gameval}{3} - 2n\hilb{\payoff}$.
\end{proof}

As in the case of the \KM iteration, the Halpern iteration also quickly detects all the edges between vertices of different value and starts to ignore them. As previously, we consider the reduced game $\dgraph'$ obtained from $\dgraph$ by removing the edges between vertices with different values and we denote by $\shapley'$ the Shapley operator of the reduced game.

\begin{theorem}\label{value_condition_of_halpern}
After $O(n \Cond^{\val}(\payoff))$ iterations, the Halpern value iteration detects all the edges between vertices of different values and starts to ignore them. More formally, consider the Halpern value iteration sequence on the original game: $u^0 = 0$ and $u^k = \frac{k}{k+2}\shapley(u^{k-1})$. Furthermore, fix $k^{*} \ge \lfloor 1 + (12n+3)\Cond^{\val}(\payoff) \rfloor$ and consider the Halpern value iteration sequence with shifted weights on the reduced game played on $\dgraph'$: $\tilde{u}^0 = u^{k^*}$ and $\tilde{u}^k = \frac{k+k^*}{k+k^*+2}\shapley'(\tilde{u}^{k-1})$. Then, for every $k \ge k^*$ we have $u^{k} = \tilde{u}^{k - k^*}$. Moreover, $u^k$ induces the same policies in the game on $\dgraph$ and in the reduced game on $\dgraph'$.
\end{theorem}
\begin{proof}
Fix $k \ge k^*$. We claim that $\shapley(u^k) = \shapley'(u^k)$. To see this, fix $i \in \Maxvertices$ and note that $\shapley(u^k)_i = \max_{(i,j) \in \edges}\{\payoff_{ij} + u^k_j\}$ and $\shapley'(u^k)_i = \max_{(i,j) \in \edges, \gameval_i = \gameval_j}\{\payoff_{ij} + u^k_j\}$. Consider any two edges $(i,j), (i,l) \in \edges$ such that $\gameval_i = \gameval_j$ and $\gameval_i \neq \gameval_l$. Since $i \in \Maxvertices$, we have $\gameval_i > \gameval_l$ and therefore $\gameval_j > \gameval_l$. Combining \cref{halpern_estimate_hilbert} and the fact that $k > (12n+3)\Cond^{\val}(\payoff)$ we get
\[
u^{k}_j - u^{k}_l \ge k \frac{\gameval_j - \gameval_l}{3} - 4n \hilb{\payoff} > \hilb{\payoff} \ge \payoff_{il} - \payoff_{ij} \, .
\]
Hence, $\payoff_{ij} + u^{k}_j > \payoff_{il} + u^{k}_l$. In particular, the maximum in $\max_{(i,j) \in \edges}\{\payoff_{ij} + u^k_j\}$ cannot be achieved by an edge that satisfies $\gameval_i \neq \gameval_j$. Therefore $\shapley(u^k)_i = \shapley'(u^k)_i$. The proof is analogous if $i \in \Minvertices$. In particular, $u^k$ induces the same policies in the game on $\dgraph$ and in the reduced game on $\dgraph'$. To finish the proof, note that $u^{k^*} = \tilde{u}^0$ by definition. For every $k > k^*$ we get by induction that
\[
u^{k} = \frac{k}{k+2}\shapley(u^{k-1}) = \frac{k}{k+2}\shapley'(u^{k-1}) = \frac{(k-k^*) + k^*}{(k-k^*) + k^* + 2}\shapley'(\tilde{u}^{k-1 - k^*}) = \tilde{u}^{k-k^*} \, .
\]
\end{proof}

\section{Algorithms whose complexity is bounded by our condition number}\label{sec:complexity_condition}

In this section, we give a number of algorithms for solving deterministic mean payoff games. For each algorithm, we prove that the algorithm eventually finds a pair of optimal policies for any game, and that the complexity of the algorithm can be bounded by the condition number $\Cond(\payoff)$ for all generic weights. The first three algorithms are based on \KM and Halpern iteration, while the last algorithms are versions of the policy iteration algorithms from \cite{LoffSkomra:2024}.

\subsection{\KM iteration}\label{sec:complexity_KM}
We start by analyzing the \KM iteration. The theorem below bounds the number of iterations that the algorithm needs to find a pair of optimal policies. In the subsequent remark, we discuss some implementation details of this algorithm.

\begin{theorem}\label{overall_complexity_KM}
Consider the \KM iteration initialized at zero, $u^0 \coloneqq 0$, $u^{k} \coloneqq \frac{1}{2}u^{k-1} + \frac{1}{2}\shapley(u^{k-1})$. Then, for all sufficiently large $k$, every policy induced by $u^k$ is optimal. Furthermore, if $\payoff \in \Generic$, then the above is true for all
\[
k \ge 2 + 2(2n+1)\Cond(\payoff) +  144(4n+3)^2n^2\Cond(\payoff)^2 \, .
\]
\end{theorem}
\begin{proof}
Let $\Cond^{\val}(\payoff)$ be the value condition number of the game. Let $k^* \coloneqq \lfloor 1 + (4n+2)\Cond^{\val}(\payoff) \rfloor$. By \cref{value_condition_of_km}, for all $k \ge k^*$, the \KM iteration ignores all of the edges between vertices of different values. Denote by $\gameval$ the value of the game, and let $\edges' \coloneqq \edges \setminus \{(i,j) \in \edges, \gameval_i \neq \gameval_{j}\}$. Consider the reduced graph $\dgraph' = (\vertices, \edges')$ and let $\dgraph_1, \dots, \dgraph_p$ be the connected components of $\dgraph'$ with vertex sets $\vertices_{1}, \dots, \vertices_{p}$. For $k \ge k^*$, the \KM iteration on $\dgraph$ behaves in the same way as running $p$ \KM iterations on every $\dgraph_{\ell}$ in parallel starting from $(u^{k^*}_i)_{i \in \vertices_{\ell}}$. Since each game on $\dgraph_{\ell}$ has constant value by \cref{game_decomposition}, \cref{halpern_for_bias} shows that for all sufficiently large $k$ any policy induced by $u^k$ is optimal on all games $\dgraph_{\ell}$. Such a policy is also optimal in the game played on $\dgraph$ by \cref{game_decomposition}. This shows the first claim.

To show the second claim, suppose that $\payoff \in \polyh^{\sigma,\tau}$ for some $(\sigma,\tau)$. Then, $\gameval = \gameval^{\sigma,\tau}(\payoff)$ by \cref{cone_gives_bias}. Denote $\bias = \bias^{\sigma,\tau}(\payoff)$. By \cref{constant_value_subgame} the games on graphs $\dgraph_{\ell}$ have a single equilibrium cycle, and $(\sigma,\tau)$ is the unique pair of bias induced policies on $\dgraph_{\ell}$. Fix any $\ell \in [p]$, let $\tilde{\shapley}$ be the Shapley operator of the game on $\dgraph_{\ell}$, and let $\tilde{u}^q \in \R^{\vertices_\ell}$ be the restricted version of the vector $u^k$, i.e., $\tilde{u}^q_i = u^{q+k^*}_i$ for $i \in \vertices_{\ell}$ and $q \ge 0$. If every vertex of $\dgraph_{\ell}$ has out-degree $1$, then it has a single pair of policies, so for every $q \ge 1$ any policy induced by $\tilde{u}^{q}$ is optimal. Otherwise, let
\begin{align*}
\varepsilon &\coloneqq \frac{1}{4n+3}\sk(\payoff)  \le \frac{1}{4n+3} \min\{\abs{\payoff_{ij} - \gameval + \bias_j - \bias_i} \colon (i,j) \in \edges \cap (\vertices_\ell \times \vertices_{\ell}), \payoff_{ij} - \gameval + \bias_j - \bias_i \neq 0\} \, .
\end{align*}
For every $q \ge 1$ by \cref{halpern_for_bias} we have
\[
\hilb{\tilde{u}^q - \tilde{\shapley}(\tilde{u}^q)} \le \frac{3}{\sqrt{q}}\hilb{\tilde{u}^{0} - \tilde{u}^*} \le \frac{3}{\sqrt{q}}\hilb{\tilde{u}^*} + \frac{3}{\sqrt{q}}\hilb{\tilde{u}^{0}}\, ,
\]
where $\tilde{u}^*$ is any bias for the game on $\dgraph_{\ell}$. By \cref{short_bias} there exists $\tilde{\bias}^*$ such that $\hilb{\tilde{u}^*} \le 2n\hilb{\payoff}$. Furthermore, by applying \cref{km_estimate_hilbert} to $u^{k^*}$ we get $k^* \frac{\gameval_i}{2} - n\hilb{\payoff} \le \tilde{u}^0_i \le k^* \frac{\gameval_i}{2} + n\hilb{\payoff}$ for all $i \in \vertices_{\ell}$. Since $\gameval_i$ is constant within $\vertices_{\ell}$, we get $\hilb{\tilde{u}^0} \le 2n\hilb{\payoff}$. Hence, we have
\[
\hilb{\tilde{u}^q - \tilde{\shapley}(\tilde{u}^q)} \le \frac{12n\hilb{\payoff}}{\sqrt{q}}
\]
for all $q \ge 1$. In particular, for $q \ge (\frac{12n\hilb{\payoff}}{\varepsilon})^2 = 144(4n+3)^2n^2\Cond(\payoff)^2$ we have $\hilb{\tilde{u}^q - \tilde{\shapley}(\tilde{u}^q)} \le \varepsilon$. By \cref{approx_bias_policies}, $(\sigma,\tau)$ is the only pair of policies induced by $\tilde{u}^q$ on $\dgraph_{\ell}$. Since this is true for all $\ell$, we get by \cref{value_condition_of_km} that $(\sigma,\tau)$ is the only pair of policies induced by $u^{k}$ for all $k \ge 1+ k^* + 144(4n+3)^2n^2\Cond(\payoff)^2$. Since $\Cond^{\val}(\payoff) \le \Cond(\payoff)$ (\cref{value_cond_bound}), we get the claim.
\end{proof}

\begin{remark}\label{overall_complexity_km}
To bound the complexity of the \KM iteration, note that each iteration can of the algorithm can be implemented in $O(m)$ complexity. Furthermore, one can check if a pair of policies is optimal in $O(mn)$ complexity by \cite[Theorem~4.1]{zwick_paterson}. One can implement the algorithm by checking if policies induced by $u^k$ are optimal only every $n$ iterations. In this way, the average cost of every iteration is still $O(m)$ and the overall complexity of the algorithm is $O(mn + mn^2\Cond(\payoff) + mn^4\Cond(\payoff)^2)$ for $\payoff \in \Generic$. \end{remark}

\subsection{Halpern iteration}
To prove the analogous result for the Halpern iteration, we proceed as in the case of \KM iteration. However, we need to take into account the fact that the sequence $\tilde{u}^k$ from \cref{value_condition_of_halpern} on the reduced game comes not from the Halpern iteration, but from the Halpern iteration with shifted weights. To do, so we use the analogue of \cref{halpern} for this kind of the Halpern iteration.

\begin{proposition}\label{shifted_halpern}
Let $(\nspace, \| \cdot \|)$ be a normed space and let $f \colon \nspace \to \nspace$ be a non-expansive function. Furthermore, suppose that this function has at least one fixed point $x^{*} \in \nspace$. Let $p \ge 3$, fix a pair of points $x^0,y^0 \in \nspace$, and consider the Halpern sequence with shifted weights $(x^k)_{k \ge 0}$ given by $x^{k} = \frac{2}{k+p+2}y^{0} + \frac{k+p}{k+p+2}f(x^{k-1})$ for $k \ge 1$. Then, for every $k \ge 0$ we have
\[
\| x^k - f(x^k)\| \le \frac{2p+2}{k+p+1}\max\{\|x^{0} - x^*\|, \|y^{0} - x^*\|\} \, .
\]
\end{proposition}
\begin{proof}
The claim follows from \cite[Proposition 4]{ContrerasCominetti:2023}. More precisely, using the notation from \cite{ContrerasCominetti:2023}, if we consider the sequence $\beta_{n} = \frac{n+p}{n+p+2}$ for $n \ge 1$ and we put $a = p$, $\kappa = p + 1 \ge 4$, then
\[
(1 - \beta_n)^2 + \frac{\kappa}{n+a}\beta_n \le \frac{\kappa}{n + a + 1} \, .
\]
Indeed, we have
\begin{align*}
&(1 - \beta_n)^2 + \frac{\kappa}{n+a}\beta_n \le \frac{\kappa}{n + a + 1} \iff (1 - \frac{n+p}{n+p+2})^2 + \frac{p+1}{n+p} \cdot \frac{n+p}{n+p+2} \le \frac{p+1}{n+p+1} \\
&\iff \frac{4}{(n+p+2)^2} \le \frac{p+1}{(n+p+1)(n+p+2)} \, ,
\end{align*}
and the last expression is trivially true. Hence, \cite[Proposition 4]{ContrerasCominetti:2023} shows that for all $n \ge 0$ we have $\| x^n - f(x^n)\| \le \frac{2\kappa}{n+a+1}\max\{\|x^{0} - x^*\|, \|y^{0} - x^*\|\}$, which gives the claim.
\end{proof}
As in \cref{halpern_for_bias}, the result above shows that shifted Halpern iteration can be used to find an approximate bias of a constant value game.

\begin{corollary}\label{shifted_halpern_for_bias}
Let $\shapley$ be a Shapley operator of a mean payoff game with constant value and let $\bias^*$ be any bias of this game. Let $p \ge 3$, fix a pair of points $u^0, w^0 \in \R^n$, and consider the Halpern sequence with shifted weights $(\bias^k)_{k \ge 0}$ given by $u^k = \frac{2}{k+p+2}w^{0} + \frac{k+p}{k+p+2}\shapley(u^{k-1})$ for $k \ge 1$. Then, for every $k \ge 0$ we have
\[
\hilb{u^k - \shapley(u^k)} \le \frac{2p+2}{k+p+1}\max\{\hilb{u^{0} - u^*},\hilb{w^{0} - u^*}\}\, .
\]
In particular, for any sufficiently large $k$, all policies induced by $\bias^k$ are optimal.
\end{corollary}
\begin{proof}
The claim follows from \cref{shifted_halpern} in the same way as \cref{halpern_for_bias} follows from \cref{halpern}.
\end{proof}

The results above now enable us to prove an analogue of \cref{overall_complexity_KM} for the Halpern iteration.

\begin{theorem}\label{overall_complexity_Halpern}
Consider the Halpern iteration initialized at zero, $u^0 \coloneqq 0$, $u^{k} \coloneqq \frac{k}{k+2}\shapley(u^{k-1})$. Then, for all sufficiently large $k$, every policy induced by $u^k$ is optimal. Furthermore, if $\payoff \in \Generic$, then the above is true for all
\[
k \ge 4 + (192n^2+156n+3)\Cond(\payoff) + 36n(3 + 16n + 16n^2)\Cond(\payoff)^2 \, .
\]
\end{theorem}
\begin{proof}
We repeat the proof of \cref{overall_complexity_KM} with slight changes. Let $\Cond^{\val}(\payoff)$ be the value condition number of the game. Let $k^* \coloneqq \lfloor 3 + (12n+3)\Cond^{\val}(\payoff) \rfloor$. By \cref{value_condition_of_halpern}, for all $k \ge k^*$, the Halpern iteration ignores all of the edges between vertices of different values. Denote by $\gameval$ the value of the game, and let $\edges' \coloneqq \edges \setminus \{(i,j) \in \edges, \gameval_i \neq \gameval_{j}\}$. Consider the reduced graph $\dgraph' = (\vertices, \edges')$ and let $\dgraph_1, \dots, \dgraph_p$ be the connected components of $\dgraph'$ with vertex sets $\vertices_{1}, \dots, \vertices_{p}$. For $k \ge k^*$, the Halpern iteration on $\dgraph$ behaves in the same way as running $p$ Halpern iterations with shifted weights on every $\dgraph_{\ell}$ in parallel starting from $(u^{k^*}_i)_{i \in \vertices_{\ell}}$ (and with anchor equal to $0$). Since each game on $\dgraph_{\ell}$ has constant value by \cref{game_decomposition}, \cref{shifted_halpern_for_bias} shows that for all sufficiently large $k$ any policy induced by $u^k$ is optimal on all games $\dgraph_{\ell}$. Such a policy is also optimal in the game played on $\dgraph$ by \cref{game_decomposition}. This shows the first claim.

To show the second claim, suppose that $\payoff \in \polyh^{\sigma,\tau}$ for some $(\sigma,\tau)$. Then, $\gameval = \gameval^{\sigma,\tau}(\payoff)$ by \cref{cone_gives_bias}. Denote $\bias = \bias^{\sigma,\tau}(\payoff)$. By \cref{constant_value_subgame} the games on graphs $\dgraph_{\ell}$ have a single equilibrium cycle, and $(\sigma,\tau)$ is the unique pair of bias induced policies on $\dgraph_{\ell}$. Fix any $\ell \in [p]$, let $\tilde{\shapley}$ be the Shapley operator of the game on $\dgraph_{\ell}$, and let $\tilde{u}^q \in \R^{\vertices_\ell}$ be the restricted version of the vector $u^k$, i.e., $\tilde{u}^q_i = u^{q+k^*}_i$ for $i \in \vertices_{\ell}$ and $q \ge 0$. If every vertex $\dgraph_{\ell}$ has out-degree $1$, then it has a single pair of policies, so for every $q \ge 1$ any policy induced by $\tilde{u}^{q}$ is optimal. Otherwise, let
\begin{align*}
\varepsilon &\coloneqq \frac{1}{4n+3}\sk(\payoff)  \le \frac{1}{4n+3} \min\{\abs{\payoff_{ij} - \gameval + \bias_j - \bias_i} \colon (i,j) \in \edges \cap (\vertices_\ell \times \vertices_{\ell}), \payoff_{ij} - \gameval + \bias_j - \bias_i \neq 0\} \, .
\end{align*}
For every $q \ge 0$ by \cref{shifted_halpern_for_bias} we have
\[
\hilb{\tilde{u}^q - \tilde{\shapley}(\tilde{u}^q)} \le \frac{2k^*+2}{q+k^*+1}\max\{\hilb{\tilde{u}^{0} - \tilde{u}^*}, \hilb{\tilde{u}^*} \} \le \frac{2k^*+2}{q+k^*+1}\hilb{\tilde{u}^*} + \frac{2k^*+2}{q+k^*+1}\hilb{\tilde{u}^{0}}\, ,
\]
where $\tilde{u}^*$ is any bias for the game on $\dgraph_{\ell}$. By \cref{short_bias} there exists $\tilde{\bias}^*$ such that $\hilb{\tilde{u}^*} \le 2n\hilb{\payoff}$. Furthermore, by applying \cref{halpern_estimate_hilbert} to $u^{k^*}$ we get $k^* \frac{\gameval_i}{3} - 2n\hilb{\payoff} \le \tilde{u}^0_i \le k^* \frac{\gameval_i}{3} + 2n\hilb{\payoff}$ for all $i \in \vertices_{\ell}$. Since $\gameval_i$ is constant within $\vertices_{\ell}$, we get $\hilb{\tilde{u}^0} \le 4n\hilb{\payoff}$. Hence, we have
\[
\hilb{\tilde{u}^q - \tilde{\shapley}(\tilde{u}^q)} \le \frac{12n(k^*+1)\hilb{\payoff}}{q+k^*+1}
\]
for all $q \ge 0$. In particular, for $q \ge \frac{12n(k^*+1)\hilb{\payoff}}{\varepsilon} - 1 - k^* = 12n(4n+3)(1+k^*)\Cond(\payoff) - 1 - k^*$ we have $\hilb{\tilde{u}^q - \tilde{\shapley}(\tilde{u}^q)} \le \varepsilon$. By \cref{approx_bias_policies}, $(\sigma,\tau)$ is the only pair of policies induced by $\tilde{u}^q$ on $\dgraph_{\ell}$. Since this is true for all $\ell$, by \cref{value_condition_of_halpern} we get that $(\sigma,\tau)$ is the only pair of policies induced by $u^{k}$ for all $k \ge \max\{1+k^*,12n(4n+3)(1+k^*)\Cond(\payoff)\}$. In particular, this is true for all $k \ge 1+k^* + 12n(4n+3)(1+k^*)\Cond(\payoff)$. Since $k^* \le 3 + (12n+3)\Cond(\payoff)$, the claim is true for all $k \ge 1 + 3 + (12n+3)\Cond(\payoff) + 12n(4n+3)\bigl(1 + 3 + (12n+3)\Cond(\payoff)\bigr)\Cond(\payoff)$, which is the bound given in the statement of the theorem.
\end{proof}

\begin{remark}
Similarly to the case of \KM iteration (\cref{overall_complexity_km}), Halpern iteration can be implemented to run in $O(mn + mn^2\Cond(\payoff) + mn^3 \Cond(\payoff)^2)$ complexity for all $\payoff \in \Generic$.
\end{remark}

\subsection{Combining the Zwick--Paterson iteration with the Halpern iteration}\label{sec:zph_vi}
The analysis from the previous section has a weakness: even though the Halpern iteration achieves better complexity bounds than the \KM iteration on constant value games (\cref{halpern_for_bias}), our analysis does not extend this result to arbitrary games. Indeed, the bounds that we obtained for both types of iterations are quadratic in the condition number. The problem arises from the fact that the sequence from \cref{value_condition_of_halpern} on the reduced game comes from the Halpern iteration with shifted weights. To avoid this problem, one could proceed in two stages: first find the reduced game and then run the Halpern iteration directly on the reduced game. This is the idea behind the value iteration algorithm given in \cref{alg:zph_vi}. This algorithm uses the Zwick--Paterson iteration to find the reduced game, and then applies the Halpern iteration on this reduced game. Since we are unable to check when the Zwick--Paterson iteration actually finds the reduced game, we run the two iterations in parallel, and reset the Halpern iteration each time the Zwick--Paterson iteration detects an edge between two vertices with different values.

\begin{figure}[t]
\centering
\begin{footnotesize}
  \begin{algorithmic}[1]
    \Procedure{$\ZPHVI$}{$\vertices,\edges,\payoff$}
    \State $u \gets 0, w \gets 0, k \gets 1$
    \State $A \gets \edges$
    \While{\texttt{True}}
         \State $w \gets \shapley(w)$ \Comment{$\shapley$ is the Shapley operator of the game on $(\vertices,\edges,\payoff)$ }
         \For {$(i,j) \in A$}
         	\If {$\abs{w_i - w_j} > 2n\hilb{\payoff}$} \Comment{The edge $(i,j)$ connects vertices of different values }
		\State $A \gets A \setminus \{(i,j)\}$ \Comment{Remove the edge $(i,j)$ from $A$ and reset the Halpern iteration}
		\State $u \gets 0$
		\State $k \gets 1$
         	\EndIf
         \EndFor
         \State\label{line:halpern} $u \gets \frac{k}{k+2}\shapley_A(u)$ \Comment{$\shapley_A$ is the Shapley operator of the game on $(\vertices,A,\payoff)$ }
         \State $k \gets k+1$
         \If {$n \mid k$}
         	\State $(\sigma,\tau) \gets$ any policies induced by $u$
		\If {$(\sigma,\tau)$ are optimal in the mean payoff game on $(\vertices,\edges,\payoff)$}
	 	\State \Return $(\sigma, \tau)$
		\EndIf
         \EndIf
      \EndWhile
\EndProcedure
\end{algorithmic}
\end{footnotesize}
\caption{Combining the Zwick--Paterson iteration with the Halpern iteration.}
\label{alg:zph_vi}
\end{figure}

\begin{theorem}\label{overall_complexity_ZPH}
The algorithm $\ZPHVI$ halts for every input and is correct. Furthermore, if $\payoff \in \Generic$, then the algorithm performs $O\bigl(n+ n^2 \Cond(\payoff)\bigr)$ iterations of the main loop.
\end{theorem}
\begin{proof}
The algorithm can only halt when it finds a pair of optimal policies. Therefore, it is correct whenever it halts. To prove that it halts on every input, suppose that it performs the main loop more than $1 + 4n\Cond^{\val}(\payoff)$ times. Then, \cref{condition_of_zp} shows that $A = \edges \setminus \{(i,j) \colon \gameval_i \neq \gameval_j\}$. In particular, the set $A$ does not change when more iterations of the main loop are performed. Let $t \le 2 + 4n\Cond^{\val}(\payoff)$ be the first iteration of the main loop such that $A = \edges \setminus \{(i,j) \colon \gameval_i \neq \gameval_j\}$ at \cref{line:halpern} of the algorithm. Then, we get $u = \frac{1}{3}\shapley_A(0)$. From this point onwards, the sequence of vectors $u$ evolves exactly as the Halpern iteration sequence performed on the reduced game on $\dgraph' = (\vertices, A)$. Let $\dgraph_1, \dots, \dgraph_p$ be the connected components of $\dgraph'$. The Halpern iteration on $\dgraph'$ behaves in the same way as running $p$ Halpern iterations on every $\dgraph_{\ell}$ in parallel starting from $0$. Hence, it will eventually output a pair of optimal policies by \cref{game_decomposition} combined with \cref{halpern_for_bias}. If we further assume that $\payoff \in \Generic$, then \cref{km_for_one_cycle} combined with \cref{decomposition_condition_upper} shows that $\ZPHVI$ will output a pair of optimal policies within $O(t + n + n^2 \Cond(\payoff))$ iterations of the main loop. Hence, the inequality $\Cond^{\val}(\payoff) \le \Cond(\payoff)$ (\cref{value_cond_bound}) finishes the proof.
\end{proof}

\begin{remark}
As in the case of \KM iteration (\cref{overall_complexity_km}), the average cost of each loop of the algorithm $\ZPHVI$ is $O(m)$, giving a total complexity of $O(mn + mn^2 \Cond(\payoff))$ for $\payoff \in \Generic$.
\end{remark}

\subsection{Policy iteration methods based on increasing discount factor}

In \cite{LoffSkomra:2024}, the condition number defined for mean payoff games played on ergodic graphs is used to bound the complexity of a certain policy iteration method that is based on approximating the mean payoff game by a sequence of discounted games with increasing discount factor. This policy iteration method is capable of solving both mean payoff games and discounted games. We now show that the same statement is true for our condition number on arbitrary graphs, namely that the complexity of the policy iteration method from \cite{LoffSkomra:2024} on any graph is bounded by the condition number that we defined earlier. This is based on the following theorem, which is a direct generalization of \cite[Theorem~9]{LoffSkomra:2024}.

\begin{theorem}\label{conditioned_discount}
Suppose that the graph of the game has at least one vertex with out-degree at least $2$. Moreover, suppose that $\payoff \in \polyh^{\sigma,\tau}$ for some pair of policies $(\sigma,\tau)$ and fix $1 > \disc > 1 - \frac{1}{4n^2\Cond(\payoff)}$. Then, $(\sigma,\tau)$ is the unique pair of optimal policies in the discounted game with discount factor $\disc$.
\end{theorem}

As in \cite{LoffSkomra:2024}, the proof of \cref{conditioned_discount} relies on the characterization of the value of discounted games from \cite{shapley_stochastic} and on the estimate of the difference between the value of a discounted game and the value of the mean payoff game from \cite{zwick_paterson}.

\begin{theorem}[\cite{shapley_stochastic}]\label{th:shapley}
Both players in a discounted game have optimal policies. More precisely, the value $\gameval^{(\disc)} \in \R^n$ of the discounted game with discount factor $\disc$ is the unique solution of the equations
\begin{align*}
\gameval^{(\disc)}_i \coloneqq \begin{cases}
\max_{(i,j) \in \edges}\{(1 - \disc)\payoff_{ij} + \disc \gameval^{(\disc)}_j\} &\text{if $i \in \Maxvertices$},\\[0.1cm]
\min_{(i,j) \in \edges}\{(1 - \disc)\payoff_{ij} + \disc \gameval^{(\disc)}_j\} &\text{if $i \in \Minvertices$}.
\end{cases}
\end{align*}
A policy of Max (resp. Min) is optimal if and only if it uses edges that achieve the maxima (resp. minima) in these equalities.
\end{theorem}

\begin{proposition}[{\cite[Theorem~5.2]{zwick_paterson}}]\label{discount_to_mean}
If $\gameval^{(\disc)}$ denotes the value of a discounted game with discount factor $\disc$ and $\gameval$ denotes the value of the mean payoff game with the same weights, then $\supnorm{\gameval^{(\disc)} -  \gameval} \le 2n(1 - \disc)\supnorm{\payoff}$.
\end{proposition}
We note that in the results above we use a convention in which the value of a discounted game is normalized, so that $\supnorm{\gameval^{(\disc)}} \le \supnorm{\payoff}$ for every $\disc$. Some authors (e.g., \cite{kallenberg2002handbook}) prefer to work with the un-normalized value $\frac{\gameval^{(\disc)}}{1 - \disc}$.

\begin{proof}[Proof of \cref{conditioned_discount}]
We follow the proof of \cite[Theorem~6.1]{LoffSkomra:2024}, with modifications that extend it to general graphs. Let $\alpha$ be such that $\hilb{\payoff} = 2 \supnorm{\payoff + \alpha \ones}$. Consider the modified game with weights $\hat{\payoff}_{ij} \coloneqq \payoff_{ij} + \alpha$, so that $\supnorm{\hat{\payoff}} = \frac{1}{2}\hilb{\payoff}$. Note that this modification does not change the optimal policies, neither in the mean payoff game nor in the discounted game. Moreover, if $\hat{\gameval}$ denotes the value of the mean payoff game with weights $\hat{\payoff}$ and $\gameval$ denotes the value of the mean payoff game with weights $\payoff$, then $\hat{\gameval}_i = \alpha + \gameval_i$ for all $i$. The same is true for the values $\hat{\gameval}^{(\disc)}$ and $\gameval^{(\disc)}$ of the discounted games with weights $\hat{\payoff}$ and $\payoff$. Furthermore, we have $\gameval^{\sigma,\tau}_i(\hat{\payoff}) = \alpha + \gameval^{\sigma,\tau}_i(\payoff)$ and $\bias^{\sigma,\tau}_i(\hat{\payoff}) = \bias^{\sigma,\tau}_i(\payoff)$ for all $i \in [n]$. In particular, $\hat{\payoff} \in \polyh^{\sigma,\tau}$ and $\sk(\hat{\payoff}) = \sk(\payoff)$. Hence, we have $\Cond(\payoff) = \frac{2\supnorm{\hat{\payoff}}}{\sk(\hat{\payoff})}$.

Suppose that player Max plays according to $\sigma$ in the discounted game. Let $\tau'$ be a policy of Min that is an optimal response to $\sigma$ in the discounted game. We will show that $\tau' = \tau$. To do so, suppose that $\tau'(i) \neq \tau(i)$ for some $i \in \Minvertices$ and consider three cases.

First case: we have $\vertices^{\sigma,\tau}(\tau'(i)) \neq \vertices^{\sigma,\tau}(i)$. By the definition of slack and \cref{cone_gives_bias}, we have $\hat{\gameval}_{\tau'(i)} \ge \sk(\hat{\payoff}) + \hat{\gameval}_{i}$. Let $\chi$ be the mean weight of the cycle reached from $i$ in $\dgraph^{\sigma,\tau'}$. Then, $\chi \ge \sk(\hat{\payoff}) + \hat{\gameval}_{i}$ because $\hat{\gameval}_{\tau'(i)} \ge \sk(\hat{\payoff}) + \hat{\gameval}_{i}$ and $\sigma$ is optimal in the mean payoff game. Let $\xi$ denote the value of the zero-player discounted game played on $\dgraph^{\sigma,\tau'}$ and starting from $i$. By applying \cref{discount_to_mean} to this zero-player game, we get $\chi \le 2n(1-\disc)\supnorm{\hat{\payoff}} + \xi$. Furthermore, we have $\xi \le \hat{\gameval}^{(\disc)}_i$ because $\tau'$ is an optimal response to some policy $\sigma$. Hence $\chi \le 2n(1-\disc)\supnorm{\hat{\payoff}} + \hat{\gameval}^{(\disc)}_i$. Moreover, by applying \cref{discount_to_mean} to the entire two-player game, we get $\hat{\gameval}^{(\disc)}_i \le 2n(1-\disc)\supnorm{\hat{\payoff}} + \hat{\gameval}_i$. Therefore $\sk(\hat{\payoff}) + \hat{\gameval}_{i} \le \chi \le 4n(1-\disc)\supnorm{\hat{\payoff}} + \hat{\gameval}_i$ and $1 - \disc \ge \frac{\sk(\hat{\payoff})}{4n\supnorm{\hat{\payoff}}} = \frac{1}{2n\Cond(\payoff)}$, which gives a contradiction.

Second case: the vertex $i$ is on a cycle in the graph $\dgraph^{\sigma,\tau'}$. Let $i = i_0, i_1, \dots, i_{l-1}$ denote the cycle in $\dgraph^{\sigma,\tau'}$ that contains $i$, with the convention that $i = i_{l}$. Let $\chi$ denote the mean weight of this cycle. By the first case, we have $\vertices^{\sigma,\tau}(i_{q+1}) = \vertices^{\sigma,\tau}(i)$ for every $q$. Therefore, for every $q$ such that $i_q \in \Maxvertices$ we have $i_{q+1} = \sigma(i_q)$ and $\gameval^{\sigma,\tau}_{i_q}(\hat{\payoff}) + \bias^{\sigma,\tau}_{i_q}(\hat{\payoff}) = \hat{\payoff}_{i_q i_{q+1}} + \bias^{\sigma,\tau}_{i_{q+1}}(\hat{\payoff})$. Likewise, for every $q$ such that $i_q \in \Minvertices$ and $i_{q+1} = \tau(i_q)$, we have $\gameval^{\sigma,\tau}_{i_q}(\hat{\payoff}) + \bias^{\sigma,\tau}_{i_q}(\hat{\payoff}) = \hat{\payoff}_{i_q i_{q+1}} + \bias^{\sigma,\tau}_{i_{q+1}}(\hat{\payoff})$. However, by the definition of slack, for every $i_q \in \Minvertices$ such that $i_{q+1} \neq \tau(i_q)$, we have $\gameval^{\sigma,\tau}_{i_q}(\hat{\payoff}) + \bias^{\sigma,\tau}_{i_q}(\hat{\payoff}) \le -\sk(\hat{\payoff}) + \hat{\payoff}_{i_q i_{q+1}} + \bias^{\sigma,\tau}_{i_{q+1}}(\hat{\payoff})$. Furthermore, we have $\gameval_{i_q} ^{\sigma,\tau}(\hat{\payoff}) = \gameval^{\sigma,\tau}_i(\hat{\payoff}) = \hat{\gameval}_i$ for all $q$. Since $i_0$ is such that $i_1 \neq \tau(i_0)$ by adding these inequalities for all $q$ we get $l \hat{\gameval}_i \le -\sk(\hat{\payoff}) + \sum_{q = 0}^{l-1}\hat{\payoff}_{i_qi_{q+1}}$. Therefore $\chi \ge \frac{1}{n}\sk(\hat{\payoff}) + \hat{\gameval}_{i}$. On the other hand, the same reasoning as in the first case gives the inequality $\chi \le 4n(1-\disc)\supnorm{\hat{\payoff}} + \hat{\gameval}_i$. Hence $1 - \disc \ge \frac{1}{2n^2\Cond(\payoff)}$, which gives a contradiction.

Third case: the vertex $i$ is not on a cycle in the graph $\dgraph^{\sigma,\tau'}$. Let $j_0, j_1, \dots, j_{l-1}$ denote the cycle in $\dgraph^{\sigma,\tau'}$ that is reached from $i$, where $j_0$ is the vertex with the smallest index of this cycle. By the second case, the cycle $j_0, j_1, \dots, j_{l-1}$ also appears in $\dgraph^{\sigma,\tau}$. Furthermore, if $i = i'_0, \dots, i'_{q'} = j_0$ denotes the path from $i$ to $j_0$ in $\dgraph^{\sigma,\tau'}$, then by applying the first case on the vertices along this path we get that $\vertices^{\sigma,\tau}(i) = \vertices^{\sigma,\tau}(j_0)$. Hence, the cycle $j_0, j_1, \dots, j_{l-1}$ is also reached from $i$ in the graph $\dgraph^{\sigma,\tau}$. Let $i = i_0, \dots, i_q = j_0$ be the path from $i$ to $j_0$ in $\dgraph^{\sigma,\tau}$. Let $\xi$ denote the value of $j_0$ in the discounted zero-player game played on $\dgraph^{\sigma,\tau'}$, which is the same as the value of $j_0$ in the discounted zero-player game played on $\dgraph^{\sigma,\tau}$. Note that all the vertices on the paths $i_0,\dots,i_q$ and $i'_0,\dots,i'_{q'}$ have value $\hat{\gameval}_i$ in the mean payoff game. Denote $\bias_j = \bias^{\sigma,\tau}_j(\hat{\payoff})$ for all $j \in [n]$. Note that, by our choice of $j_0$, we have $\bias_{j_0} = 0$. The definition of $\bias^{\sigma,\tau}(\hat{\payoff})$ and \cref{th:shapley} imply that the value of $i$ in the discounted game played on $\dgraph^{\sigma,\tau}$ is equal to
\begin{align*}
&(1 - \disc)\bigl(\hat{\payoff}_{i_0i_1} + \disc \hat{\payoff}_{i_1i_2} + \dots + \disc^{q-1}\hat{\payoff}_{i_{q-1}i_q}\bigr) + \disc^q \xi \\
&=(1 - \disc)\bigl((\hat{\gameval}_i + \bias_{i_0} - \bias_{i_1} ) + \dots + \disc^{q-1}(\hat{\gameval}_i + \bias_{i_{q-1}}- \bias_{i_q} ) \bigr) + \disc^q \xi \\
&= \hat{\gameval}_i + \disc^q (\xi - \hat{\gameval}_i)  + (1-\disc)\bias_i -(1 - \disc)^2\sum_{s=1}^{q-1}\disc^{s-1}\bias_{i_s} \, .
\end{align*}
Moreover, as in the second case, for every $s$ we have $\hat{\payoff}_{i'_s{i'_{s+1}}} \ge \hat{\gameval}_i + \bias_{i'_s} - \bias_{i'_{s+1}}$ and $\hat{\payoff}_{i'_0i'_1} \ge \sk(\hat{\payoff}) + \hat{\gameval}_i + \bias_{i'_0} - \bias_{i'_1}$. Hence, the value of $i$ in the discounted game played on $\dgraph^{\sigma,\tau'}$ is bounded from below by
\[
(1-\disc)\sk(\hat{\payoff}) + \hat{\gameval}_i + \disc^{q'}( \xi - \hat{\gameval}_i) + (1-\disc)\bias_i -(1 - \disc)^2\sum_{s=1}^{q'-1}\disc^{s-1}\bias_{i'_s} \, .
\]
Therefore (dividing by $1 - \disc$ and subtracting the term $\frac{1}{1-\disc} (\xi - \hat{\gameval}_i)$) to get a contradiction it is enough to prove that
\begin{equation}\label{eq:disc_on_path}
\sk(\hat{\payoff}) + \frac{\disc^{q'} -1}{1 - \disc} (\xi - \hat{\gameval}_i) - (1 - \disc)\sum_{s = 1}^{q' - 1}\disc^{s-1}\bias_{i'_{s}}  > \frac{\disc^q - 1}{1 - \disc} (\xi - \hat{\gameval}_i) - (1 - \disc)\sum_{s = 1}^{q - 1}\disc^{s-1}\bias_{i_{s}} \, .
\end{equation}
To do so, note that by applying \cref{discount_to_mean} to the zero-player game on $\dgraph^{\sigma,\tau'}$ we get $\abs{\xi - \hat{\gameval}_i} \le 2n(1-\disc)\supnorm{\hat{\payoff}}$. Furthermore, $\abs{\frac{\disc^q - 1}{1 - \disc}} = 1 + \disc + \dots + \disc^{q-1} < n$ and $\supnorm{\bias} \le 2n\supnorm{\hat{\payoff}}$ by \cref{one_cycle_graph}. Hence,
\begin{align*}
\left\lvert\frac{\disc^q - 1}{1 - \disc}(\xi - \hat{\gameval}_i) - (1 - \disc)\sum_{s= 1}^{q - 1}\disc^{s-1}\bias_{i_{s}}\right\rvert &\le 2n^2(1-\disc)\supnorm{\hat{\payoff}} + 2n^2(1 - \disc)\supnorm{\hat{\payoff}} \\
&= 4n^2(1-\disc)\supnorm{\hat{\payoff}} < \frac{1}{2}\sk(\hat{\payoff}) \, .
\end{align*}
The same is true if we replace $q$ by $q'$ and so \cref{eq:disc_on_path} holds, giving the desired contradiction.

Therefore, we have shown that $\tau$ is the unique best response policy to $\sigma$ in the discounted game. Analogously, $\sigma$ is the unique best response policy to $\tau$. This shows that $(\sigma,\tau)$ are optimal. Moreover, they are the only pair of optimal policies. Indeed, if we denote by $\zeta$ the value of the discounted game, then any optimal policy $\tau'$ of Min must guarantee a payoff at most $\zeta$ against $\sigma$, so $\tau'$ is not worse than $\tau$ against $\sigma$. This implies that $\tau' = \tau$ because $\tau$ is the unique best response to $\sigma$. The same argument holds for $\sigma$.
\end{proof}

As a corollary, the algorithms based on policy iteration with increasing discount factor discussed in \cite{LoffSkomra:2024} have complexity that is bounded by the condition number $\Cond(\payoff)$. We give slightly simplified versions of these algorithms in \cref{alg:increasing_discount_mpg} and \cref{alg:increasing_discount_disc}. The first one ($\IncDiscPI$) solves mean payoff games, while the second ($\IncDiscPItwo$) solves discounted games. In both cases, $\DiscPI(\sigma,\tau,\disc)$ is any procedure that takes as an input a discounted game with discount factor $\disc$ and a pair of policies $(\sigma,\tau)$ and outputs a pair of optimal policies. One way to implement $\DiscPI(\sigma,\tau,\disc)$ is to use policy iteration algorithm equipped with the Howard greedy all-switches rule. As discussed in \cite{LoffSkomra:2024}, this choice leads to a $\DiscPI(\sigma,\tau,\disc)$ method that performs $O\bigl((\frac{m}{1- \disc})^2 \log^2 \frac{1}{1 - \disc}\bigr)$ switches for both players (and each switch can be implemented in $O(m)$ complexity). This result follows from the analysis of the Howard rule presented in \cite{ye2011simplex,hansen2013strategy,AkianGaubert:2013}. Another choice could be to use the Dantzig rule analyzed in \cite{post_ye,HansenKaplanZwick:2014,JiaWenYe:2020}. The advantage of the Dantzig rule lies in the fact that this rule halts in $\poly(n,m)$ time on one-player deterministic games, independently of the discount factor. Hence, the overall complexity of the Dantzig rule is only linear in $\frac{1}{1 - \disc}$ and not quadratic. We summarize the two choices in the following theorems.

\begin{theorem}[\cite{AkianGaubert:2013}]\label{discount_iter_bound}
The policy iteration algorithm equipped with the Howard rule performs $O\bigl((\frac{m}{1- \disc})^2 \log^2 \frac{1}{1 - \disc}\bigr)$ switches to solve a deterministic discounted game.
\end{theorem}

\begin{theorem}[{Combine \cite[Theorem~4.2]{JiaWenYe:2020} and \cite[Theorem~4.1]{HansenKaplanZwick:2014}}]\label{dantzig_iter_bound}
The policy iteration algorithm equipped with the Dantzig rule performs $O\bigl(\frac{m^3n^3 \log^2 n}{1- \disc} \log \frac{n}{1 - \disc}\bigr)$ switches to solve a deterministic discounted game.
\end{theorem}

We note that policy iteration algorithms require access to a subroutine that solves one-player games. In the theorems above we suppose that these one-player games are solved by the same policy iteration rule as the rule used for the two-player game. By relaxing this assumption, one can propose algorithms that solve deterministic discounted games with better complexity bounds. For instance, one can combine the fact that one-player discounted games can be solved in $O(mn)$ complexity \cite{madani2010discounted} with the Howard rule to obtain a policy iteration algorithm that performs switches only on one player's policy and solves such games in $O\bigl(\frac{m}{1- \disc} \log \frac{1}{1 - \disc}\bigr)$ switches, leading to an algorithm whose overall complexity is $O\bigl(\frac{m^2n}{1- \disc} \log \frac{1}{1 - \disc}\bigr)$.

\begin{figure}[t]
\centering
\begin{footnotesize}
  \begin{algorithmic}[1]
    \Procedure{$\IncDiscPI$}{$\sigma,\tau$} \Comment{$(\sigma,\tau)$ are any initial policies}
    \State $\disc \gets 0$
    \Repeat
         \State $\disc \gets (1 + \disc)/2$
         \State $(\sigma, \tau) \gets \DiscPI(\sigma,\tau,\disc)$
      \Until $(\sigma, \tau)$ are optimal in the mean payoff game
      \State \Return $(\sigma, \tau)$
\EndProcedure
\end{algorithmic}
\end{footnotesize}
\caption{Policy iteration algorithm for mean-payoff games with polynomial smoothed complexity.}
\label{alg:increasing_discount_mpg}
\end{figure}

\begin{figure}[t]
\centering
\begin{footnotesize}
  \begin{algorithmic}[1]
    \Procedure{$\IncDiscPItwo$}{$\sigma,\tau,\bar{\disc}$} \Comment{$(\sigma,\tau)$ are any initial policies, $\bar{\disc}$ is the target discount}
    \State $\disc \gets 0$
    \Repeat
         \State $\disc \gets \min\{(1 + \disc)/2, \bar{\disc}\}$
         \State $(\sigma, \tau) \gets \DiscPI(\sigma,\tau,\disc)$
      \Until $(\sigma, \tau)$ are optimal in the discounted game with discount factor $\bar{\disc}$
      \State \Return $(\sigma, \tau)$
\EndProcedure
\end{algorithmic}
\end{footnotesize}
\caption{Policy iteration algorithm for discounted games with polynomial smoothed complexity.}
\label{alg:increasing_discount_disc}
\end{figure}

The following theorem bounds the complexity of both algorithms using our condition number. In the estimates, we count only the number of switches that the algorithms perform. As discussed above, a switch for both the Howard and Dantzig pivoting rules can be performed in $O(m)$ complexity. Furthermore, we note that checking if a given pair of policies is optimal in the mean payoff game can be done in $O(mn)$ complexity by \cite[Theorem~4.1]{zwick_paterson}, while doing the same for a discounted game can be done in $O(m)$ complexity by checking that none of the players has an improving switch.

\begin{theorem}\label{increasing_discount_condition}
The algorithm $\IncDiscPI$ halts for every input and is correct. Furthermore, if $\payoff \in \Generic$, then, putting $\Cond^{+}(\payoff) \coloneqq \max\{2,\Cond(\payoff)\}$, we have the following complexity estimates:
\begin{itemize}
\item if $\DiscPI$ is implemented by policy iteration using Howard's rule, then the algorithm $\IncDiscPI$ performs
\[
O\Bigl(m^2 + m^2n^4 \Cond(\payoff)^2 \log^3\bigl(n\Cond^{+}(\payoff)\bigr) \Bigr)
\]
switches for both players.
\item if $\DiscPI$ is implemented by policy iteration using Dantzig's rule, then the algorithm $\IncDiscPI$ performs
\[
O\Bigl(m^3n^3\log^2n + m^3n^5 \Cond(\payoff) \cdot \log^2 n \cdot \log^2\bigl(n\Cond^{+}(\payoff)\bigr) \Bigr)
\]
switches for both players.
\end{itemize}
The same statement holds for the algorithm $\IncDiscPItwo$.
\end{theorem}
\begin{proof}
We first prove the statement for $\IncDiscPI$. The fact that the algorithm halts for any input and is correct follows from \cref{discount_to_mean}. Indeed, denote by $\gameval^{(\disc)}$ the value of the discounted game with discount factor $\disc$, and by $\gameval$ the value of the mean payoff game. Suppose that $\sigma$ is an optimal policy of player Max in the discounted game and denote by $\gameval^\sigma$ the value of the one-player mean payoff game obtained by fixing $\sigma$. Then, by applying \cref{discount_to_mean} twice, to the two-player mean payoff game and to the one-player game with fixed $\sigma$, we get
\[
\supnorm{\gameval - \gameval^{\sigma}} \le \supnorm{\gameval - \gameval^{(\disc)}} + \supnorm{\gameval^{(\disc)} - \gameval^{\sigma}} \le 4n(1 - \disc)\supnorm{\payoff} \, .
\]
Hence, for any $\varepsilon > 0$ and $\disc$ sufficiently close to $1$, the policy $\sigma$ obtained in this way is $\varepsilon$-optimal in the mean payoff game. Since the number of policies is finite, there exists $\varepsilon > 0$ such that any $\varepsilon$-optimal policy in the mean payoff game is in fact optimal. Hence, for $\disc$ sufficiently close to $1$, the policy $\sigma$ obtained from solving the discounted game is optimal in the mean payoff game. The same holds for a policy of player Min. Therefore, the algorithm will halt on any input and outputs a pair of optimal policies in the mean payoff game.

To bound the number of switches performed by the algorithm, suppose that $\payoff \in \polyh^{\bar{\sigma},\bar{\tau}}$ for some pair of policies $(\bar{\sigma},\bar{\tau})$. If all vertices of the game graph have out-degree $1$, then the algorithm does not perform a single switch. Otherwise, by \cref{conditioned_discount}, if $1 - \disc < \frac{1}{4n^2\Cond(\payoff)}$, then $\DiscPI$ outputs $(\bar{\sigma},\bar{\tau})$, which are optimal policies by \cref{cone_gives_bias}. Since the quantity $1 - \disc$ is halved whenever $\IncDiscPI$ updates the discount factor, we get that the algorithm performs $O\bigl(\log(n\Cond^{+}(\payoff)) \bigr)$ updates of the discount factor. Furthermore, the procedure $\DiscPI$ is always called with a discount factor that satisfies $1 - \disc \ge \frac{1}{8n^2\Cond(\payoff)}$ (except, possibly, for the first call at which $\disc = \frac{1}{2}$).  Hence, the estimates on the number of switches follow from \cref{discount_iter_bound} and \cref{dantzig_iter_bound}.

The claim for $\IncDiscPItwo$ follows in an analogous way. This algorithm halts for every input and is correct because it may eventually reach the discount factor equal to $\bar{\disc}$. Moreover, if $\payoff \in \polyh^{\bar{\sigma},\bar{\tau}}$ for some pair of policies $(\bar{\sigma},\bar{\tau})$, then we have three cases. If $1 - \bar{\disc} \ge \frac{1}{2}$, then the algorithm simply solves the discounted game with discount factor $\bar{\disc}$. If $\frac{1}{2} > 1 - \bar{\disc} \ge \frac{1}{4n^2\Cond(\payoff)}$, then the algorithm performs $O\bigl(\log(n\Cond^{+}(\payoff)) \bigr)$ updates of the discount factor before reaching the discount factor equal to $\bar{\disc}$. Moreover, each call to $\DiscPI$ uses a discount factor that satisfies $1 - \disc \ge \frac{1}{4n^2\Cond(\payoff)}$. If $1 - \bar{\disc} < \frac{1}{4n^2\Cond(\payoff)}$, then $(\bar{\sigma},\bar{\tau})$ are optimal for the discounted game with discount factor $\bar{\disc}$ by  \cref{conditioned_discount}. Moreover, $\DiscPI$ will output $(\bar{\sigma},\bar{\tau})$ as soon as the discount factor satisfies $1 - \disc < \frac{1}{4n^2\Cond(\payoff)}$. Hence, as above, the algorithm performs $O\bigl(\log(n\Cond^{+}(\payoff)) \bigr)$ updates of the discount factor and each call to $\DiscPI$ uses a discount factor that satisfies $1 - \disc \ge \frac{1}{8n^2\Cond(\payoff)}$. Therefore, the claim follows from \cref{discount_iter_bound} and \cref{dantzig_iter_bound}.
\end{proof}

\section{Probabilistic analysis of the condition number}\label{sec:probabilistic_condition}

In this section, we provide a smoothed analysis of the condition number $\Cond(\payoff)$. Our probabilistic analysis follows the approach that is used in the literature on the smoothed analysis of the simplex method \cite{dadush2018friendly,HuibertsLeeZhang:2025,BachHuiberts:2025}. More precisely, we first analyze what happens if the weights of the game are drawn from a log-Lipschitz distribution. Subsequently, we generalize the result to the Gaussian distribution using the fact that Gaussian density is approximately log-Lipschitz. We start by defining log-Lipschitz densities. In the definition, $\| \cdot \|$ denotes any norm.

\begin{definition}[$L$-log-Lipschitz density]
Suppose that $\dens \colon \R^n \to \R$ is a strictly positive density function. Then, we say that $\dens$ is \emph{$L$-log-Lipschitz (with respect to the norm $\|\cdot\|$)} if the natural logarithm of $\dens$ is an $L$-Lipschitz function with respect to $\|\cdot\|$. In other words, $\dens$ is $L$-log-Lipschitz if for every $x,y \in \R^n$ we have the inequality
\[
e^{-L\|x - y\|} \le \frac{\dens(x)}{\dens(y)} \le e^{L\|x - y\|} \, .
\]
\end{definition}

The simplest example of a log-Lipschitz density is the exponential density with respect to the norm $\| \cdot \|$, i.e., the density function that is proportional to $e^{-\|x\|}$. This density is $1$-log-Lipschitz. The usefulness of log-Lipschitz densities for our purposes comes from the following simple but important observation.

\begin{lemma}\label{shifted_event}
Let $X$ be a random variable whose density function is $L$-log-Lipschitz. Furthermore, let $A \subseteq \R^n$ be any event and let $v \in \R^n$ be any vector. Then, we have the inequality
\[
e^{-L\|v\|}\Prob(X \in A) \le \Prob(X \in v + A) \le e^{L\|v\|}\Prob(X \in A) \, .
\]
\end{lemma}
\begin{proof}
Let $\dens$ be the density function of $X$. By change of variables $y = v +x$ we have
\[
\Prob(X \in v + A) = \int_{v + A} \dens(x) dx = \int_{A} \dens(y - v) dy \le e^{L\|v\|}\int_{A} \dens(y) dy = e^{L\|v\|}\Prob(X \in A) \, .
\]
The other inequality follows in the same way.
\end{proof}

The lemma above has the following corollary, which serves as a basis for our approach.

\begin{proposition}\label{shifted_decomposition}
Let $X$ be a random variable whose density function is $L$-log-Lipschitz. Furthermore, let $A_1, \dots, A_N \subseteq \R^n$ be disjoint events such that $\R^{n} \setminus (\cup_i A_i)$ has Lebesgue measure zero. For every $i \in [N]$ let $v_i \in \R^n$ be a vector such that $v_i + A_i \subseteq A_i$. Furthermore, let $\delta > 0$ be such that $\|v_i\| \le \delta$ for all $i \in [N]$. Then, we have
\[
\Prob\bigl(X \notin \bigcup_{i = 1}^{N} (v_i + A_i)\bigr) \le \delta L \, .
\]
\end{proposition}
\begin{proof}
By our assumptions, the events $v_i + A_i$ are disjoint. Hence, \cref{shifted_event} gives
\[
\Prob\bigl(X \in \bigcup_{i = 1}^{N} (v_i + A_i)\bigr) = \sum_{i = 1}^{N} \Prob(X \in v_i + A_i) \ge \sum_{i = 1}^{N} e^{-\delta L}\Prob(X \in A_i) = e^{-\delta L}.
\]
Therefore $\Prob\bigl(X \notin \bigcup_{i = 1}^{N} (v_i + A_i)\bigr) \le 1 - e^{-\delta L} \le \delta L$.
\end{proof}

Intuitively, \cref{shifted_decomposition} can be used as follows in the context of games: we partition the space of possible weights into a disjoint union of some sets, and we find shifts such that the games belonging to the shifted version of this partition are well-conditioned. The cones $\polyh^{\sigma,\tau}$ have exactly these properties. This is formalized in the next lemmas.

\begin{lemma}\label{slack_log_lip}
Suppose that the weights $\payoff \in \R^m$ of a mean payoff game are drawn from a distribution that has $L$-log-Lipschitz density with respect to the supremum norm $\supnorm{\cdot}$. Then, for any $\alpha > 0$ we have
\[
\Prob\bigl(\sk(\payoff) < \alpha \bigr) \le \alpha nL \, .
\]
\end{lemma}
\begin{proof}
Consider the collection of all nonempty cones $(\polyh^{\sigma,\tau})$, so that the set $\Generic$ is their union. By \cref{policy_cone} we have that $\polyh^{\sigma,\tau}$ are disjoint and that $\R^m \setminus \Generic$ has Lebesgue measure zero. Furthermore, \cref{wide_cone} shows that every nonempty cone $\polyh^{\sigma,\tau}$ contains a point $\payoff^{\sigma,\tau}$ such that $\supnorm{\payoff^{\sigma,\tau}} = 1$ and $\sk(\payoff^{\sigma,\tau}) \ge \frac{1}{n}$. Furthermore, if $\payoff \in \alpha \payoff^{\sigma,\tau} + \polyh^{\sigma,\tau}$, then $\payoff \in \polyh^{\sigma,\tau}$ (because $\polyh^{\sigma,\tau}$ is a convex cone) and $\sk(\payoff) \ge \frac{\alpha}{n}$ (by linearity of the functions $\gameval^{\sigma,\tau},\bias^{\sigma,\tau}$). Hence, by \cref{shifted_decomposition} we get
\[
\Prob\bigl(\sk(\payoff) < \frac{\alpha}{n} \bigr) \le \Prob\bigl(\payoff \notin \bigcup_{\sigma,\tau} (\alpha \payoff^{\sigma,\tau} + \polyh^{\sigma,\tau})\bigr) \le \alpha L \, .
\]
\end{proof}

To bound the entire condition number, and not just the slack, we need a tail bound on the log-Lipschitz distribution. For the sake of simplicity, we will only consider one log-Lipschitz perturbation of the weights: we assume that the every weight $\payoff_{ij}$ is drawn independently from the Laplace distribution, i.e., that  we have $\payoff_{ij} = \bar{\payoff}_{ij} + \delta_{ij}$, where the density function of $\delta_{ij}$ is given by $\frac{1}{2\rho}e^{-|x|/\rho}$. A similar bound can be obtained for other log-Lipschitz distributions. In all our estimates, $\log$ denotes the natural logarithm.

\begin{theorem}\label{condition_laplace_noise}
Fix $\bar{\payoff} \in [-1,1]^m$ and suppose that the weights of a mean payoff game are drawn randomly by perturbing each $\bar{\payoff}_{ij}$ with independent Laplace noise, i.e., that we have $\payoff_{ij} = \bar{\payoff}_{ij} + \delta_{ij}$, where the variables $\delta_{ij}$ are i.i.d. with density $\frac{1}{2\rho}e^{-|x|/\rho}$. Then for any $0 < \varepsilon < 1$ we have
\[
\Prob\bigl( \Cond(\payoff) > \frac{4mn \log \frac{2m}{\varepsilon}}{\varepsilon} + \frac{4mn}{\varepsilon \rho}) \le \varepsilon \, .
\]
\end{theorem}
\begin{proof}
Fix $\alpha > 0$ and $\beta > 2$. Using $\hilb{\payoff} \le 2\supnorm{\payoff}$ and the union bound we have
\[
\Prob(\Cond(\payoff) > \alpha) \le \Prob(\hilb{\payoff} > \beta) +  \Prob(\sk(\payoff) < \frac{\beta}{\alpha}) \le \Prob(\supnorm{\payoff} > \frac{\beta}{2}) + \Prob(\sk(\payoff) < \frac{\beta}{\alpha}) \, .
\]
Furthermore, note that the density of the vector $\payoff$ is proportional to $e^{-\|x - \bar{\payoff}\|_1 /\rho}$ and that for all $x,y \in \R^m$ we have
\[
\big | \sum_{ij} \abs{x_{ij} - \bar{\payoff}_{ij}} - \sum_{ij} \abs{y_{ij} - \bar{\payoff}_{ij}} \big | \le \sum_{ij} \big | \abs{x_{ij} - \bar{\payoff}_{ij}} -  \abs{y_{ij} - \bar{\payoff}_{ij}} \big | \le \sum_{ij} \abs{x_{ij} - y_{ij}} \le m \supnorm{x - y} \, .
\]
Therefore, the density of $\payoff$ is $\frac{m}{\rho}$-log-Lipschitz with respect to the supremum norm. Hence, by \cref{slack_log_lip} we have $\Prob\bigl(\sk(\payoff) < \frac{\beta}{\alpha} \bigr) \le \frac{\beta mn}{\alpha \rho}$. Furthermore, since we assume that $\beta > 2$, for every $(i,j) \in \edges$ we have
\[
\Prob(\abs{\payoff_{ij}} > \frac{\beta}{2}) \le \Prob(\abs{\delta_{ij}} > \frac{\beta}{2} - 1) = \frac{1}{\rho}\int_{\frac{\beta}{2} -1}^{+\infty} e^{-s/\rho} ds = e^{-(\beta-2)/(2\rho)} \, .
\]
Hence, we get the inequality
\[
\Prob(\Cond(\payoff) > \alpha) \le m e^{-(\beta-2)/(2\rho)} + \frac{\beta m n}{\alpha \rho} \, .
\]
If we now take $\alpha \coloneqq \frac{4mn \log \frac{2m}{\varepsilon}}{\varepsilon} + \frac{4mn}{\varepsilon \rho}$ and $\beta \coloneqq 2 + 2\rho \log \frac{2m}{\varepsilon}$, then we have $\alpha = \frac{2mn \beta}{\varepsilon \rho}$, so $\frac{\beta mn}{\alpha \rho} = \frac{\varepsilon}{2}$. Moreover, $m e^{-(\beta-2)/(2\rho)} = \frac{\varepsilon}{2}$, so $\Prob(\Cond(\payoff) > \alpha) \le \varepsilon$.
\end{proof}

We now show how the proofs above can be extended to analyze the condition number of games whose weights are drawn from the Gaussian distribution. This is based on the fact that the Gaussian distribution is ``approximately'' log-Lipschitz. To state this formally, we need a few preliminary results about the Gaussian distribution. We use the following notation. We denote by $\Phi \colon \R \to \R$,
\[
\Phi(x) \coloneqq \frac{1}{\sqrt{2\pi}}\int_{-\infty}^{x} e^{-t^2/2} dt
\]
the cumulative distribution of the standard univariate normal distribution, with the convention that $\Phi(-\infty) = 0, \Phi(+\infty) = 1$. Furthermore, we denote by $h \colon \R \to \R$,
\[
h(x) \coloneqq e^{x^2/2}\int_{x}^{+\infty} e^{-t^2/2} dt
\]
the \emph{Mills' ratio} of the Gaussian distribution. We note that $\Phi(-x) = 1 - \Phi(x)$, which implies that
\begin{equation}\label{eq:gauss_mill}
\Phi(x) = \frac{1}{\sqrt{2\pi}}h(-x)e^{-x^2/2} \,
\end{equation}
for all $x \in \R$. We will use the following result about Mills' ratio.

\begin{lemma}[see, e.g., \cite{Baricz:2008}]\label{mills_ratio}
The function $h(x)$ is strictly deceasing. Furthermore, we have
\[
\frac{2}{\sqrt{4 + x^2} + x} < h(x) < \frac{4}{\sqrt{8 + x^2} + 3x} < \frac{1}{x}
\]
for all $x > 0$.
\end{lemma}

We use the lemma above to prove the following technical estimates on $\Phi(x)$.

\begin{lemma}\label{inverse_cdf}
Fix $0.9 \ge \varepsilon > 0$. Then, for every $1 > x \ge \varepsilon$ we have $\Phi^{-1}(x) > -\sqrt{2\log \frac{1}{\varepsilon}}$.
\end{lemma}
\begin{proof}
Since $\Phi$ is strictly increasing, $\Phi^{-1}$ is also strictly increasing. Therefore, it is enough to show that $\varepsilon > \Phi\bigl(-\sqrt{2\log \frac{1}{\varepsilon}}\bigr)$. By \cref{eq:gauss_mill} and \cref{mills_ratio} we get
\[
\Phi\bigl(-\sqrt{2\log \frac{1}{\varepsilon}}\bigr) = \frac{1}{\sqrt{2\pi}}h(\sqrt{2\log \frac{1}{\varepsilon}})\varepsilon \le \frac{1}{2\sqrt{\pi \log\frac{1}{\varepsilon}}} \varepsilon \le \frac{1}{2\sqrt{\pi \log \frac{10}{9} }} \varepsilon < \varepsilon \, .
\]
\end{proof}

\begin{lemma}\label{taylor_cdf}
Fix $0.6 \ge \varepsilon > 0$. Then, for every $x \in \R$ and $\delta \ge 0$ such that $x - \delta \ge -2\sqrt{2 \log \frac{1}{\varepsilon}}$ we have $\Phi(x - \delta) \ge \exp(-4\delta \sqrt{\log \frac{1}{\varepsilon}}) \Phi(x)$.
\end{lemma}
\begin{proof}
Denote $g(x) = \log \Phi(x)$. Then, by \cref{eq:gauss_mill}, the derivative of $g$ is given by
\[
g'(x) = \frac{1}{\sqrt{2\pi}}\frac{e^{-x^2/2}}{\Phi(x)} = \frac{1}{h(-x)} \, .
\]
Hence, \cref{mills_ratio} shows that $g'$ is strictly decreasing and positive. Since $\Phi$ is strictly increasing, the mean value theorem gives that, for all $x \in \R$ and $\delta \ge 0$,
\[
0 \le \log \Phi(x) - \log \Phi(x - \delta) \le \delta g'(x - \delta)\, .
\]
If $x - \delta \ge -2\sqrt{2 \log \frac{1}{\varepsilon}}$, the monotonicity of $g'$ gives
\[
0 \le \log \Phi(x) - \log \Phi(x - \delta) \le \delta g'(x - \delta) \le \delta g'(-2\sqrt{2 \log \frac{1}{\varepsilon}}) = \frac{\delta}{h(2\sqrt{2 \log \frac{1}{\varepsilon}})} \, .
\]
If we further assume that $\varepsilon \le 0.6$, then $\log \frac{1}{\varepsilon} \ge \log \frac{10}{6} > \frac{1}{2}$ and \cref{mills_ratio} gives
\begin{align*}
\frac{1}{h(2\sqrt{2 \log \frac{1}{\varepsilon}})} &< \frac{1}{2}(\sqrt{4 + 8\log \frac{1}{\varepsilon}} + 2\sqrt{2 \log \frac{1}{\varepsilon}}) < \frac{1}{2}(\sqrt{16\log \frac{1}{\varepsilon}} + 2\sqrt{2 \log \frac{1}{\varepsilon}}) \\
&= (2 + \sqrt{2})\sqrt{\log \frac{1}{\varepsilon}} < 4\sqrt{\log \frac{1}{\varepsilon}} \, .
\end{align*}
Hence, $\log \Phi(x) - \log \Phi(x - \delta) < 4\delta \sqrt{\log \frac{1}{\varepsilon}}$ and $\Phi(x - \delta) \ge \exp(-4\delta \sqrt{\log \frac{1}{\varepsilon}})\Phi(x)$.
\end{proof}

Given the estimates above, we can now state the analogue of \cref{shifted_event} for the Gaussian distribution. We denote by
\[
\gauss_n(A) = \frac{1}{\sqrt{(2\pi)^n}}\int_{A} e^{-\|x\|^2/2} dx
\]
the standard Gaussian measure on $\R^n$. We have the following inequality about the Gaussian measure of a shifted set, see \cite{LiKuelbs:1998,Bobkov:1999,BartheCorderoFradelizi:2001} for more information.

\begin{theorem}[\cite{LiKuelbs:1998}]\label{gaussian_shift}
Suppose that $A \subseteq \R^n$ is any Borel set and that $v \in \R^n$ is any vector. Let $x \coloneqq \gauss_n(A)$. Then, we have the inequality
\[
\Phi\bigl( \Phi^{-1}(x) - \|v\|_2\bigr) \le \gauss_n(v + A) \le \Phi\bigl( \Phi^{-1}(x) + \|v\|_2\bigr) \, .
\]
\end{theorem}

Combining this inequality with the estimates above, we get the following corollary, which is an analogue of \cref{shifted_event}.

\begin{corollary}\label{relative_gaussian_shift}
Fix a Borel set $A \subseteq \R^n$ and let $0.6 \ge \varepsilon > 0$ be such that $\gauss_n(A) \ge \varepsilon$. Then, for any vector $v \in \R^n$ that satisfies $\|v\|_2 \le \sqrt{2\log\frac{1}{\varepsilon}}$ we have $\gauss_n(v + A) \ge \exp\bigl(-4\|v\|_2 \sqrt{\log \frac{1}{\varepsilon}}\bigr) \gauss_n(A)$.
\end{corollary}
\begin{proof}
Let $x \coloneqq \gauss_n(A)$. By \cref{inverse_cdf} we have $\Phi^{-1}(x) > -\sqrt{2 \log \frac{1}{ \varepsilon}}$, so $\Phi^{-1}(x) - \|v\|_2 \ge -2\sqrt{2 \log \frac{1}{ \varepsilon}}$. Hence, \cref{taylor_cdf} gives $\Phi(\Phi^{-1}(x) - \|v\|_2) \ge \exp(-4\|v\|_2 \sqrt{\log \frac{1}{\varepsilon}}) \gauss_n(A)$. Combining this with \cref{gaussian_shift} we get $\gauss_n(v + A) \ge \Phi(\Phi^{-1}(x) - \|v\|_2) \ge \exp(-4\|v\|_2 \sqrt{\log \frac{1}{\varepsilon}}) \gauss_n(A)$.
\end{proof}

The corollary above can be used to show an analogue of \cref{shifted_decomposition} for the Gaussian distribution.

\begin{proposition}\label{shifted_decomposition_gaussian}
Suppose that $X$ is a Gaussian random vector with distribution $\mathcal{N}(\theta, \rho^2 I)$ for some $\theta \in \R^n$ and $\rho > 0$. Let $A_1, \dots, A_N \subseteq \R^n$ for $N \ge 1$ be disjoint events such that $\R^{n} \setminus (\cup_i A_i)$ has Lebesgue measure zero. For every $i \in [N]$ let $v_i \in \R^n$ be a vector such that $v_i + A_i \subseteq A_i$. Furthermore, let $\delta > 0$ be such that $\|v_i\|_2 \le \delta\le \frac{\rho}{4}$ for all $i \in [N]$. Then, we have
\[
\Prob\bigl(X \notin \bigcup_{i = 1}^{N} (v_i + A_i)\bigr) < 5 \frac{\delta}{\rho}\sqrt{\log \frac{N \rho}{\delta}} \, .
\]
\end{proposition}
\begin{proof}
To start, suppose that $\theta = 0$ and $\rho = 1$. Let $\varepsilon \coloneqq \frac{\delta}{N}$. Since $\delta \le \frac{1}{4}$ and $N \ge 1$, we have $\varepsilon \le \frac{1}{4}$ and $\sqrt{\log \frac{1}{\varepsilon}} > 1$. By permuting indices, we can suppose there exists $N' \le N$ such that $\Prob(X \in A_i) \ge \varepsilon$ for all $i \le N'$ and $\Prob(X \in A_i) < \varepsilon$ for all $i > N'$. Then, by \cref{relative_gaussian_shift} we get
\begin{align*}
\Prob \bigl(X \in \bigcup_{i = 1}^{N} (v_i + A_i)\bigr) &\ge \sum_{i = 1}^{N'}\Prob(X \in v_i + A_i) \ge \exp(-4 \delta \sqrt{ \log \frac{1}{\varepsilon}})\sum_{i = 1}^{N'}\Prob(X \in A_i)  \\
&= \exp(-4 \delta \sqrt{ \log \frac{1}{\varepsilon}})\bigl( 1 -  \sum_{i = N'+1}^{N}\Prob(X \in A_i) \bigr) \ge \exp(-4 \delta \sqrt{ \log \frac{1}{\varepsilon}})(1- N\varepsilon) \\
&\ge (1 - 4 \delta \sqrt{\log \frac{1}{\varepsilon}})(1 - N \varepsilon) > 1 - N\varepsilon - 4 \delta \sqrt{\log \frac{1}{\varepsilon}}  = 1- \delta - 4 \delta \sqrt{\log \frac{1}{\varepsilon}} \\
&> 1 - 5 \delta \sqrt{\log \frac{1}{\varepsilon}} \, .
\end{align*}
Hence, $\Prob\bigl(X \notin \bigcup_{i = 1}^{N} (v_i + A_i)\bigr) < 5 \delta \sqrt{\log \frac{1}{\varepsilon}} = 5 \delta \sqrt{\log \frac{N}{\delta}}$.

For arbitrary $\theta$ and $\rho$ we consider the change of variables $Y = \frac{1}{\rho}(X - \theta)$. Then, $Y$ has the distribution $\mathcal{N}(0,I)$. Moreover, if we denote $A'_i = \frac{1}{\rho}(-\theta + A_i)$ and $v_i' = \frac{1}{\rho} v_i$ for all $i \in [N]$, then $\Prob(X \in v_i + A_i) = \Prob(Y \in v'_i + A'_i)$ for all $i$. Furthermore, the events $A'_i$ are disjoint and we have $v_i' + A_i' \subseteq A_i'$. Hence, if we set $\delta' = \frac{\delta}{\rho}$, then $\Prob\bigl(X \notin \bigcup_{i = 1}^{N} (v_i + A_i)\bigr) = \Prob\bigl(Y \notin \bigcup_{i = 1}^{N} (v'_i + A'_i)\bigr) < 5 \delta' \sqrt{\log \frac{N}{\delta'}}$.
\end{proof}

In order to apply this proposition to mean payoff games, we suppose that each weight $\payoff_{ij}$ of the game is drawn independently from the distribution $\mathcal{N}(\bar{\payoff}_{ij}, \rho^2)$, where $\bar{\payoff}_{ij} \in \interval{-1}{1}$ for all $(i,j)$. Then, we have the following estimate.

\begin{theorem}\label{condition_gaussian_noise}
Fix $\bar{\payoff} \in [-1,1]^m$ and suppose that the weights of a mean payoff game are drawn randomly by perturbing each $\bar{\payoff}_{ij}$ with independent Gaussian noise, i.e., that we have $\payoff \sim \mathcal{N}(\bar{\payoff}, \rho^2I)$ for some $\rho > 0$. Then for any $0 < \varepsilon < 0.3$ we have
\[
\Prob\Bigl( \Cond(\payoff) > 40n^{3/2}\sqrt{m \log n} \cdot \frac{\log\frac{1}{\varepsilon}}{\varepsilon} \cdot \bigl( 3\sqrt{\log \frac{2m}{\varepsilon}} +  \frac{2}{\rho} \bigr) \Bigr) \le \varepsilon \, .
\]
\end{theorem}
\begin{proof}
We adapt the proof of \cref{condition_laplace_noise}. Fix $\alpha \coloneqq 40n^{3/2}\sqrt{m \log n} \cdot \frac{\log\frac{1}{\varepsilon}}{\varepsilon} \cdot \bigl( 3\sqrt{\log \frac{2m}{\varepsilon}} +  \frac{2}{\rho} \bigr)$, $\beta \coloneqq 2 + 3\rho \sqrt{\log \frac{2m}{\varepsilon}}$, $\kappa \coloneqq n\sqrt{m} \frac{\beta}{\alpha \rho}$, and $\omega \coloneqq \frac{1}{40\sqrt{n \log n}}$. Under this notation, we have $\kappa = \frac{\omega \varepsilon}{\log \frac{1}{\varepsilon}}$. Moreover, the assumption $\varepsilon < 0.3$ implies that $\log \frac{1}{\varepsilon} > 1$.
As in the proof of \cref{condition_laplace_noise}, we have
\[
\Prob(\Cond(\payoff) > \alpha) \le \Prob(\supnorm{\payoff} > \frac{\beta}{2}) + \Prob(\sk(\payoff) < \frac{\beta}{\alpha}) \, .
\]
To bound the second term, we proceed as in the proof of \cref{slack_log_lip} and consider the collection of all nonempty cones $(\polyh^{\sigma,\tau})$ with points $\payoff^{\sigma,\tau} \in \polyh^{\sigma,\tau}$ such that $\supnorm{\payoff^{\sigma,\tau}} = 1$ and $\sk(\payoff^{\sigma,\tau}) \ge \frac{1}{n}$. As in \cref{slack_log_lip} we have
\[
\Prob\bigl(\sk(\payoff) < \frac{\beta}{\alpha} \bigr) \le \Prob\bigl(\payoff \notin \bigcup_{\sigma,\tau} (\frac{n\beta}{\alpha} \payoff^{\sigma,\tau} + \polyh^{\sigma,\tau})\bigr) \, .
\]
Since $\supnorm{\payoff^{\sigma,\tau}} = 1$ we get $\|\payoff^{\sigma,\tau}\|_2 \le \sqrt{m}$ and therefore $\| \frac{n\beta}{\alpha} \payoff^{\sigma,\tau}\|_2 \le n\sqrt{m}\frac{\beta}{\alpha} = \rho \kappa = \rho \frac{\omega \varepsilon}{\log\frac{1}{\varepsilon}} < \frac{\rho}{40}$. Moreover, note that the collection $(\polyh^{\sigma,\tau})$ consists of at most $n^n$ cones (for each vertex of the graph we choose a vertex to which the policies $(\sigma,\tau)$ lead). Hence, \cref{shifted_decomposition_gaussian} gives
\[
\Prob\bigl(\sk(\payoff) < \frac{\beta}{\alpha} \bigr) \le 5\kappa \sqrt{\log \frac{n^n}{\kappa}} = 5 \kappa \sqrt{n\log n + \log \frac{1}{\kappa}} \le 5 \kappa \sqrt{n \log n}  + 5 \kappa \sqrt{\log \frac{1}{\kappa}} \, .
\]
The first term of the right-had side can be bounded by $5 \kappa \sqrt{n \log n} = \frac{\kappa}{8\omega} = \frac{\varepsilon}{8 \log \frac{1}{\varepsilon}} \le \frac{\varepsilon}{8}$. The second term can be bounded by
\begin{align*}
5 \kappa \sqrt{\log \frac{1}{\kappa}} = 5\frac{\omega\varepsilon}{\log \frac{1}{\varepsilon}}\sqrt{\log \frac{\log \frac{1}{\varepsilon}}{\omega \varepsilon}} \le 5\frac{\omega \varepsilon}{\log \frac{1}{\varepsilon}} \sqrt{\log \frac{1}{\varepsilon}} + 5\frac{\omega \varepsilon}{\log \frac{1}{\varepsilon}}\sqrt{\log \frac{\log \frac{1}{\varepsilon}}{\omega}} \, .
\end{align*}
Since $\omega < \frac{1}{40}$, we get $5\frac{\omega \varepsilon}{\sqrt{\log \frac{1}{\varepsilon}}} \le \frac{\varepsilon}{8}$. Moreover, the function $g(t) \coloneqq \frac{\sqrt{\log t}}{t}$ satisfies $g(t) \le \frac{1}{20}$ for all $t \ge 40$. By taking $t \coloneqq \frac{\log \frac{1}{\varepsilon}}{\omega}$ we get $5\frac{\omega \varepsilon}{\log \frac{1}{\varepsilon}}\sqrt{\log \frac{\log \frac{1}{\varepsilon}}{\omega}} \le \frac{\varepsilon}{4}$. Hence, we have $\Prob\bigl(\sk(\payoff) < \frac{\beta}{\alpha} \bigr) \le \frac{\varepsilon}{2}$.

To bound $\Prob(\supnorm{\payoff} > \frac{\beta}{2})$, note that $\beta > 2 + 2\rho$. Hence, by \cref{eq:gauss_mill} and \cref{mills_ratio}, for every $(i,j) \in \edges$ we have
\begin{align*}
\Prob(\abs{\payoff_{ij}} > \frac{\beta}{2}) &\le \Prob(\abs{\payoff_{ij} - \bar{\payoff}_{ij}} > \frac{\beta}{2} - 1) = \Prob(\abs{\frac{\payoff_{ij} - \bar{\payoff}_{ij}}{\rho}} > \frac{\beta -2}{2\rho}) \\
&= 2\Phi\bigl(\frac{2-\beta}{2\rho}\bigr) = \sqrt{\frac{2}{\pi}}h\bigl(\frac{\beta-2}{2\rho}\bigr)\exp \bigl(-\frac{(\beta-2)^2}{8\rho^2} \bigr) \\
&< \sqrt{\frac{2}{\pi}} \frac{2\rho}{\beta - 2} \exp(- \frac{(\beta-2)^2}{8\rho^2}) < \exp(- \frac{(\beta-2)^2}{8\rho^2}) \, .
\end{align*}
Therefore, by the definition of $\beta$ we have
\begin{align*}
\Prob(\supnorm{\payoff} > \frac{\beta}{2}) \le m\exp(- \frac{(\beta-2)^2}{8\rho^2}) = m\exp\bigl(- \frac{9}{8} \log\frac{2m}{\varepsilon}\bigr) < \frac{\varepsilon}{2} \, .
\end{align*}
By combining the estimates, we get $\Prob(\Cond(\payoff) > \alpha) \le \varepsilon$, as claimed.
\end{proof}

To finish, we combine the estimates of the condition number with the complexity estimates from \Cref{sec:complexity_condition} to get our main theorem.

\begin{theorem}\label{th:main_thm}
Fix $\bar{\payoff} \in [-1,1]^m$ and suppose that the weights of a mean payoff game are drawn randomly by perturbing each $\bar{\payoff}_{ij}$ with independent Gaussian noise, i.e., that we have $\payoff \sim \mathcal{N}(\bar{\payoff}, \rho^2I)$ for some $0< \rho < 0.3$. Then, for any $0 < \varepsilon < 0.3$, with probability at least $1 - \varepsilon$ we have:
\begin{enumerate}
\item The \KM iteration finds a pair of optimal policies in $O(n^7 m \log n \log m \frac{(\log \frac{1}{\varepsilon})^3}{\varepsilon^2 \rho^2})$ iterations.
\item The Halpern iteration finds a pair of optimal policies in $O(n^6 m \log n \log m \frac{(\log \frac{1}{\varepsilon})^3}{\varepsilon^2 \rho^2})$ iterations.
\item The combined Zwick--Paterson/Halpern iteration (algorithm $\ZPHVI$) finds a pair of optimal policies in $O(n^{7/2} \sqrt{m \log n \log m} \frac{(\log \frac{1}{\varepsilon})^{3/2}}{\varepsilon \rho})$ iterations of the main loop.
\item The algorithm $\IncDiscPI$ equipped with Howard's policy iteration rule finds a pair of optimal policies by performing $O(n^7 m^3 (\log n \log m)^4 \frac{(\log \frac{1}{\varepsilon})^6 (\log \frac{1}{\rho})^3}{\varepsilon^2 \rho^2})$ switches for both players.
\item The algorithm $\IncDiscPI$ equipped with Dantzig's policy iteration rule finds a pair of optimal policies by performing $O(n^{13/2} m^{7/2} (\log n)^{9/2} (\log m)^{5/2} \frac{(\log \frac{1}{\varepsilon})^{7/2}(\log \frac{1}{\rho})^2}{\varepsilon \rho})$ switches for both players.
\end{enumerate}
The last two claims also hold for the discounted games and the algorithm $\IncDiscPItwo$.
\end{theorem}
\begin{proof}
The proof follows by combining the estimate for the condition number \cref{condition_gaussian_noise} with the complexity estimates given in \cref{sec:complexity_condition} (\cref{overall_complexity_KM}, \cref{overall_complexity_Halpern}, \cref{overall_complexity_ZPH}, and \cref{increasing_discount_condition}).
\end{proof}

\section{Conclusion and future work}

\newcounter{nparcounter}
\setcounter{nparcounter}{0}

\NewDocumentCommand{\npar}{o}{%
  \IfValueTF{#1}{%
    \medskip\noindent\refstepcounter{nparcounter}\textbf{\thenparcounter}\quad \textit{#1}.
  }{%
    \medskip\noindent\refstepcounter{nparcounter}\textbf{\thenparcounter}.
  }%
}

We gave a number of algorithms that have smoothed polynomial complexity for deterministic discounted and mean payoff games played on arbitrary graphs. This answers some of the questions that were left as open in \cite{LoffSkomra:2024}, but not all of them. Furthermore, our analysis of the value iteration algorithms leads to a few more open questions, which we discuss below.

\npar As in \cite{LoffSkomra:2024}, our policy-iteration rule is not one of the standard rules (Howard, lexicographic, RandomFacet, \emph{etc}). Therefore, we can raise the same questions about policy iteration: do these standard rules also have polynomial smoothed complexity on deterministic two-player games? How about other ``combinatorial'' algorithms? For value iteration, one can ask if the different versions of value iteration discussed in \cite{cadilhac_2025_fast} have smoothed polynomial complexity.

\npar Another natural question is to ask if we can extend our results to \emph{stochastic} two-player games. For policy iteration rules, a counter-example of Christ and Yannakakis \cite{christ2023smoothed} shows that the Howard all-switches rule does not have polynomial smoothed complexity, even on stochastic MDPs. Therefore, one needs to study different rules in order to prove a smoothed polynomial complexity result. For value iteration methods, we note that our analysis of the correctness of \KM and Halpern iteration uses the fact that the underlying games are deterministic. It is not difficult to show that these methods will eventually produce optimal policies for \emph{constant-value} stochastic games, but it is not immediately clear if they produce optimal policies for arbitrary stochastic games. The analysis from \cite{LeeRyu:2025} shows that this is true in the one-player setting, and one would need to extend their analysis to the two-player case to even show that these iterative methods output optimal policies. Afterwards, one can study if these iterative methods have smoothed polynomial complexity in the stochastic setting, or if there is a counter-example showing that this is not the case.

\npar Another question concerns the convergence rate of the \KM iteration for deterministic games. In our work, we have used general results showing that the error bound behaves like $O(1/\sqrt{k})$, but we do not know if this rate is optimal in the setting of deterministic games. In this context, we note that a lower bound $\Omega(1/k)$ was shown in \cite{LeeRyu:2025} for more general iterative methods in the setting of MDPs.

\npar The next question concerns the dependence of our smoothed polynomial algorithms on the amount of noise $\rho$ in the Gaussian perturbation $ \mathcal{N}(\bar{\payoff}, \rho^2I)$. As discussed in \cref{th:main_thm}, the best dependence on $\rho$ is achieved by the combined Zwick--Paterson/Halpern iteration and is linear in $\frac{1}{\rho}$. In the smoothed analysis of the simplex method, an optimal sub-linear dependence on $\frac{1}{\rho}$ has recently been established \cite{BachHuiberts:2025}. One can ask if such sub-linear dependence can also be achieved in the setting of deterministic games.

\npar As in \cite{LoffSkomra:2024}, one can ask if other problems in $\mathsf{UEOPL}$, such as P-Matrix Linear Complementarity Problem are amenable to smoothed analysis.

\end{document}